\documentclass[12pt]{article}

\usepackage[margin=1in]{geometry}  
\usepackage{epsfig}
\usepackage{graphics,graphicx,color}
\usepackage{amssymb,amsfonts,amsthm,amsmath}
\usepackage{bbm}
\usepackage{multirow, booktabs}
\usepackage{mathrsfs}
\usepackage{natbib}
\usepackage{slashbox}
\usepackage{longtable,booktabs}
\usepackage{arydshln}
\usepackage{paralist}
\usepackage{times}
\usepackage{fancyhdr}
\usepackage[ruled,vlined]{algorithm2e}
\usepackage{titling}
\usepackage{setspace}
\usepackage{booktabs,makecell,threeparttable}
\usepackage{pdflscape} 
\usepackage{rotating}

\newtheorem{lemma}{Lemma}
\newtheorem{definition}{Definition}
\newtheorem{proposition}{Proposition}

\newtheorem{theorem}{Theorem}

\newtheoremstyle{nonitalic} 
  {\topsep}                 
  {\topsep}                 
  {\normalfont}             
  {}                        
  {\bfseries}               
  {.}                       
  {.5em}                    
  {}                        

\theoremstyle{nonitalic}
\newtheorem{example}{Example}

\newcommand{\tr}{\mbox{tr}}

\begin{document}

\title{TCARD: Nearly Balanced Two-Level Designs with Treatment Cardinality Constraints with an Application to LLM Prompt Engineering}
\author{Kexin Xie$^{a}$, Ryan Lekivetz$^{b}$, and Xinwei Deng$^{a*}$}
\date{%
    $^a$Department of Statistics, Virginia Tech, USA\\%
    $^b$SAS Institute, USA\\[2ex]%
}
\maketitle
\def\thefootnote{*}\footnotetext{Address for correspondence: 
Xinwei Deng, Professor of Statistics, Co-Director of Statistics and Artificial Intelligence Laboratory at Virginia Tech (Email: xdeng@vt.edu).}

\vspace{-12pt}
\begin{abstract}
Modern experimental designs often face the so-called treatment cardinality constraint, which is the constraint on the number of included factors in each treatment. Experiments with such constraints are commonly encountered in engineering simulation, AI system tuning, and large-scale system verification. This calls for the development of adequate designs to enable statistical efficiency for modeling and analysis within feasible constraints. In this work, we study two-level designs under this $k$-treatment cardinality constraint (TCARD), where the design matrix $\mathbf{X} \in \{0,1\}^{n \times p}$ has constant row sums equal to $k$. Although TCARDs are closely related to balanced incomplete block designs (BIBDs), exact BIBD structure is unavailable for many practical $(n,p,k)$ combinations. This leads to the notion of nearly balanced TCARDs, which we prove minimize the first two components of the generalized word-length pattern. We also show that good projection behavior in this setting is governed by two count-based regularities: balanced factor replications and uniform pairwise concurrences. Motivated by this characterization, we then propose the Balanced Concurrence Deviation ($\Phi_{\mathrm{BCD}}$), a model-free objective that jointly penalizes replication imbalance and concurrence dispersion. We further show that this criterion is closely connected to classical optimality principles, including $(M,S)$-optimality, centered $\mathrm{UE}(s^2)$ criterion, and Bayesian $D$-optimality. To construct designs minimizing $\Phi_{\mathrm{BCD}}$, we develop a coordinate-exchange (CE) algorithm with efficient incremental updates, together with a simulation-based procedure for calibrating the criterion weights to the intended downstream task. Numerical experiments confirm that the proposed method compares favorably with existing alternatives across a range of problem sizes and constraint strengths.
\end{abstract}
\noindent\textbf{Keywords:} Experimental Design;  Treatment Cardinality Constraint; Row constraint Space Balanced Incomplete Block Design; Coordinate-Exchange Algorithm; Bayesian $D$-Optimality.





\section{Introduction}\label{sec:intro}

Many statistical problems involve structural constraints that limit how experimental factors may be simultaneously activated. In one common setting, each experimental run is restricted to activating only a fixed number of factors, so that only a subset may depart from baseline levels at any one time. We refer to this restriction as the $k$-treatment cardinality constraint (TCARD), under which each of $n$ experimental runs activates exactly $k$ of the $p$ available factors. This constraint
leads to binary design matrices whose rows have fixed cardinality, and excludes the full factorial structure typically assumed in classical design theory.

Such cardinality restrictions arise across a striking range of applications. In sensor scheduling, energy and bandwidth considerations limit how many sensors can be powered simultaneously \citep{mo2011sensor, wang2009optimization}. In feature-flag experiments on digital platforms, only a small number of candidate features can be enabled in each user session to control risk and interaction complexity \citep{ternava2022interaction}. In high-throughput screening for drug
discovery, physical and biochemical constraints---well volume, minimum detectable concentration, solution ionic strength---restrict how many compounds can be pooled in a single experimental unit \citep{kainkaryam2009pooling, smucker2025large}. Similar limitations appear in mixture and subsystem screening in engineering, where safety, cost, or throughput considerations bound the number of active components. A newer and rapidly growing class of TCARD problems has emerged at the
interface of statistics and artificial intelligence, in the design of prompts for large language models (LLMs). A modern prompt is assembled from reusable components such as reasoning triggers, worked examples, and self-verification instructions, each of which can measurably shift the
model's accuracy on the downstream task. Yet practitioners do not concatenate every component at their disposal: transformer attention is finite, long prompts suffer measurable reasoning degradation well before any technical context limit \citep{liu2024lost, levy2024same}, and even semantically irrelevant additions reliably reduce accuracy on arithmetic reasoning benchmarks \citep{shi2023large}. Identifying which components carry signal, under a small per-prompt budget and a limited experimental budget, is a screening problem whose operational structure is exactly that of TCARD design. We return to this scenario as an empirical anchor throughout the paper.

Formally, let $\mathbf X = (x_{ij}) \in \{0,1\}^{n \times p}$ denote a design matrix in which entry $x_{ij} = 1$ indicates that factor $j$ is activated in run $i$. The $k$-treatment cardinality constraint requires that each row of $\mathbf X$ sums to exactly $k$:
\begin{equation}\label{eq:tcard}
  \sum_{j=1}^{p} \mathbf{1}\{x_{ij} = 1\} = k, \quad i = 1, \ldots, n,
\end{equation}
so that every experimental run activates the same number of factors. We denote the class of all such two-level designs by $\mathcal{D}(n, 2^p, k)$.

The need for principled experimental design under such TCARD constraints has been reinforced by several recent developments. In the screening literature, the principle of effect sparsity, namely that only a small number of factors drive the response, has long guided the construction and analysis of fractional factorial and Plackett--Burman designs \citep{box1986analysis, box1961multi, plackett1946design}. More recently, definitive screening designs \citep{jones2011class} and supersaturated designs \citep{booth1962some, wu1993construction, jones2014optimal} have extended the reach of screening to settings with many more factors than runs. In particular, \citet{smucker2025large} introduced row-constrained supersaturated designs for high-throughput screening, where the number of active factors per row is physically bounded. Their work demonstrates both the practical prevalence of per-run activation limits and the statistical consequences of ignoring them. However, their construction targets the supersaturated regime ($p > n$) and the $\mathrm{UE}(s^2)$ criterion, leaving open the question of optimal design under cardinality constraints when $n \geq p$ and broader projection-based criteria are of interest. 


From a statistical perspective, the TCARD constraint has important consequences for information structure and inference. Fixed cardinality restrictions induce intrinsic dependencies among the columns of the design matrix, leading to singular or nearly singular information matrices after intercept removal. As a result, many classical notions of balance, orthogonality, and optimality must be reconsidered in this setting. In particular, the quality of projected subdesigns becomes central, since statistical analyses often rely on reduced models obtained through factor screening or model simplification. These considerations motivate a careful examination of how replication patterns and pairwise co-activation structure govern the distribution of information under cardinality-constrained experimental designs.

Several fundamental challenges arise when one attempts to construct designs with good projection properties under the TCARD constraint. Algebraically, the $k$-slice of the hypercube is an irregular subset of $\{0,1\}^{n \times p}$, so strength-$t$ orthogonal arrays are typically infeasible for $t \geq 2$, and classical full or fractional factorial constructions do not apply directly. Balanced incomplete block designs (BIBDs) provide a natural reference by equalizing factor replications and pairwise co-occurrences when they exist. However, BIBD existence is sporadic in $(p, k, n)$ and, even when available, governs only the binary support structure rather than broader projection behavior. On the algorithmic side, existing design criteria were not built with these quantities in mind. Distance-based heuristics such as the maximin criterion promote inter-run separation but do not directly control column replication or pairwise concurrence. Information-based $A$-, $D$-, and $E$-criteria, together with the eigenvalue-spread surrogates derived from them, are typically formulated for unconstrained or orthogonal-array-feasible spaces and,  under the row-sum restriction~\eqref{eq:tcard}, no longer correspond to the screen-and-refit inference they are meant to support. More broadly, no existing criterion is built around the two quantities identified above---replication balance and pairwise concurrence balance---or offers a constructive route to designs that control both simultaneously under a fixed cardinality budget. This is the gap the present paper addresses.

This observation motivates a single weighted criterion, the Balanced Concurrence Deviation ($\Phi_{\mathrm{BCD}}$), whose relative emphasis is controlled by a tuning parameter $w_1/w_2$. The question of how to set such weights is not new. Weighted and compound design criteria have a long history in optimal design, combining multiple inferential goals into a single weighted objective \citep{nathanson1985multiple, cook1994equivalence, gilmour2012optimum, mcgree2008compound, morgan2017optimal, lu2011optimization}. 
However, despite this extensive literature on compound and multiobjective design criteria, there is little principled guidance on how the weights themselves should be chosen. In practice, they are often selected through visual inspection of trade-off curves \citep{nathanson1985multiple}, ad hoc rules of thumb \citep{gilmour2012optimum}, or subjective specification \citep{morgan2017optimal}. More importantly, existing approaches do not link weight selection to the performance of the downstream analysis that the experiment is intended to support. The present work addresses this gap by treating the weight in the proposed criterion not as a fixed universal constant but as a task-dependent calibration parameter, tuned offline via simulation against downstream inferential performance before any responses are collected. 

In this work, we develop the theory and algorithms for optimal two-level TCARD designs, organized around a counts-only perspective on projection quality. Our contributions are threefold: (1) First, we establish the connection between TCARD designs and BIBDs and formalize the notion of \emph{nearly balanced} TCARDs, i.e., designs in which factor replications and pairwise concurrences are as uniform as arithmetic permits, extending the classical nearly balanced framework of \citet{cheng1981nearly} to the TCARD setting. (2) Second, we introduce the Balanced Concurrence Deviation criterion $\Phi_{\mathrm{BCD}}$ (defined precisely in~\eqref{eq:Phi_BCD}) that penalizes replication imbalance and concurrence dispersion. We show that this model-free objective is sufficient for $(M, S)$-optimality under the main-effects model. We further establish statistical optimality connections between $\Phi_{\mathrm{BCD}}$ and the centered $\mathrm{UE}(s^2)$ criterion of \citet{jones2014optimal} and Bayesian $D$-optimality. (3) Finally, we develop a coordinate-exchange algorithm that preserves the fixed-row-sum constraint via within-row swaps, admits $O(k)$ incremental updates per candidate edit. Moreover, we propose a simulation-based tuning procedure for the criterion weights that aligns the design with the intended downstream analysis task before any data are collected.  
Numerical experiments across a range of problem sizes and constraint strengths confirm that the proposed criterion and algorithm compare favorably with existing alternatives in both design diagnostics and downstream variable-selection performance. A case study on prompt-component screening for a large language model evaluated on the GSM8K mathematical-reasoning benchmark \citep{cobbe2021training} using the open-weight Llama~3.1~8B model \citep{dubey2024llama} also confirms these findings in a real-world setting.


The remainder of this article is organized as follows. Section~\ref{sec:algebraic} develops the algebraic foundations: Section~\ref{subsec:bibd} presents the relationship between TCARD designs and BIBDs, and Section~\ref{subsec:nearly_balanced_design} formalizes the notion of nearly balanced TCARDs and establishes existence results along with necessary conditions. Section~\ref{sec:opt_criterion} proposes a new criterion and develops the spectral and information-theoretic links connecting the counts-only criterion to $(M, S)$-optimality, centered $\mathrm{UE}(s^2)$, and Bayesian $D$-optimality. Section~\ref{sec:algorithm} introduces the algorithmic construction framework: Section~\ref{subsec:CE} describes the coordinate-exchange algorithm, Section~\ref{subsec:tuning} presents the simulation-based weight-tuning procedure. Section~\ref{sec:simulation} reports the simulation study and empirical investigation. Section~\ref{sec:case_study} presents the case study on prompt-component engineering of LLM. Proofs of the main theoretical results are collected in the Appendix.

\section{Algebraic Structure of TCARD Designs}\label{sec:algebraic}

This section investigates the algebraic structure, which will underpin our criterion in Section~\ref{sec:opt_criterion}. We study the two-level $k$-treatment cardinality constrained design $\mathcal{D}(n,2^p,k)$ defined in~\eqref{eq:tcard}. Our goals are twofold: first, to identify the TCARD that is simultaneously balanced in replications and pairwise concurrences. And second, to formalize how far a given TCARD can deviate from this fully balanced structure while retaining the properties we care about. In Section~\ref{subsec:bibd} we establish an exact algebraic correspondence between TCARDs and balanced incomplete block designs (BIBDs), showing that a TCARD whose supports form a BIBD is simultaneously balanced in replications and pairwise concurrences. Because exact BIBDs do not exist for all triples $(p,k,n)$, in Section~\ref{subsec:nearly_balanced_design} we introduce the notion of a \emph{nearly balanced TCARD}, which retains those two balances to the extent that arithmetic permits, and establish existence conditions together with an optimality result showing that nearly balanced TCARDs simultaneously minimize the two lowest-order components of the generalized word-length pattern.

Throughout the paper, projection behavior of a TCARD will be measured through two families of count summaries. For a design $\mathbf{X}\in\{0,1\}^{n\times p}$ satisfying the cardinality constraint~\eqref{eq:tcard}, the \emph{replication} of factor $j$ and the \emph{pairwise concurrence} of factors $j,\ell$
are
\begin{equation}
r_j \;=\; \sum_{i=1}^n x_{ij},
\qquad
\lambda_{j\ell} \;=\; \sum_{i=1}^n x_{ij}x_{i\ell}
\qquad (j\ne\ell),
\label{eq:counts}
\end{equation}
i.e., the number of runs in which factor $j$ is active, and the
number of runs in which factors $j$ and $\ell$ are simultaneously
active. We denote $\bar r$ and $\bar\lambda$ for the average
replication and the average pairwise concurrence,
\begin{equation}
\bar r \;:=\; \frac{1}{p}\sum_{j=1}^p r_j,
\qquad
\bar\lambda \;:=\; \frac{1}{\binom{p}{2}}\sum_{j<\ell}\lambda_{j\ell}.
\label{eq:targets-def}
\end{equation}
Under the cardinality constraint these averages are determined by
$(n,p,k)$ alone and take the values
\begin{equation}
\bar r \;=\; \frac{nk}{p},
\qquad
\bar\lambda \;=\; \frac{nk(k-1)}{p(p-1)},
\label{eq:targets-value}
\end{equation}
independently of the specific design $\mathbf{X}$, so $\bar r$ and $\bar\lambda$ serve as fixed arithmetic targets against which individual counts can be compared.

\begin{definition}[Replication and concurrence balance]
\label{def:balance}
A TCARD $\mathbf{X}\in\mathcal{D}(n,2^p,k)$ is
\emph{balanced in replications} if $r_j = \bar r$ for all $j=1,\dots,p$, and \emph{balanced in pairwise concurrences} if $\lambda_{j\ell} = \bar\lambda$ for all $1\le j<\ell\le p$. A TCARD that is balanced in both senses is called
\emph{fully balanced}.
\end{definition}

Under the $k$-cardinality constraint, the admissible set is an irregular slice of $\{0,1\}^{n\times p}$, so neither the full nor fractional factorial families are available. When $n=\tilde\lambda\binom{p}{k}$ is a multiple of the number of $k$-subsets, the full $k$-combination design $\mathbf{X}_{\mathrm{full}}$ (obtained by taking all $\binom{p}{k}$ supports, each replicated $\tilde\lambda$ times) serves as a natural benchmark: it is fully balanced according to Definition~\ref{def:balance} and equalizes the positive eigenvalues of the main-effects information matrix. For $n\ne\tilde\lambda\binom{p}{k}$, the question becomes how to select $n$ supports that are as close to fully balanced as arithmetic permits while maintaining good projection properties. To measure projection quality we adopt the $J$-characteristics and the generalized word-length pattern (GWLP) of \cite{deng1999minimum}.
Encoding each entry of $\mathbf{X}$ in the $\{\pm 1\}$ convention via $x_{ij}\mapsto 2x_{ij}-1$, so that baseline maps to $-1$ and active to $+1$, the $J$-characteristic of a subset $u\subset\{1,\dots,p\}$ of cardinality $|u|=j$ is $J_u(\mathbf{X}) \;=\; \sum_{i=1}^n \prod_{\ell\in u}(2x_{i\ell}-1)$, and the GWLP component of order $j$ is
\begin{equation}
B_j(\mathbf{X}) \;=\; \frac{1}{n^2}\sum_{|u|=j}J_u(\mathbf{X})^2,
\qquad j=1,\dots,p.
\label{eq:GWLP}
\end{equation}
For a TCARD $\mathbf{X}\in\mathcal{D}(n,2^p,k)$, the $J$-characteristics of order one and two are determined by the replication and concurrence counts via $J_{\{j\}}(\mathbf{X}) = 2r_j - n$ and $J_{\{j,\ell\}}(\mathbf{X}) = 4\lambda_{j\ell}- 2(r_j+r_\ell) + n$ $(j\ne\ell)$. Hence
\begin{equation}\label{eq:B1B2_counts}
B_1(\mathbf{X}) = \frac{1}{n^2}\sum_{j=1}^p(2r_j-n)^2,
\qquad
B_2(\mathbf{X}) = \frac{1}{n^2}\!\!\sum_{1\le j<\ell\le p}\!\!
\bigl(4\lambda_{j\ell}-2(r_j+r_\ell)+n\bigr)^2.
\end{equation}
Thus $B_1$ is a squared-deviation penalty on the replications
$\{r_j\}$ alone, and $B_2$ is a squared-deviation penalty on the
replications and pairwise concurrences $\{r_j,\lambda_{j\ell}\}$.
When the counts attain their arithmetic targets
$r_j\equiv\bar r$ and $\lambda_{j\ell}\equiv\bar\lambda$, then \eqref{eq:B1B2_counts} reduces to
\begin{equation}\label{eq:B1B2_fully_balanced}
B_1 = \frac{p}{n^2}(2\bar r - n)^2,
\qquad
B_2 = \frac{1}{n^2}\binom{p}{2}(4\bar\lambda-4\bar r+n)^2,
\end{equation}
quantities determined by $(n,p,k)$ alone that are the attainable minima of $B_1$ and $B_2$ whenever full balance is feasible. The
canonical example is the full $k$-combination design
$\mathbf{X}_{\mathrm{full}}$ with $n=\tilde\lambda\binom{p}{k}$ introduced above, which attains the closed-form reference values
\begin{equation}\label{eq:B_full}
B_1(\mathbf{X}_{\mathrm{full}}) =
\frac{p\bigl(2\tilde\lambda\binom{p-1}{k-1}-n\bigr)^2}{n^2},
\qquad
B_2(\mathbf{X}_{\mathrm{full}}) =
\frac{\binom{p}{2}\bigl(4\tilde\lambda\binom{p-2}{k-2}
-4\tilde\lambda\binom{p-1}{k-1}+n\bigr)^2}{n^2}.
\end{equation}
Note that these two values serve as the reference $B_1$ and $B_2$ in the efficiency metrics of Section~\ref{sec:simulation}.

\subsection{Relationship between TCARD and BIBD}\label{subsec:bibd}
 
The structure of a two-level $k$-TCARD $\mathcal{D}(n,2^p,k)$ is closely related to, but formally distinct from, that of a balanced incomplete block design. Recall that a \emph{balanced incomplete block design} (BIBD) with parameters $(p,b,r_{\mathrm{B}},k, \lambda_{\mathrm{B}})$ is a treatment-assignment scheme for experimental units within blocks, in which each of the $p$ treatments appears in exactly $r_{\mathrm{B}}$ blocks, each block contains exactly $k$ treatments, and every unordered pair of treatments
co-occurs in exactly $\lambda_{\mathrm{B}}$ blocks. A TCARD, by contrast, specifies the factor combination used in each experimental run: in every run exactly $k$ of the $p$ factors depart from a
baseline, while the remaining factors stay at baseline. Thus, BIBD theory concerns balanced treatment allocation across blocks, whereas TCARD theory concerns balanced factor activation and co-activation
across runs. 

Despite this difference in framing, both objects share the same combinatorial skeleton: a collection of $k$-subsets of a ground set of size $p$. We make this shared structure explicit and then show that the fully balanced TCARDs are exactly those whose supports form a BIBD. Specifically, for any TCARD $\mathbf{X}\in\mathcal{D}(n,2^p,k)$, each row
$\mathbf{x}_i\in\{0,1\}^p$ is uniquely identified with the $k$-subset
\begin{equation*}
S_i \;:=\; \{\,j\in\{1,\ldots,p\} : x_{ij}=1\,\}
\;\subset\; \{1,\ldots,p\}
\end{equation*}
of factors active in run $i$. The design is therefore equivalent to
the multiset $\mathcal{S}(\mathbf{X}) := \{S_1,\ldots,S_n\}$ of $k$-subsets of the factor index set $\{1,\ldots,p\}$. Conversely, given any such
collection $\mathcal{S}$ of $n$ $k$-subsets, we denote by $\mathbf{X}(\mathcal{S})$ the corresponding TCARD, recovered by row-indicators. This correspondence lets us move freely between the matrix view $\mathbf{X}$ and the set-system view $\mathcal{S}(\mathbf{X})$. Under this identification, a BIBD $(p,b,r_{\mathrm{B}},k, \lambda_{\mathrm{B}})$ on $p$ treatments with block size $k$ and $b=n$ blocks is itself a collection of $n$ $k$-subsets of
$\{1,\ldots,p\}$. The map $\mathcal{S}\mapsto\mathbf{X}(\mathcal{S})$
therefore produces an associated TCARD $\mathbf{X}(\mathcal{S}_{\mathrm{BIBD}})\in\mathcal{D}(n,2^p,k)$, so-called BIBD-induced TCARD. The defining count properties of the BIBD translate directly into count summaries of the induced TCARD: every factor appears in
$r_{\mathrm{B}}=\bar r$ runs and every pair of factors co-appears
in $\lambda_{\mathrm{B}}=\bar\lambda$ runs, so the BIBD-induced TCARD satisfies the fully balanced conditions. The following proposition records this characterization: fully balanced TCARDs and BIBD-induced TCARDs are the same objects.


\begin{proposition}[BIBD characterization of full balance]
\label{prop:bibd-full-balance}
Let $\mathbf{X}\in\mathcal{D}(n,2^p,k)$. The following are equivalent:
\begin{enumerate}
\item[(a)] $\mathbf{X}$ is fully balanced according to Definition~\ref{def:balance};
\item[(b)] The support collection $\mathcal{S}(\mathbf{X})$ is a BIBD with
parameters $(p,b,r_{\mathrm{B}},k,\lambda_{\mathrm{B}}) =
(p,n,\bar r,k,\bar\lambda)$;
\item[(c)] $\mathbf{X}^\top\mathbf{X} = (\bar r - \bar\lambda)\,\mathbf{I}_p
+ \bar\lambda\,\mathbf{J}_p$.
\end{enumerate}
\end{proposition}

\subsection{Nearly Balanced TCARD}\label{subsec:nearly_balanced_design}
 
When the arithmetic targets $\bar r, \bar\lambda$ in~\eqref{eq:targets-value} are non-integer, full balance is infeasible and Proposition~\ref{prop:bibd-full-balance} does not apply. We therefore weaken full balance to its integer relaxation: replications are as equal as arithmetic permits, and pairwise
concurrences are as uniform as possible. This section makes that relaxation precise, establishes existence results, and shows that the resulting \emph{nearly balanced} TCARDs simultaneously minimize
the two lowest-order GWLP components $B_1$ and $B_2$. The formulation parallels the classical nearly balanced block designs of~\citet{cheng1981nearly}, tailored here to the TCARD setting.

For a TCARD $\mathbf{X}\in\mathcal{D}(n,2^p,k)$, summing the cardinality constraint $\sum_\ell x_{t\ell}=k$ against $x_{tj}$ over $t$ gives the row identity
\begin{equation}\label{eq:row_lambda}
\sum_{\ell\ne j}\lambda_{j\ell} = r_j(k-1),
\qquad j=1,\dots,p,
\end{equation}
which ties the concurrences involving factor $j$ to its
replication. To track per-factor deviations from the arithmetic
targets, we denote
\begin{equation}\label{eq:f-s-defs}
f := \lfloor\bar r\rfloor,
\qquad
s := p - (nk - pf),
\qquad
\boldsymbol\Lambda_i := \{\lambda_{ij} : j\ne i\},
\end{equation}
so that exactly $s$ factors have replication $f$ and the remaining $p-s$ have replication $f+1$, and $\boldsymbol\Lambda_i$ collects the concurrence counts involving factor~$i$. The following definition formalizes the notion of a nearly balanced TCARD, which relaxes the full balance conditions to their integer approximations.

 
\begin{definition}\label{def:nearly-balanced}
A TCARD $\mathcal{D}(n,2^p,k)$ is \emph{nearly balanced} if it satisfies:
\begin{enumerate}\itemsep2pt
\item[\textnormal{(NB1)}] $r_j\in\{f,\,f+1\}$ for all $j$ \textnormal{(}replications differ by at most one\textnormal{)}.
\item[\textnormal{(NB2)}] For each fixed $j$, the multiset $\boldsymbol{\Lambda}_j$ is as uniform as arithmetic permits:
\[
|\lambda_{j\ell}-\lambda_{j\ell'}|\le 1 \quad \text{for all } \ell,\ell'\neq j.
\]
Equivalently, every $\lambda_{j\ell}$ lies in $\bigl\{\lfloor r_j(k-1)/(p-1)\rfloor,\;\lceil r_j(k-1)/(p-1)\rceil\bigr\}$.
\end{enumerate}
\end{definition}
 
Condition NB1 is always achievable by distributing replications as evenly as possible. The main structural difficulty lies in NB2, which requires the concurrence pattern around each factor to be as uniform as arithmetic permits. Under NB2, each $\lambda_{j\ell}$ takes one of at most two values,
differing by one. This dichotomy admits a natural graph-theoretic encoding. Let
$\kappa := \lfloor (k-1)f/(p-1) \rfloor$ and
\begin{equation}\label{eq:remainder_form}
\omega \;:=\; p - 1 - \bigl((k-1)f - \kappa(p-1)\bigr)
\;\in\; \{0,1,\dots,p-1\}.
\end{equation}
For factors with $r_j = f$, $\kappa$ is the smaller concurrence level and $\kappa+1$ the larger. Exactly $\omega$ of the $p-1$ partners of factor $j$ appear with it $\kappa$ times and the remaining $p-1-\omega$ appear $\kappa+1$ times. Define the \emph{concurrence-excess graph} $G(\mathbf{X})$ to be the simple graph on vertex set $\{1,\dots,p\}$ in which $\{j,\ell\}$ is an edge if and only if $\lambda_{j\ell} = \kappa+1$. That is, $G$ records the pairs of factors whose concurrence exceeds the lower level.
Under NB1 and NB2, the degree $d_j$ of vertex $j$ is then forced:
\begin{equation}\label{eq:deg-forced}
d_j \;=\;
\begin{cases}
p-1-\omega, & \text{if } r_j=f,\\[2pt]
p-1-\omega+(k-1), & \text{if } r_j=f+1,
\end{cases}
\end{equation}
with one caveat: when $\omega < k-1$, the shift from $r_j=f$ to $r_j=f+1$ forces $d_j$ to exceed $p-1$, which is impossible in a simple graph. Resolving this case requires $\boldsymbol\Lambda_j$ to admit a third concurrence level, not merely two. Following \citet{cheng1981nearly}, we call the regime $\omega \ge k-1$ where the two-level structure is feasible \emph{Type~I}, and the regime $\omega < k-1$ \emph{Type~II}. Existence of a nearly balanced TCARD is thus equivalent to the existence of a simple graph
$G$ with the prescribed degree sequence (Type~I) or to a closely related combinatorial condition (Type~II).

\begin{theorem}\label{thm:NB_TCARD_exist}
Let $n,p,k$ be positive integers with $1\le k\le p$. Define $f,s,\kappa,\omega$ as above.
 
\noindent\textbf{(A) Type~I ($\omega\ge k-1$).}
If a nearly balanced TCARD of Type~I exists, then the degree sequence
\[
\underbrace{(p-\omega-1,\;\ldots,\;p-\omega-1)}_{s\;\text{times}},\quad
\underbrace{(p-\omega+k-2,\;\ldots,\;p-\omega+k-2)}_{(p-s)\;\text{times}}
\]
is graphical. In particular, the Erd\H{o}s--Gallai inequalities imply the necessary bound
\[
s(\omega-s+1) \;\le\; (p-s)(\omega-k+1).
\]
 
\noindent\textbf{(B) Type~II ($\omega<k-1$).}
If a nearly balanced TCARD of Type~II exists, then
\[
\omega+1 \;\le\; s \;\le\; p-k+\omega, \qquad
\omega s \;\text{even}, \qquad
(p-s)(p-k+\omega-s) \;\text{even}.
\]
 
\noindent\textbf{(C) Guaranteed existence in two boundary regimes.}
\begin{enumerate}\itemsep2pt
\item[\textnormal{(C1)}] $k=2$ \textnormal{(}pairwise-choice runs\textnormal{):} a nearly balanced TCARD always exists and is of Type~I.
\item[\textnormal{(C2)}] $k=p-1$ \textnormal{(}all-but-one active\textnormal{):} nearly balanced designs are obtained by distributing the $n$ runs among the $p$ singleton-complement supports as evenly as possible, supplemented by complete copies when needed.
\end{enumerate}
\end{theorem}

The proof of parts (A) and (B) reduces the concurrence-pattern
problem to the feasibility of a degree sequence in $G(\mathbf{X})$
and then invokes the Erd\H os--Gallai theorem. Details follow
\citet{cheng1981nearly} and are given in Appendix~\ref{app:thm1-proof}. For $3 \le k \le p-2$ outside the boundary regimes in (C), explicit
algebraic constructions are substantially more delicate and, to
our knowledge, remain open in general. We therefore rely on the
algorithmic search developed in Section~\ref{sec:algorithm}. That said, operationally, the graph-theoretic encoding suggests a
two-step search: (i) find a simple graph $G$ on $p$ vertices
realizing the forced degree sequence; and (ii) find a choice of $n$ rows,
each a $k$-subset of $\{1,\dots,p\}$, such that the resulting
per-factor and per-pair counts reproduce the replications
$\{r_j\}$ and the concurrences $\{\lambda_{j\ell}\}$. When both
steps succeed, the resulting TCARD is nearly balanced. If step
(ii) is infeasible, one tries another realization of $G$. The following example works through a successful instance for
$(p,k,n)=(6,3,7)$.

\begin{figure}[ht]
\centering
\includegraphics[width=1\textwidth]{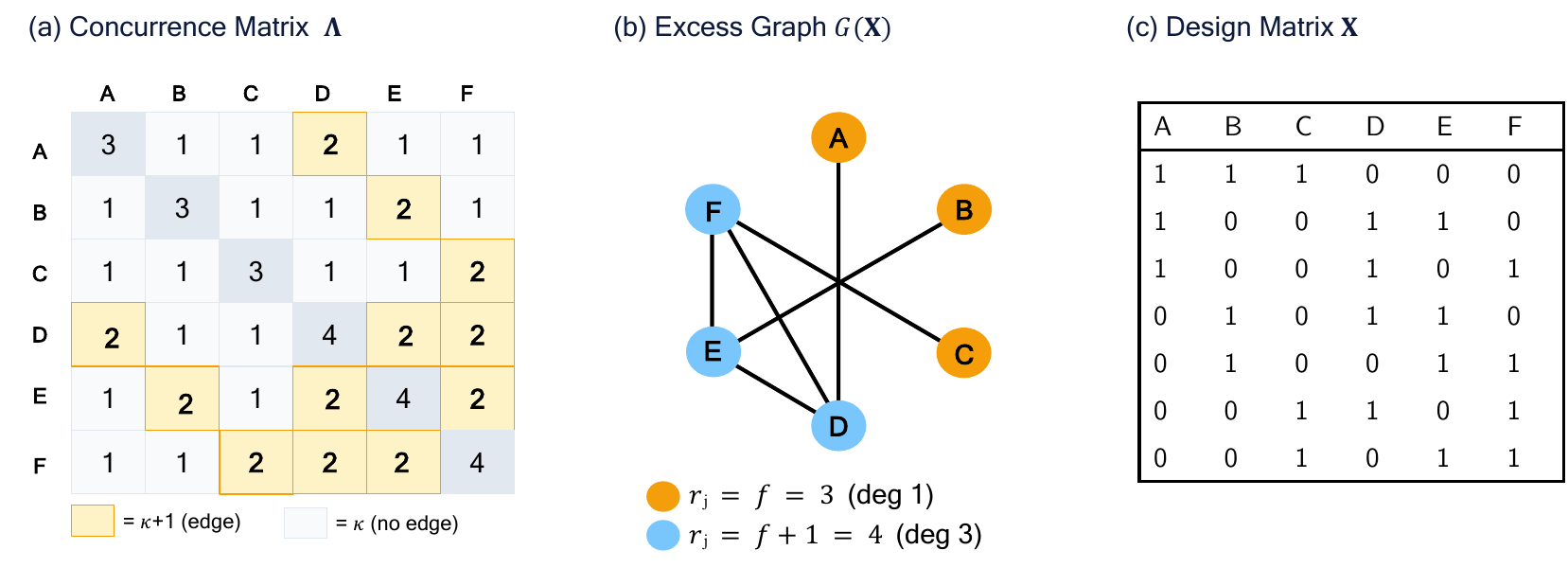}
\caption{Nearly balanced TCARD of Example~\ref{ex:nb-small}
for $(p,k,n)=(6,3,7)$.
\textbf{(a)} The Gram matrix $\mathbf{X}^\top\mathbf{X}$ (equivalently,
the concurrence matrix $\boldsymbol\Lambda$ padded by the replications
on the diagonal). Its diagonal entries $(3,3,3,4,4,4)$ are the
replications $r_j$, and its off-diagonal entries are the pairwise
concurrences $\lambda_{j\ell} \in \{1, 2\} = \{\kappa, \kappa+1\}$.
\textbf{(b)} Concurrence-excess graph $G(\mathbf{X})$: vertices
coloured by replication class, and edges joining exactly the pairs
$(j,\ell)$ with $\lambda_{j\ell} = \kappa+1 = 2$.
\textbf{(c)} Design matrix $\mathbf{X}\in\{0,1\}^{7\times 6}$: each
row has exactly $k=3$ active entries, and the column sums recover
the diagonal of (a).}
\label{fig:nb-small}
\end{figure}

\begin{example}[A nearly balanced TCARD for $(p,k,n)=(6,3,7)$]
\label{ex:nb-small}
The arithmetic targets are $\bar r = 3.5$ and $\bar\lambda = 1.4$,
both non-integer. Computing the structural parameters gives
$f = 3$, $s = 3$, $\kappa = 1$, and $\omega = 4$, with
$\omega \ge k-1 = 2$ falling in the Type~I regime. By Part (A)
of Theorem~\ref{thm:NB_TCARD_exist}, any nearly balanced TCARD
must induce a graph with degree sequence
$(1,1,1,3,3,3)$, which is graphical and satisfies the
Erd\H os--Gallai bound $s(\omega-s+1)=6 \le (p-s)(\omega-k+1)=9$. We now apply the two-step pipeline.

\emph{Step (i): Fix a realizing graph.} Take $G$ to be a
perfect matching between the three $r_j=3$ vertices $\{A,B,C\}$
and the three $r_j=4$ vertices $\{D,E,F\}$, together with a
triangle on $\{D,E,F\}$ (Figure~\ref{fig:nb-small}(b)); this
graph has the required degree sequence. Reading off edges yields
the candidate concurrence matrix
$\boldsymbol{\Lambda}$ with $\lambda_{j\ell} = 2$ for
pairs $(A,D), (B,E), (C,F), (D,E), (D,F), (E,F)$ and
$\lambda_{j\ell} = 1$ otherwise
(Figure~\ref{fig:nb-small}(a)).

\emph{Step (ii): Recover the design matrix.} We search for seven
3-subsets of $\{1,\ldots,6\}$ whose per-factor and per-pair counts
match the replications $r_j$ and the target concurrences
$\lambda_{j\ell}^\star$. Such a collection exists and is unique up
to row order: the supports 
$\{A,B,C\}$, $\{A,D,E\}$, $\{A,D,F\}$, $\{B,D,E\}$, $\{B,E,F\}$,
$\{C,D,F\}$, $\{C,E,F\}$ each used once, giving
the design in Figure~\ref{fig:nb-small}(c).

One can verify that the replications are $(3,3,3,4,4,4)$,
satisfying NB1. For each factor the concurrences satisfy
$\boldsymbol\Lambda_j \in \{\{1,1,1,1,2\}, \{1,1,2,2,2\}\}$,
both single-unit ranges, satisfying NB2. The resulting
GWLP components are $B_1(\mathbf{X}) = 6/49 \approx 0.122$ and
$B_2(\mathbf{X}) = 9/7 \approx 1.286$, which match the
theoretical minima among TCARDs satisfying NB1 and NB2 by
Theorem~\ref{thm:B1B2_opt}.
\end{example}

With near balance formalized, we now state the following result:

\begin{theorem}[Optimality of nearly balanced TCARD for $B_1$ and $B_2$]\label{thm:B1B2_opt}
Fix $n, p, k$ and let $\mathbf{X}$ range over
$\mathcal{D}(n, 2^p, k)$. With $B_1, B_2$ defined
by~\eqref{eq:B1B2_counts}:
\begin{enumerate}\itemsep2pt
\item[\textnormal{(i)}] $B_1$ is minimized if and only if \textnormal{NB1} in Definition~\ref{def:nearly-balanced} holds.
\item[\textnormal{(ii)}] Conditional on \textnormal{NB1}, $B_2$ is minimized if and only if \textnormal{NB2} in Definition~\ref{def:nearly-balanced} holds.
\end{enumerate}
\end{theorem}
 
Theorem~\ref{thm:B1B2_opt} formalizes the intuition that nearly balanced TCARDs are simultaneously optimal for the two lowest-order projection imbalances: part~(i) shows that $B_1$ forces replications to two adjacent integers, and part~(ii) shows that, conditional on this replication balance, $B_2$ is minimized by flattening co-occurrences as much as arithmetic and symmetry permit. The proof is given in Appendix Section~\ref{app:thm2-proof}.

Note that Theorems~\ref{thm:NB_TCARD_exist} and~\ref{thm:B1B2_opt}
characterize nearly balanced TCARDs, but do not by themselves
provide a practical way to find one. The two-step pipeline in Example~\ref{ex:nb-small} rests on solving an integer linear feasibility problem. Two features make this problem difficult
for the parameter ranges of interest to us. First, the number of
unknowns $\binom{p}{k}$ grows combinatorially in $(p,k)$. Second, feasibility depends on the choice of realizing graph $G$ in step~(i).
When step~(ii) returns infeasible for one $G$, one must try
another non-isomorphic realization of the forced degree sequence,
and the number of such realizations is itself combinatorial.
Therefore these features make the exact reconstruction pipeline
impractical beyond small $(p,k)$. We therefore step back from enforcing Definition~\ref{def:nearly-balanced}
as a hard constraint. Rather than search for TCARDs that are
exactly nearly balanced, we introduce in Section~\ref{sec:opt_criterion} a real-valued criterion $\Phi_{\mathrm{BCD}}(\mathbf{X})$ that quantifies how close a TCARD is to nearly balanced, vanishing precisely on the nearly
balanced regime identified by Theorems~\ref{thm:NB_TCARD_exist}
and~\ref{thm:B1B2_opt}. Minimizing $\Phi_{\mathrm{BCD}}$ is a
continuous relaxation of the exact feasibility problem: it
relaxes ``counts hit their targets'' to ``counts are close to their
targets'', is tractable at moderate $(p,k)$, and recovers a nearly
balanced TCARD whenever one exists.

\section{Optimality Criteria}\label{sec:opt_criterion}

In this section, we construct and analyze a scoring function for
TCARDs. The function is designed to be minimized. It is
\emph{counts-only} and \emph{model-free}, depending on
$\mathbf{X}$ only through the replication counts $\{r_i\}$ and
the pairwise concurrence counts $\{\lambda_{ij}\}$. Its two
components penalize deviation from the arithmetic replication
and concurrence targets of~\eqref{eq:targets-value}. Minimizing this composite penalty drives the replications and
pairwise concurrences toward the balanced structure of
Definition~\ref{def:balance}. Under a main-effects model, this
translates into statistical optimality links with the $(M,S)$-optimality
principle, the centered $\mathrm{UE}(s^2)$ criterion of
\citet{jones2014optimal}, and Bayesian $D$-optimality. These
statistical optimality links are developed in
Section~\ref{subsec:opt_eff}.

For a TCARD $\mathbf{X} \in \mathcal{D}(n,2^p,k)$, define the
\emph{replication-imbalance} and \emph{concurrence-dispersion}
summaries
\begin{equation}\label{eq:V1V2}
\mathcal{V}_1(\mathbf{X}) \;:=\; \sum_{i=1}^p (r_i - \bar r)^2,
\qquad
\mathcal{V}_2(\mathbf{X}) \;:=\!\!\sum_{1\le i<j\le p}\!
(\lambda_{ij} - \bar\lambda)^2,
\end{equation}
with $\bar r, \bar\lambda$ as in~\eqref{eq:targets-value}.
These are the squared deviations of the count vectors
$\{r_i\}$ and $\{\lambda_{ij}\}$ from their arithmetic targets:
$\mathcal{V}_1(\mathbf{X}) = 0$ precisely when $\mathbf{X}$ is
balanced in replications, $\mathcal{V}_2(\mathbf{X}) = 0$
precisely when $\mathbf{X}$ is balanced in pairwise
concurrences, and $\mathcal{V}_1 = \mathcal{V}_2 = 0$
characterizes the fully balanced case of
Definition~\ref{def:balance}. When the arithmetic targets are
non-integer, neither $\mathcal{V}_1$ nor $\mathcal{V}_2$ can
attain zero, and minimizing either quantity
individually involves different structural trade-offs. The
criterion introduced next combines them into a single
scalarization.

We evaluate a design by the \emph{Balanced Concurrence Deviation}
(BCD) criterion
\begin{equation}\label{eq:Phi_BCD}
\Phi_{\mathrm{BCD}}(\mathbf{X})
\;=\;
\frac{w_1}{p}\,\mathcal{V}_1(\mathbf{X})
\;+\;
\frac{w_2}{\binom{p}{2}}\,\mathcal{V}_2(\mathbf{X}),
\qquad w_1, w_2 > 0,
\end{equation}
a weighted average of the two dispersion summaries, with scaling factors $1/p$ and $1/\binom{p}{2}$ that place each term on a per-count basis. The weights $(w_1, w_2)$ trade off
replication imbalance against concurrence dispersion;
sensible defaults derived from the link to $\mathrm{UE}(s^2)$
are given in Section~\ref{subsec:opt_eff}, and a
simulation-based tuning procedure in
Section~\ref{subsec:tuning}.

\subsection{Statistical Optimality Justification}\label{subsec:opt_eff}

We ground the counts-only criterion $\Phi_{\mathrm{BCD}}$ in
standard statistical efficiency by analyzing TCARDs under a
main-effects linear model. Let $\mathbf{y} \in \mathbb{R}^n$
denote the response vector and $\mathbf{X} \in \{0,1\}^{n\times p}$
the TCARD design matrix under constraint~\eqref{eq:tcard}. We assume
\begin{equation}\label{eq:main-model}
\mathbf{y} \;=\; \mu\,\mathbf{1}_n \;+\; \mathbf{X}\boldsymbol{\beta}
\;+\; \boldsymbol{\varepsilon},
\qquad
\boldsymbol{\varepsilon} \sim \mathcal{N}(\mathbf{0},\sigma^2 \mathbf{I}_n),
\end{equation}
where $\mu$ is an intercept, $\boldsymbol{\beta} \in \mathbb{R}^p$
is the vector of main effects to be estimated, and
$\boldsymbol{\varepsilon}$ is a vector of i.i.d.\ Gaussian errors.
After eliminating the intercept via the centering projection
$\mathbf{H} := \mathbf{I}_n - \tfrac{1}{n}\mathbf{1}_n\mathbf{1}_n^\top$,
the information matrix for $\boldsymbol{\beta}$ is
\begin{equation}\label{eq:C_def}
\mathbf{C} \;:=\; \mathbf{X}^\top \mathbf{H} \mathbf{X}
\;=\; \mathbf{X}^\top\mathbf{X} - \tfrac{1}{n}\mathbf{X}^\top\mathbf{1}_n\mathbf{1}_n^\top\mathbf{X},
\end{equation}
a symmetric positive semidefinite $p \times p$ matrix whose
spectrum governs main-effect estimation precision. All
optimality statements in this subsection are expressed in terms
of~$\mathbf{C}$.

The replication and concurrence counts control $\mathbf{C}$
directly. For $i \ne j$,
$(\mathbf{X}^\top\mathbf{X})_{ii} = r_i$ and
$(\mathbf{X}^\top\mathbf{X})_{ij} = \lambda_{ij}$, so
\begin{equation}\label{eq:C_counts}
\mathbf{C}_{ii} = r_i - \frac{r_i^2}{n},
\qquad
\mathbf{C}_{ij} = \lambda_{ij} - \frac{r_i r_j}{n}
\quad(i \ne j).
\end{equation}
The diagonal of $\mathbf{C}$ is determined by the
replications $\{r_i\}$ and controls the total information
$\operatorname{tr}(\mathbf{C})$. The off-diagonal is determined
by the concurrences $\{\lambda_{ij}\}$ and controls the extent
of non-orthogonality among main-effect contrasts. The two
terms in $\Phi_{\mathrm{BCD}}$ can therefore be read as
surrogates for (i) maximizing total main-effect information
and (ii) reducing imbalance and correlation among main-effect
estimates.

The remainder of this subsection develops three connections
between $\Phi_{\mathrm{BCD}}$ and standard optimality
benchmarks. Section~\ref{subset:MS} links
$\Phi_{\mathrm{BCD}}$ to the $(M,S)$-principle of
\citet{eccleston1974theory} through trace and eigenvalue-spread
identities. Section~\ref{sec:relation_UE} establishes an exact
algebraic equivalence to the centered $\mathrm{UE}(s^2)$
criterion of \citet{jones2014optimal}. Section~\ref{sec:relation_Dopt}
connects $\Phi_{\mathrm{BCD}}$ to Bayesian $D$-optimality through
a perturbation bound and a large-$\alpha$ trace expansion, which is a
determinant-based benchmark that remains well-defined despite
the structural singularity of $\mathbf{C}$ under the TCARD
constraint.

\subsubsection{Connection to the $(M,S)$-Optimality Principle}\label{subset:MS}
Following the two-stage framework of \cite{eccleston1974theory}, a design is called \((M,S)\)-optimal if it (i) first maximizes the total information $\operatorname{tr}(\mathbf{C})$ and then (ii) subject to that maximum, minimizes the dispersion of the positive eigenvalues of $\mathbf{C}$. The second stage can be expressed by minimizing an eigenvalue–spread index such as $\operatorname{tr}(\mathbf{C}^2)$, which favors the equalization of the nonzero eigenvalues. When the equal-eigenvalue point is feasible (e.g., BIBD), it yields simultaneous A-, D-, and E-optimality by standard convexity arguments, as theoretically illustrated in \cite{kiefer1958nonrandomized}. When exact balance is not attainable, \citet{jacroux1980some} show that the corresponding nearly balanced incidence structures remain $(M,S)$-optimal under their sufficient conditions. These conditions carry over naturally to the nearly balanced TCARD setting. This motivates examining whether our counts-only criterion $\Phi_{\mathrm{BCD}}$ can be used to search directly for $(M,S)$-optimal designs, with nearly balanced TCARDs as the primary targets.

Based on model~\eqref{eq:main-model}, the following trace identities will be used repeatedly:
\begin{equation}\label{eq:tr_identities}
\operatorname{tr}(\mathbf{C})=nk-\frac{1}{n}\sum_{i=1}^p r_i^2,\qquad
\operatorname{tr}(\mathbf{C}^2)=\sum_{i=1}^p\Big(r_i-\frac{r_i^2}{n}\Big)^2
+2\sum_{1\le i<j\le p}\Big(\lambda_{ij}-\frac{r_i r_j}{n}\Big)^2.
\end{equation}
The first identity shows that maximizing $\tr(\mathbf C)$ over the TCARD space is equivalent to
minimizing $\sum_i r_i^2$. The second identity decomposes
$\tr(\mathbf C^2)$ into a replication contribution and a concurrence contribution, clarifying how
replication balance and concurrence uniformity jointly govern eigenvalue dispersion. We next summarize how the counts-only criterion $\Phi_{\mathrm{BCD}}$ in~\eqref{eq:Phi_BCD} aligns with the
$(M,S)$ principle. 

\begin{proposition}[Forcing the \(M\)-stage]\label{prop:force_M}
If
\begin{equation}\label{eq:w_threshold}
\frac{w_1}{w_2}\;>\;\frac{2}{p-1}\Big\lfloor \frac{nk}{p}\Big\rfloor (k-1),
\end{equation}
then any global minimizer of \(\Phi_{\mathrm{BCD}}\) has no pair \(a,b\) with
\(r_a\ge r_b+2\). Hence its replication vector attains the discrete minimum of \(\sum_i r_i^2\), namely
\(r_i\in\{\lfloor nk/p\rfloor,\ \lceil nk/p\rceil\}\) for all \(i\). By~\eqref{eq:tr_identities},
\(\operatorname{tr}(\mathbf{C})\) is therefore maximized.
\end{proposition}

\textit{Remark}: Condition~\eqref{eq:w_threshold} is a worst-case sufficient requirement that guarantees replication balance at the global optimum by ensuring that every feasible balancing swap must decrease \(\Phi_{\mathrm{BCD}}\). In practice, much smaller ratios often work for two concrete reasons tied to the structure of the centered criterion. 
A single balancing swap reduces \(\sum_i(r_i-\bar r)^2\) by a fixed amount, which translates to an \(O(1/p)\) decrease in \(\Phi_{\mathrm{BCD}}\) after averaging. The same swap modifies at most \(2(k-1)\) pairwise concurrences,      giving an \(O(1/p^2)\) impact on \(\Phi_{\mathrm{BCD}}\). Consequently, for moderate \(p\) the replication component naturally dominates the early optimization dynamics even when \(w_1=w_2\), which explains the good empirical performance of the default choice \(w_1=w_2=1\) in our numerical simulations.
For tuning, one may still use conservative, verifiable checks based on the current design. Let $\lambda_{\max}(\mathbf{X}) := \max_{i<j}\lambda_{ij}$ denote the
maximum pairwise concurrence. Since a single balancing swap modifies only the \(2(k-1)\) pairs incident to the swapped columns, a sufficient condition ensuring that such a swap (when \(r_a\ge r_b+2\)) decreases \(\Phi_{\mathrm{BCD}}\) is $\frac{w_1}{w_2} > \frac{2}{p-1}\lambda_{\max}(\mathbf{X})(k-1)$,
which is usually far smaller than the worst-case bound and can be evaluated on the fly.

Proposition~\ref{prop:force_M} ensures the first stage is achieved by any minimizer of \(\Phi_{\mathrm{BCD}}\). We next show that, conditional on this replication vector, the second term of \(\Phi_{\mathrm{BCD}}\) drives exactly the \(S\)-stage.

\begin{proposition}\label{prop:force_S}
Fix a replication vector \(\mathbf{r}\) that minimizes \(\sum_i r_i^2\).
\begin{enumerate}\itemsep0.2em
\item If the replication vector is perfectly balanced (\(r_i\equiv\bar r\)), then minimizing \(\sum_{1\le i<j\le p}(\lambda_{ij}-\bar\lambda)^2\) is equivalent to minimizing \(\operatorname{tr}(\mathbf{C}^2)\).
\item If \(\mathbf r\) has two replication levels \(r_i\in\{\lfloor\bar r\rfloor,\lceil\bar r\rceil\}\), then on the local-search neighborhood generated by class-preserving rectangle swaps (swaps that exchange a $1$
  and a $0$ within the same row, restricted to columns belonging
  to the same replication class), minimizing \(\sum_{1\le i<j\le p}(\lambda_{ij}-\bar\lambda)^2\) is equivalent to minimizing \(\operatorname{tr}(\mathbf{C}^2)\).
\end{enumerate}
\end{proposition}

These propositions clarify the distinct roles of the two components in \(\Phi_{\mathrm{BCD}}\), but they do not imply that the current coordinate-exchange algorithm in Section~\ref{subsec:CE} automatically follows a strict two-stage \((M,S)\) trajectory for arbitrary weight choices. Because the implemented CE procedure uses within-row \(1\!\leftrightarrow\!0\) swaps, replication balancing and concurrence regularization may remain coupled throughout the search. Thus, the present CE algorithm should be viewed as a practical descent method for the composite criterion \(\Phi_{\mathrm{BCD}}\), rather than as a procedure that intrinsically guarantees exact \((M,S)\)-optimality. That said, under sufficiently large values of \(w_1/w_2\), one can show that once the search enters the discrete-optimal replication set, it cannot leave it, and any further accepted move must be within \(M\)-stage. From that point onward, descent is driven entirely by the concurrence component. Even then, such a strict separation does not necessarily yield the best designs for other optimality criteria or downstream performance measures. If exact \((M,S)\)-optimization is desired, one may instead combine \(\Phi_{\mathrm{BCD}}\) with a replication-preserving local-search neighborhood.

\subsubsection{Connection to Centered $UE(s^2)$ Criterion}
\label{sec:relation_UE}

Quadratic dispersion criteria play a central role in assessing projection quality for nonregular and supersaturated designs. Among these, the centered $UE(s^2)$ criterion of \cite{jones2014optimal} measures variability in the inner-product structure of the augmented design matrix after centering. While originally developed for supersaturated designs, it has recently been applied to large row-constrained supersaturated designs by \cite{smucker2025large}. In this section, we demonstrate that within our framework, $UE(s^2)$ aligns directly with our proposed design criterion.

Specifically, we show that under the two-level treatment cardinality constraint, the centered $UE(s^2)$ criterion is exactly equivalent to a particular weighted form of the counts-only criterion $\Phi_{\mathrm{BCD}}$. This equivalence holds algebraically for all admissible $(n,p,k)$ and does not rely on asymptotic arguments or approximations.

To apply $\mathrm{UE}(s^2)$ in our setting, we re-code the TCARD design matrix $\mathbf{X} \in \{0,1\}^{n \times p}$ on the $\pm 1$ scale by setting $\mathbf{Z} := 2\mathbf{X} - \mathbf{1}_n\mathbf{1}_p^\top$ and form the augmented design matrix $\widetilde{\mathbf{X}} := [\mathbf{1}_n, \mathbf{Z}] \in \mathbb{R}^{n\times(p+1)}$
corresponding to an main-effects model with intercept on the $\pm 1$ factors. Its Gram matrix is defined as
$\mathbf{S} := \widetilde{\mathbf{X}}^\top\widetilde{\mathbf{X}}$, and has block structure
\[
\mathbf{S}
\;=\;
\begin{pmatrix}
n & \mathbf{s}_0^\top \\
\mathbf{s}_0 & \mathbf{S}_Z
\end{pmatrix},
\]
where $\mathbf{s}_0 = \mathbf{Z}^\top \mathbf{1}_n \in \mathbb{R}^p$ collects the intercept--factor inner products, and $\mathbf{S}_Z = \mathbf{Z}^\top \mathbf{Z} \in \mathbb{R}^{p\times p}$ collects the factor--factor inner products. The centered $\mathrm{UE}(s^2)$ criterion is the sum of squared deviations of the off-diagonal entries of $\mathbf{S}$ from their blockwise averages:
\begin{equation}\label{eq:UE_def}
\mathrm{UE}(s^2)(\mathbf{X})
\;:=\;
\sum_{i=1}^p \bigl(S_{1,i+1} - \bar S_{1,\cdot}\bigr)^2
\;+\;
\sum_{1 \le i < j \le p}
\bigl(S_{i+1,j+1} - \bar S_{\cdot,\cdot}\bigr)^2,
\end{equation}
where $\bar S_{1,\cdot} := \frac{1}{p}\sum_{i=1}^p S_{1,i+1}$ is the mean of the intercept--factor block and
$\bar S_{\cdot,\cdot} := \binom{p}{2}^{-1}\sum_{1\le i<j\le p} S_{i+1, j+1}$ is the mean of the factor--factor block. The following theorem establishes an exact algebraic identity between $\mathrm{UE}(s^2)(\mathbf{X})$ and the two count-based dispersion quantities $\mathcal{V}_1(\mathbf{X})$ and $\mathcal{V}_2(\mathbf{X})$ of~\eqref{eq:V1V2}.


\begin{theorem}[Equivalence between $UE(s^2)$ and $\Phi_{\mathrm{BCD}}$]
\label{thm:UE_equiv}
Let $\mathbf{X} \in \{0,1\}^{n\times p}$ be a two-level TCARD with row sums equal to $k$, and let
$\mathbf{Z} = 2\mathbf{X} - \mathbf{1}_n\mathbf{1}_p^\top$ be its $\pm 1$ re-coding. Define the centered $\mathrm{UE}(s^2)$ criterion by~\eqref{eq:UE_def}, and the count-based quantities $\mathcal{V}_1, \mathcal{V}_2$ by~\eqref{eq:V1V2}. Then:
\begin{equation}
UE(s^2)(\mathbf{X})
= 4(p - 4k + 3)\,\mathcal{V}_1(\mathbf{X})
+ 16\,\mathcal{V}_2(\mathbf{X}).
\label{eq:UE_identity}
\end{equation}
Moreover, if $p>4k-3$, then minimizing $UE(s^2)$ is equivalent to minimizing
$\Phi_{\mathrm{BCD}}(\mathbf{X})$ with weights $w_1 = (p - 4k + 3)/4p$ and $w_2 = (p-1)/2p$.
\end{theorem}

When $p<4k-3$, the replication-imbalance term in~\eqref{eq:UE_identity} receives a negative coefficient, so minimizing $UE(s^2)$ would favor replication dispersion rather than balance. In this case, \(\Phi_{\mathrm{BCD}}\) is better behaved because it retains nonnegative penalties on both replication and concurrence imbalances and continues to target nearly balanced TCARDs. More broadly, Theorem~\ref{thm:UE_equiv} provides a transparent combinatorial interpretation of $UE(s^2)$ under cardinality constraints. At the same time, it motivates $\Phi_{\mathrm{BCD}}$ as a natural objective for computation because it preserves the same structure while enabling fast incremental evaluation during search.

\subsubsection{Bridge to Bayesian-D Optimality and Perturbation Analysis}\label{sec:relation_Dopt}

We now establish an information-theoretic interpretation of the counts-only criterion by connecting it to Bayesian D-optimality \citep{dumouchel1994simple, jones2008bayesian}. Before introducing Bayesian-D, we highlight a structural singularity that arises in the TCARD setting under the intercept-eliminated main-effects model. In our construction the treatment-cardinality constraint induces a nontrivial linear dependence among the centered columns of $\mathbf{X}$, so that $\mathbf{C}=\mathbf{X^\top H X}$ is positive semidefinite but structurally singular. As a consequence, the usual D-optimality criterion $\log\det(\mathbf{C})$ is not well-defined, and determinant- or inverse-based summaries (including A- and E-type metrics) can be numerically unstable whenever small eigenvalues are present. To obtain a criterion that is always well-defined, we adopt the Gaussian Bayesian linear model of \cite{jones2014optimal} for the centered main effects,
\begin{equation}\label{eq:bayes_model}
\mathbf{Y} \mid \boldsymbol{\beta} \sim \mathcal N(\mathbf{X}\boldsymbol{\beta},\sigma^2 \mathbf{I}_n),\qquad
\boldsymbol{\beta} \sim \mathcal N(\theta,\sigma^2 \alpha^{-1} \mathbf{I}_p),
\end{equation}
where $\alpha>0$ is a prior precision parameter and $\boldsymbol{\theta}\in\mathbb{R}^p$ is fixed. This yields posterior precision $\alpha \mathbf{I}_p+\sigma^{-2}\mathbf{X^\top X}$ and leads to the Bayes-$D$ objective
\begin{equation}\label{eq:BayesD_normalized}
f_\alpha(\mathbf{X})=\log\det(\mathbf{C}+\alpha \mathbf{I}_p),\qquad \alpha>0,
\end{equation}
after restricting attention to centered main-effect contrasts, which
replaces $\mathbf{X}^\top\mathbf{X}$ by $\mathbf{C}=\mathbf{X}^\top\mathbf{H}\mathbf{X}$. 

To benchmark a generic TCARD design matrix $\mathbf{X}$ against the ideal balanced structure, we introduce a reference information matrix $\mathbf{C}_0$ built from the arithmetic targets $\bar r$ and $\bar\lambda$ of Section~\ref{sec:algebraic}. Denote $\mathbf{C}_0 := \boldsymbol{\Lambda}_0 - \tfrac{1}{n}\mathbf{r}_0\mathbf{r}_0^\top$ as the centered information matrix.
Here $\mathbf{r}_0:= \bar r\,\mathbf{1}_p$ is the replication vector of a fully balanced TCARD. $\boldsymbol{\Lambda}_0$ with $(\boldsymbol{\Lambda}_0)_{ii} := \bar r$ and $(\boldsymbol{\Lambda}_0)_{ij} := \bar\lambda$ is the corresponding Gram matrix $\mathbf{X}_0^\top\mathbf{X}_0$ had one existed. 
Direct algebra yields the compact spectral form
\begin{equation}\label{eq:C0_spectral}
\mathbf{C}_0
\;=\;
\delta\,\bigl(\mathbf{I}_p - \tfrac{1}{p}\mathbf{1}_p\mathbf{1}_p^\top\bigr),
\end{equation}
where $\delta := \bar r - \bar\lambda = \frac{nk(p-k)}{p(p-1)}$, so that $\mathbf{C}_0$ has eigenvalues $\{0, \delta, \dots, \delta\}$ with multiplicities $1$ and $p-1$, respectively.

The remainder of this subsection uses $\mathbf{C}_0$ to develop a closed-form bound on the Bayes-$D$ efficiency gap, and clarifies the two limiting regimes of $f_\alpha$: (1) under strong regularization (large $\alpha$), $f_\alpha$ admits a trace expansion consistent with the $(M,S)$-principle; (2) under weak regularization (small $\alpha$), it converges to centered $D$-optimality on the $(p-1)$-dimensional contrast space.


Define the Bayes-$D$ efficiency gap of $\mathbf{X}$ relative to the balanced reference $\mathbf{C}_0$ by
\begin{equation}\label{eq:Galpha_def}
G_\alpha(\mathbf{X}) \;:=\; f_\alpha(\mathbf{C}_0) - f_\alpha(\mathbf{X})
\;=\; \log\det(\mathbf{C}_0 + \alpha\mathbf{I}_p)
- \log\det(\mathbf{C} + \alpha\mathbf{I}_p).
\end{equation}
Since $\mathbf{C}_0$ is the centered information matrix of an idealized fully balanced design (not always realizable under cardinality constraints), $G_\alpha(\mathbf{X}) \ge 0$ in the regimes where $\mathbf{C}_0$ majorizes admissible $\mathbf{C}$'s. Our goal is to bound $G_\alpha(\mathbf{X})$ in terms of the count-based dispersion quantities $\mathcal{V}_1, \mathcal{V}_2$ of~\eqref{eq:V1V2}.

\begin{lemma}\label{lem:E-Frobenius_BayesD}
Let $\mathbf{X} \in \mathcal{D}(n, 2^p, k)$ with centered information matrix $\mathbf{C} = \mathbf{X}^\top\mathbf{H}\mathbf{X}$, and let $\mathbf{C}_0$ be the balanced reference. Writing
$\mathbf{E} := \mathbf{C} - \mathbf{C}_0$,
\[
\sqrt{\mathcal{V}_1 + 2\mathcal{V}_2}
\;-\;
\frac{1}{n}\bigl(2\sqrt{p}\,\bar r\,\sqrt{\mathcal{V}_1}
+ \mathcal{V}_1\bigr)
\;\le\;
\|\mathbf{E}\|_F
\;\le\;
\sqrt{\mathcal{V}_1 + 2\mathcal{V}_2}
\;+\;
\frac{1}{n}\bigl(2\sqrt{p}\,\bar r\,\sqrt{\mathcal{V}_1}
+ \mathcal{V}_1\bigr),
\]
where $\mathcal{V}_1 = \mathcal{V}_1(\mathbf{X})$ and
$\mathcal{V}_2 = \mathcal{V}_2(\mathbf{X})$ as
in~\eqref{eq:V1V2}. The lower bound is informative only when its right-hand side is nonnegative.
\end{lemma}


For later reference we denote the upper and lower bound
functions in Lemma~\ref{lem:E-Frobenius_BayesD} by $U(\mathcal{V}_1, \mathcal{V}_2)$ and $L(\mathcal{V}_1, \mathcal{V}_2)$.
Lemma~\ref{lem:E-Frobenius_BayesD} provides a purely counts-based control of $\|\mathbf{C}-\mathbf{C}_0\|_F$.
When this perturbation is small relative to the prior precision $\alpha$, a standard log-determinant perturbation
argument yields an explicit upper bound on the Bayes-$D$ gap.

\begin{theorem}\label{thm:gap-upper_BayesD}
Fix $\alpha>0$ and let $\mathbf{A}:=\mathbf{C}_0+\alpha \mathbf{I}_p$.
Let $0<\rho<1$ and assume 
\begin{equation}\label{eq:neighborhood_BayesD}
U(\mathcal{V}_1,\mathcal{V}_2) \le \rho\,\alpha.
\end{equation}
Then
\begin{equation}\label{eq:gap-upper_BayesD}
G_\alpha(\mathbf{X})
\le
\|\mathbf{A}^{-1}\|_F\,U(\mathcal{V}_1,\mathcal{V}_2)
+\frac{1}{2(1-\rho)}\cdot\frac{U(\mathcal{V}_1,\mathcal{V}_2)^2}{\alpha^2},
\end{equation}
with
\[
\|\mathbf{A}^{-1}\|_F=\sqrt{\frac{1}{\alpha^2}+\frac{p-1}{(\alpha+\delta)^2}},
\qquad
\delta=\frac{n k(p-k)}{p(p-1)}.
\]
\end{theorem}


Theorem~\ref{thm:gap-upper_BayesD} gives a closed-form Bayes-$D$ performance guarantee in a neighborhood around
the balanced reference. We next connect Bayes-$D$ to classical spectral criteria by specializing $f_\alpha$
to the strong- and weak-prior regimes.

\begin{theorem}\label{thm:BayesD_main}
Let $\boldsymbol{\Delta}_{n,p,k}$ be the finite class of feasible two-level TCARD designs with $n$ runs, $p$ factors,
and row-sum $k$. For $\mathbf{X}\in\boldsymbol{\Delta}_{n,p,k}$ let $\mathbf{C}=\mathbf{X}^\top \mathbf{H} \mathbf{X}$ and $f_\alpha(\mathbf{X})=\log\det(\mathbf{C}+\alpha \mathbf{I}_p)$.
Assume $\mathrm{rank}(\mathbf{C})=p-1$ for all $\mathbf{X}\in\boldsymbol{\Delta}_{n,p,k}$.

\textup{(i) (Small $\alpha$).}
There exists $\alpha^\ast>0$ such that for every $\alpha\in(0,\alpha^\ast)$,
\[
\arg\max_{\mathbf{X}\in\boldsymbol{\Delta}_{n,p,k}} f_\alpha(\mathbf{X})
\ \subseteq\
\arg\max_{\mathbf{X}\in\boldsymbol{\Delta}_{n,p,k}} \log\operatorname{pdet}(\mathbf{C}),
\]
where $\operatorname{pdet}(\mathbf{C})=\prod_{i=2}^p \lambda_i(\mathbf{C})$ is the pseudo-determinant. In particular, every Bayes-$D$ optimal design for sufficiently small $\alpha$ is centered $D$-optimal.

\textup{(ii) (Large $\alpha$ and $\Phi_{\mathrm{BCD}}$-minimizers).}
Assume further that $p\mid nk$, so that $\boldsymbol{\Delta}_{n,p,k}$ contains at least one column-balanced design
($r_i\equiv \bar r$). Let $X_\phi$ be a global minimizer of
$\Phi_{\mathrm{BCD}}(\mathbf{X})$ over $\boldsymbol{\Delta}_{n,p,k}$.
Then there exists $\alpha^{\ast\ast}>0$ such that $\mathbf{X}_\phi$ is Bayes-$D$ optimal for every $\alpha>\alpha^{\ast\ast}$.
\end{theorem}

Note that under the TCARD constraint, $\mathbf{C}\mathbf{1}_p = \mathbf{0}$ always holds, so $\operatorname{rank}(\mathbf{C})\leq p-1$. The assumption requires equality, i.e.\ that no additional linear
dependence exists among the centred columns of $\mathbf{X}$.
A sufficient condition is $n \geq p-1$ and the absence of any
non-trivial linear combination of the centred columns that sums to
zero.  In particular, if the design matrix $\mathbf{X}$ restricted to any
$p-1$ columns has rank $p-1$ after centering (e.g.\ if the design is
connected in the graph-theoretic sense of the treatment-concurrence
graph), then $\operatorname{rank}(\mathbf{C})=p-1$.
For designs produced by the CE algorithm, connectivity is almost
certain for $n \geq p-1$. For small $n$ or degenerate designs the
rank may fall below $p-1$ and the Bayes-D framework with the
pseudo-determinant should be used instead.

To gain intuition for Theorem~\ref{thm:BayesD_main}, it is helpful to view $f_\alpha(\mathbf{X})=\sum_{i=1}^p\log(\lambda_i(\mathbf{C})+\alpha)$: \textup{(i) Small $\alpha$.}
Because $C$ has one structural zero eigenvalue, $f_\alpha(\mathbf{X})$ always contains a common $\log(\alpha)$ term.
When $\alpha$ is small, the remaining part behaves like $\sum_{i=2}^p\log\lambda_i(\mathbf{C})$, so Bayes-$D$ ranks designs
almost exactly by the pseudo-determinant of $\mathbf{C}$, i.e., centered $D$-optimality on the $(p-1)$-dimensional contrast space.\textup{(ii) Large $\alpha$.}
When $\alpha$ is large, each term $\log(\lambda_i(\mathbf{C})+\alpha)$ changes only slightly across designs, and a standard
log-determinant expansion yields
\[
\log\det(\mathbf{C}+\alpha \mathbf{I}_p)=p\log\alpha+\alpha^{-1}\mathrm{tr}(\mathbf{C})-\tfrac{1}{2}\alpha^{-2}\mathrm{tr}(\mathbf{C}^2)+\cdots.
\]
Thus, for large $\alpha$, Bayes-$D$ rewards designs with larger total information $\mathrm{tr}(\mathbf{C})$ and, secondarily,
with a more even spectrum (smaller $\mathrm{tr}(\mathbf{C}^2)$). This explains why the same replication and concurrence summaries
used in $\Phi_{\mathrm{BCD}}$ provide a natural and effective surrogate for Bayes-$D$ when $\alpha$ is sufficiently large.

\section{Algorithmic Construction}\label{sec:algorithm}

\subsection{Coordinate-Exchange Algorithm}\label{subsec:CE}

To minimize the counts-only criterion \eqref{eq:Phi_BCD} over the TCARD space
$\{\mathbf{X}\in\{0,1\}^{n\times p}:\ \mathbf{X}\mathbf{1}_p=k\mathbf{1}_n\}$,
we employ a cyclic coordinate-exchange (CE) procedure in the spirit of
\citet{meyer1995coordinate}, specialized to enforce the fixed row-sum constraint.
Because each run must contain exactly $k$ active factors, a single-bit flip is infeasible in general.
Instead, each local move exchanges a currently active coordinate with an inactive one within the same row:
for a given run $t$, choose $a\in S_t=\{i:x_{ti}=1\}$ and $b\in \bar S_t=\{j:x_{tj}=0\}$, and perform
$x_{ta}\leftarrow 0,\ x_{tb}\leftarrow 1$. This preserves $\sum_{i=1}^p x_{ti}=k$ by construction. 

The objective \eqref{eq:Phi_BCD} depends on $\mathbf{X}$ only through the replication vector
$\mathbf{r}=\mathbf{X}^\top \mathbf{1}_n$ and concurrence counts $\lambda_{ij}=(\mathbf{X}^\top \mathbf{X})_{ij}$ for $i\neq j$.
We maintain $(\mathbf{r},\{\lambda_{ij}\})$ incrementally, so that each proposed swap can be scored by a local change $\Delta\Phi_{\mathrm{BCD}}$ without recomputing global summaries.
Specifically, swapping $a\to 0$ and $b\to 1$ in row $t$ only changes $r_a,r_b$ and the pair counts
$\{\lambda_{aj},\lambda_{bj}: j\in S_t\setminus\{a\}\}$. And all other terms in \eqref{eq:Phi_BCD} remain unchanged.
Therefore each candidate swap can be evaluated in $O(k)$ time, and a row-wise search over all
$k(p-k)$ swaps yields an $O(k^2(p-k))$ worst-case update per row (often reduced in practice by early acceptance or restricting to a subset of candidate exchanges). We run CE in sweeps over rows, accepting for each row the best-improving exchange among all
$a\in S_t,\ b\in\bar S_t$. 
The algorithm terminates when a complete sweep produces no improvement, and we use multiple random feasible starts, returning the design with the smallest achieved $\Phi_{\mathrm{BCD}}$. The algorithm terminates in finite steps because 
$\Phi_{\mathrm{BCD}}$ is non-negative on the finite 
feasible set $\mathcal{D}(n,2^p,k)$ and each accepted 
swap strictly decreases it. 
Algorithm~\ref{alg:CE_TCARD} summarizes the procedure.

\begin{algorithm}[hpt]
\DontPrintSemicolon
\caption{Coordinate-exchange (CE) minimization of $\Phi_{\mathrm{BCD}}(\mathbf{X})$ over $D(n,2^p,k)$}\label{alg:CE_TCARD}
\KwIn{$(n,p,k)$; weights $(w_1,w_2)$; restarts $R$; max sweeps $S_{\max}$.}
\KwOut{$\mathbf{X}^\star\in\{0,1\}^{n\times p}$ with $\mathbf{X}^\star\mathbf{1}_p=k\mathbf{1}_n$.}
$\bar r \leftarrow nk/p$;\quad $\bar\lambda \leftarrow n\binom{k}{2}/\binom{p}{2}$\;
$\Phi_{\mathrm{BCD}}^\star \leftarrow +\infty$;\quad $\mathbf{X}^\star \leftarrow \emptyset$\;

\For(\tcp*[f]{multi-start}){$r=1,\dots,R$}{
    Generate a feasible start $\mathbf{X}$ by sampling a $k$-subset $S_t\subset\{1,\dots,p\}$ for each row $t$ and setting
    $x_{ti}=\mathbb{I}(i\in S_t)$\;
    Compute counts $r_i=\sum_{t=1}^n x_{ti}$ and $\lambda_{ij}=\sum_{t=1}^n x_{ti}x_{tj}$ ($i<j$)\;
    Compute $\Phi_{\mathrm{BCD}} \leftarrow \Phi_{\mathrm{BCD}}(X)$\;

    \For(\tcp*[f]{CE sweeps}){$s=1,\dots,S_{\max}$}{
        $\texttt{improved}\leftarrow \texttt{false}$\;

        \For(\tcp*[f]{row-wise search}){$t=1,\dots,n$}{
            $S_t\leftarrow \{i:\ x_{ti}=1\}$;\quad $\bar S_t\leftarrow \{j:\ x_{tj}=0\}$\;
            $\Delta_{\min}\leftarrow 0$;\quad $(a^\star,b^\star)\leftarrow \emptyset$\;

            \For{$a\in S_t$}{
                \For{$b\in \bar S_t$}{
                    Compute $\Delta\Phi_{\mathrm{BCD}}(a\to 0,b\to 1)$ using only affected terms in \eqref{eq:Phi_BCD}:
                    $(r_a,r_b)$ and $\{\lambda_{aj},\lambda_{bj}: j\in S_t\setminus\{a\}\}$\;
                    \If{$\Delta\Phi_{\mathrm{BCD}} < \Delta_{\min}$}{
                        $\Delta_{\min}\leftarrow \Delta\Phi_{\mathrm{BCD}}$;\quad $(a^\star,b^\star)\leftarrow (a,b)$\;
                    }
                }
            }

            \If(\tcp*[f]{accept improving swap}){$\Delta_{\min}<0$}{
                $x_{ta^\star}\leftarrow 0$;\quad $x_{tb^\star}\leftarrow 1$\;
                $r_{a^\star}\leftarrow r_{a^\star}-1$;\quad $r_{b^\star}\leftarrow r_{b^\star}+1$\;
                \For{$j\in S_t\setminus\{a^\star\}$}{
                    $\lambda_{a^\star j}\leftarrow \lambda_{a^\star j}-1$;\quad
                    $\lambda_{b^\star j}\leftarrow \lambda_{b^\star j}+1$\;
                }
                $\Phi_{\mathrm{BCD}} \leftarrow \Phi_{\mathrm{BCD}} + \Delta_{\min}$;\quad $\texttt{improved}\leftarrow \texttt{true}$\;
            }
        }

        \If(\tcp*[f]{no improvement in a full sweep}){$\texttt{improved}=\texttt{false}$}{
            \textbf{break}\;
        }
    }

    \If{$\Phi_{\mathrm{BCD}} < \Phi_{\mathrm{BCD}}^\star$}{
        $\Phi_{\mathrm{BCD}}^\star\leftarrow \Phi_{\mathrm{BCD}}$;\quad $\mathbf{X}^\star\leftarrow \mathbf{X}$\;
    }
}
\Return{$\mathbf{X}^\star$}\;
\end{algorithm}

\subsection{Simulation-based tuning of parameters.}\label{subsec:tuning}
The criterion $\Phi_{\mathrm{BCD}}(\mathbf{X})$ is model-free and depends only on the replication and concurrence counts, but the
relative weight $w_1$ (with $w_2$ fixed to 1) controls a practically important trade-off: larger $w_1$ places
more emphasis on marginal replication regularity, whereas a smaller $w_1$ allows CE to prioritize pairwise concurrence regularity. In real experimental planning, however, the design must be finalized before collecting any responses. Hence, the tuning parameter $w_1$ cannot be calibrated using observed data. We therefore tune $w_1$ offline using a simulation plan that emulates the intended downstream analysis task under a range of plausible regimes. In this paper, we illustrate the procedure using sparse subspace selection. Throughout our simulations, we take the $F_1$ score as the target metric, though the same workflow readily accommodates other downstream tasks and performance measures. Accordingly, unless otherwise stated, the remainder of the paper uses this sparse subspace selection problem as the running example for illustration. Readers may replace it with their own analysis pipeline and corresponding metric when tuning $w_1$ for their application.

Concretely, we specify a screening plan indexed by $(q,h)$, where $q$ denotes the assumed sparsity level (e.g., number of active
main effects) and $h$ indexes additional scenario difficulty settings (e.g., effect size/SNR, correlation, or noise level). For each candidate $w_1$ on a logarithmic grid, we (i) construct a TCARD
design via a fixed-budget CE minimization of $\Phi_{\mathrm{BCD}}(\mathbf{X})$, and (ii) evaluate that design under the
screening plan using Monte Carlo experiments, producing a performance metric such as $F_1$ summarized by a robust statistic (here, Monte Carlo mean). This yields raw scores $\mathrm{Score}_{q,h}(w_1)$.

Because the absolute scale of $\mathrm{Score}_{q,h}(w_1)$ can differ across scenarios, we compare candidate $w_1$
values within each $(q,h)$ after standardizing scores across the grid $\mathcal W=\{w_1^{(1)},\ldots,w_1^{(G)}\}$.
For each fixed $(q,h)$, define
\[
\bar S_{q,h}=\frac{1}{G}\sum_{g=1}^G \mathrm{Score}_{q,h}\!\left(w_1^{(g)}\right),\qquad
s_{q,h}=\left\{\frac{1}{G-1}\sum_{g=1}^G\Big(\mathrm{Score}_{q,h}\!\left(w_1^{(g)}\right)-\bar S_{q,h}\Big)^2\right\}^{1/2},
\]
and the standardized score
\[
Z_{q,h}(w_1)=\frac{\mathrm{Score}_{q,h}(w_1)-\bar S_{q,h}}{\max\{s_{q,h},\varepsilon\}},
\]
where $\varepsilon>0$ is a small constant for numerical stability when $s_{q,h}$ is near zero. We then summarize
performance across $h$ by treating $h$ as indexing a user-specified distribution of operating regimes and taking
\[
\mu_q(w_1)=\mathbb{E}_{h}\!\left\{Z_{q,h}(w_1)\right\},
\]
implemented in practice by an empirical average over $h\in\mathcal H$. Finally, for each $q$ we select
\[
w_1^\star(q)=\arg\max_{w_1\in\mathcal W}\ \mu_q(w_1),
\]
and retain the corresponding CE-optimized design $\mathbf{X}\!\left(w_1^\star(q)\right)$. Note that this procedure is directly actionable in practice: before running a real study, the practitioner encodes prior scientific knowledge or engineering judgment about plausible sparsity and signal-to-noise regimes into the screening plan, tunes $w_1$ using simulation only, and then commits to the resulting TCARD design to collect real responses. Importantly, the tuning uses no information from future outcomes, preserving the design-before-data principle while aligning the design criterion with the intended inferential objective.

\begin{algorithm}[hpt]
\DontPrintSemicolon
\caption{Simulation-based tuning of $w_1$ for $\Phi_{\mathrm{BCD}}$-driven TCARD designs}\label{alg:tune_w1_phi_meanmean}
\KwIn{$(n,p,k)$; CE controls; fixed $w_2=1$; grid $\mathcal W=\{w_1^{(1)},\dots,w_1^{(G)}\}$;
screening plan indexed by $(q,h)\in\mathcal Q\times\mathcal H$; Monte Carlo size $B$;
downstream task: sparse subspace selection with interactions; metric: interaction $F_1$.}
\KwOut{$w_1^\star(q)$ and tuned designs $X^\star(q)$ for each $q\in\mathcal Q$.}

\BlankLine
\For{$g=1,\dots,G$}{
  $w_1 \leftarrow w_1^{(g)}$\;
  Construct $X^{(g)}$ by multi-start CE minimizing $\Phi_{\mathrm{BCD}}(X)$ with weights $(w_1,1)$\;

  \ForEach{$(q,h)\in\mathcal Q\times\mathcal H$}{
    \For{$b=1,\dots,B$}{
      Generate a Monte Carlo dataset under setting $(q,h)$ and fit the screening pipeline on $X^{(g)}$. Compute $F_{1,b}^{(g)}(q,h)$. \;
    }
    Compute the Monte Carlo mean:
    $
      \widehat{F}_1^{(g)}(q,h)\leftarrow \frac{1}{B}\sum_{b=1}^B F_{1,b}^{(g)}(q,h).
    $
    Set $\mathrm{Score}_{q,h}(w_1^{(g)})\leftarrow \widehat{F}_1^{(g)}(q,h)$\;
  }
}

\BlankLine
\ForEach{$(q,h)\in\mathcal Q\times\mathcal H$}{
  Compute the across-grid mean and sd:
  $\bar S_{q,h}\leftarrow \frac{1}{G}\sum_{g=1}^G \mathrm{Score}_{q,h}(w_1^{(g)})$, and $
    s_{q,h}\leftarrow \left\{\frac{1}{G-1}\sum_{g=1}^G\big(\mathrm{Score}_{q,h}(w_1^{(g)})-\bar S_{q,h}\big)^2\right\}^{1/2}$.\;
  \For{$g=1,\dots,G$}{
    Standardize across $w_1$ (z-score):
    $
      Z_{q,h}(w_1^{(g)})\leftarrow \frac{\mathrm{Score}_{q,h}(w_1^{(g)})-\bar S_{q,h}}{\max\{s_{q,h},\varepsilon\}}.
    $
  }
}

\BlankLine
\ForEach{$q\in\mathcal Q$}{
  \For{$g=1,\dots,G$}{
    Aggregate over regimes by the empirical mean over $h$: $\mu_q(w_1^{(g)})\leftarrow \frac{1}{|\mathcal H|}\sum_{h\in\mathcal H} Z_{q,h}(w_1^{(g)})$.
  }
  Select $g^\star(q)\leftarrow \arg\max_{g}\ \mu_q(w_1^{(g)})$\;
  Set $w_1^\star(q)\leftarrow w_1^{(g^\star(q))}$;\quad $X^\star(q)\leftarrow X^{(g^\star(q))}$\;
}
\Return{$\{w_1^\star(q),X^\star(q)\}_{q\in\mathcal Q}$}\;
\end{algorithm}

\section{Numerical Study}\label{sec:simulation}

\subsection{Simulation design and factor settings}\label{sec:sim_setup}
We run a simulation study to benchmark sparse-design construction and optimization algorithms under a treatment-cardinality constraint. A design is represented by a binary matrix $\mathbf{X}\in\{0,1\}^{n\times p}$, where $\mathbf{X}_{ti}=1$ indicates that factor $i$ is activated in run $t$. The constraint requires exact sparsity in every run, i.e., $\mathbf{X}\mathbf{1}_p=k\,\mathbf{1}_n$, so each row contains exactly $k$ ones. Each simulated instance is indexed by the number of factors $p$, the run-size ratio $n/p$, and the sparsity ratio $k/p$. We explicitly study three levels of constraint strength by varying $k/p$: strong constraint ($\frac{k}{p}=0.1$), moderate constraint ($\frac{k}{p}=0.25$), and weak constraint ($\frac{k}{p}=0.5$). For each constraint level, we consider $p\in\{20,40,60,80\}$ and $\frac{n}{p}\in\{1.5,\,3\}$ so that $n=\lfloor (n/p)\,p\rfloor$ and $k=\lfloor (k/p)\,p\rfloor$. This range spans settings that are representative of practice: for example, $p\approx20–40$ corresponds to small-to-moderate screening studies in engineering prototyping or industrial process development, while $p\approx60–80$ matches higher-dimensional screening problems such as feature-ablation experiments in machine learning pipelines, large-factor factorial screening, or assay/omics-style perturbation panels where only a small subset of factors can be activated per run. The two run budgets $n/p\in \{1.5,3\}$ reflect common regimes in which experiments are limited to on the order of one to a few runs per factor, ranging from tightly budgeted pilot studies to more adequately sampled screening campaigns. For each $(p,n/p,k/p)$ configuration, we generate $R=30$ independent replications.

Within each replication, we construct feasible designs and, when applicable, further refine them using a feasibility-preserving local search. Feasibility is always enforced by the exact row-sum constraint $\mathbf{X}\mathbf{1}_p=k\,\mathbf{1}_n$. We deliberately focus on \emph{model-free} design criteria in the optimization stage, i.e., objectives that depend only on the combinatorial structure of $\mathbf{X}$ or inter-run Hamming distances, so that all competitors are compared on the same footing without privileging a particular response model.  This restriction is imposed to ensure a fair comparison across criteria.

Feasible initial designs are generated by fast construction schemes, including \emph{Random-TCARD} (independent random \(k\)-subsets in each row) and greedy heuristics that target marginal replication balance, with optional encouragement of pairwise concurrence regularity. Starting from such a feasible seed, we apply the CE local search described in Algorithm~\ref{alg:CE_TCARD}. Across all CE runs, the search template is fixed; only the objective used to score and accept swaps varies. The full list of compared methods and their objective definitions is given in Section~\ref{sec:sim_methods}. To ensure fair comparison across iterative optimizers, all CE methods use the same initialization protocol and stopping rule. Unless otherwise stated, we use a standardized greedy seed followed by a light random perturbation. Each run terminates when it reaches a prescribed computational budget (such as a maximum number of sweeps or a wall-clock cap) or when a full sweep produces no improvement beyond a specified tolerance. For objectives with expensive evaluations, each sweep may examine only a fixed number of randomly sampled candidate swaps per row; for objectives with fast incremental updates, all $k(p-k)$ within-row swaps are evaluated.

Each returned design is evaluated using a unified set of diagnostics, covering balance and concurrence structure, and downstream statistical performance. Importantly, information-matrix-based quantities are used only as post-hoc diagnostics to avoid model-dependent advantages. The primary evaluation metrics are summarized in Section~\ref{sec:sim_metrics}.

\subsection{Methods for comparison}\label{sec:sim_methods}
We compare a CE framework applied to several \emph{model-free} criteria, together with greedy feasible-construction baselines. Across the CE runs, the algorithmic template is held fixed and methods differ only in the criterion used to evaluate and accept exchanges. 


For the main comparison (Section~\ref{subsubsec:main_sim}), we fix the weights in the $\Phi_{\mathrm{BCD}}$ at \(w_1=w_2=1\) across all scenarios and evaluation metrics. This yields a parameter-free default and avoids redefining the criterion from one scenario to another. Under this default specification, we compare methods only in terms of structural design quality and projection-based information quality. In a separate sensitivity analysis (Section~\ref{subsubsec:tune_w1}), we additionally examine a data-free reweighting strategy described in Algorithm~\ref{alg:tune_w1_phi_meanmean}: a target-aware oracle tuning scheme indexed by the projection order \(q\). These auxiliary experiments are used to assess the adaptivity of \(\Phi_{\mathrm{BCD}}(\mathbf X)\) for downstream screening performance.

The other CE variants considered are as follows. First, the CE-UE$(s^2)$ optimizes the classical UE$(s^2)$ criterion. For two-level TCARD designs we have already established an exact identity showing that UE$(s^2)$ is equivalent to a weighted combination of the same two dispersion components underlying $\Phi_{\mathrm{BCD}}(\mathbf{X})$, differing only through the weights determined by $(p,k)$. Including UE$(s^2)$ therefore provides a principled benchmark that targets the same structural quantities as $\Phi_{\mathrm{BCD}}$, but with the weight profile implied by UE$(s^2)$.

To disentangle the roles of the two structural components, CE-V1 and CE-V2 optimize the marginal replication dispersion and pairwise concurrence dispersion separately. Specifically, CE-V1 minimizes $\mathcal{V}_1(\mathbf{X})$ and CE-V2 minimizes $\mathcal{V}_2(\mathbf{X})$, as defined in~\eqref{eq:V1V2}.
These variants isolate the contribution of each term and check that gains are not driven solely by optimizing one component.

We also include distance-based benchmarks that explicitly target space-filling behavior in Hamming space. CE-$\Phi_p$ (Morris–Mitchell) optimizes the classical Morris–Mitchell criterion,
\[
\Phi_p(\mathbf{X})=\Big(\sum_{i<j} d_{ij}^{-p}\Big)^{1/p}, \qquad p>0,
\]
where $d_{ij}$ denotes the Hamming distance between runs $i$ and $j$. Although $\Phi_p$ is typically motivated for continuous space-filling designs, it is equally well-defined for two-level designs and naturally measures separation of feasible treatment combinations under the TCARD constraint. Moreover, for binary designs the Hamming distance admits an explicit representation in terms of pairwise overlaps: if each row has exactly $k$ ones, then $d_{ij}=2(k-\lambda_{ij}^{\operatorname{(row)}})$, where $\lambda_{ij}^{\operatorname{(row)}}$ is the number of shared active factors between runs $i$ and $j$. Hence $\Phi_p(\mathbf{X})$ can be viewed as a function of the row-wise concurrence structure (equivalently, $\mathbf{XX}^T$). We emphasize, however, that this is a run-to-run criterion and is not determined solely by the column-wise summaries $(\mathbf{r},\boldsymbol{\Lambda})$. It captures a different aspect of design geometry than the replication/concurrence dispersion measures.

Similarly, CE-maximin and CE-minimax optimize two robustness-style distance criteria derived from $\{d_{ij}\}$. The maximin criterion maximizes $\min_{i<j} d_{ij}$, whereas the minimax criterion controls worst-case proximity by discouraging very small distances through a worst-case penalty. These criteria provide standard robustness-oriented baselines that explicitly target worst-case inter-run separation. Like Morris–Mitchell, these objectives can be expressed through the row-wise overlap counts $\lambda_{ij}^{\operatorname{(row)}}$, but they are not reducible to the column-wise summaries $(\mathbf{r},\boldsymbol{\Lambda})$. We therefore view them as complementary baselines that target worst-case inter-run separation rather than balance of factor usage.

Finally, we consider two greedy feasible-construction baselines: \textit{greedy-rep} and \textit{greedy-rep-pair}. Both sequentially build a feasible TCARD design using lightweight heuristics that target marginal replication balance, with \textit{greedy-rep-pair} additionally encouraging pairwise concurrence regularity. In our experiments, \textit{greedy-rep} is typically much worse than the other methods across most metrics. Including it in the main figures compresses the y-axis and obscures the differences among the CE variants, which are often clustered. For this reason, we omit \textit{greedy-rep} from the main-body plots. 
The \textit{greedy-rep-pair} is retained as a more competitive greedy baseline and as the default initialization for the CE runs. All methods enforce the same feasibility constraint, and within each replication we align random seeds across methods so that observed differences primarily reflect the choice of criterion and construction strategy rather than stochastic variation in initialization.

\subsection{Preference metrics}\label{sec:sim_metrics}
Each returned design is evaluated using a unified set of diagnostics. Using the earlier definitions of $\mathbf{r}$, $\boldsymbol{\Lambda}$, and the centered information matrix $\mathbf{C}$, we evaluate each returned design via (i) balance/dispersion summaries derived from $(\mathbf{r},\boldsymbol{\Lambda})$, (ii) eigenvalue-based summaries of $\mathbf{C}$ in the main comparison study. In the separate weight-tuning study, we focus on the downstream statistical performance under a common main-effect model, assessed through Monte Carlo experiments that apply the same estimation and variable-selection pipeline to every design. 
Importantly, all compared methods optimize model-free criteria that depend on $\mathbf{X}$ alone, so that no method is advantaged by a response-model-specific objective. 

To quantify aliasing induced by replication and concurrence irregularity, we report $B_1$--efficiency and $B_2$--efficiency, defined as efficiency--normalized versions of the imbalance measures in~\eqref{eq:B1B2_counts}. Let $B_1(\mathbf{X})$ and $B_2(\mathbf{X})$ be as defined in~\eqref{eq:B1B2_counts}. We define
$$
B_1-\operatorname{eff}(\mathbf{X})=\frac{B_1(\mathbf{X}_{\operatorname{full}})}{B_1(\mathbf{X})}, \quad B_2-\operatorname{eff}(\mathbf{X})=\frac{B_2(\mathbf{X}_{\operatorname{full}})}{B_2(\mathbf{X})},
$$
where $\mathbf{X}_{\operatorname{full}}$ is the \textit{full k-combination design} containing all $p\choose k$ admissible treatment combinations. Since $\mathbf{X}_{\operatorname{full}}$ generally has $n={p\choose k}$ and therefore does not match the run size of the candidate design, it is used only as an ideal reference: it represents the best aliasing/balance one can attain under the cardinality constraint when all combinations are available. With this normalization, values closer to $1$ indicate that design $\mathbf{X}$ approaches the best balance (least aliasing) achievable within the TCARD structure.


For information-matrix quality, we report a projection-based spectral diagnostic that remains well-defined and interpretable across sparse and tightly budgeted regimes. Specifically, for each design $\mathbf X$ and each projection order $q\in\{3,4,5\}$, we sample a large number of $q$-factor subsets $J\subset\{1,\dots,p\}$ and compute the log-determinant of the corresponding projected centered information matrix
\[
\mathbf C_J
=
\mathbf X_J^\top \mathbf X_J
-
\frac{1}{n}\mathbf r_J \mathbf r_J^\top,
\]
where $\mathbf X_J$ is the restriction of $\mathbf X$ to the columns indexed by $J$ and $\mathbf r_J$ is the corresponding replication vector. We then average $\log\det(\mathbf C_J)$ over the sampled $q$-subsets. Larger values indicate better-conditioned projections and more informative low-dimensional subspaces, which are especially relevant in the sparse-design settings studied here.

Downstream screening performance is evaluated only in the weight-tuning study for \(\Phi_{\mathrm{BCD}}\). We conduct Monte Carlo experiments under a common main-effects response model while retaining the projection-based setup. In each replication, we first sample a target subspace $J\subset\{1,\dots,p\}$ of size $q$, where $q$ is the projection order under consideration, and set the active main-effect set equal to this subspace $S=J$, $|S|=q$. Thus, the downstream task is explicitly aligned with recovery of a $q$-dimensional active subspace. Given $S$, the response is generated under the main-effects model
\begin{equation}\label{eq:main_truth_model_proj}
y_i
=
\beta_0
+
\sum_{j\in S}\beta_j x_{ij}
+
\varepsilon_i,
\qquad
\varepsilon_i\sim\mathcal N(0,\sigma^2),
\end{equation}
for $i=1,\dots,n$, where $\mathbf X\in\{0,1\}^{n\times p}$ is the TCARD design matrix. In our simulations, all active main effects are assigned a common positive coefficient, $\beta_j=\nu$ for $j\in S$, while $\beta_j=0$ for $j\notin S$. Thus, the scalar $\nu$ controls the signal amplitude. The intercept is set to $\beta_0=\mu$, so that the conditional mean is anchored at $\mu$ when all factors are at the baseline level $x_{ij}=0$.

All designs are evaluated under the same single-stage screening pipeline. We first center and column-standardize the main-effect design matrix and then fit an $\ell_1$-penalized least-squares model. For each candidate value of the lasso penalty $\lambda_{\mathrm{lasso}}$, coefficients below a small threshold are truncated to zero, and in the known-positive setting considered here any negative coefficients are also truncated to zero. Candidate models are then refit by ordinary least squares on the selected main effects, and the final model is chosen by BIC from the range of $\lambda_{\mathrm{lasso}}$ values retained after cross-validation. This yields a final selected active set $\widehat S$ for each realized response vector. From the final fitted model, we compute screening and prediction metrics for the main effects. Specifically, we report $
\mathrm{Precision}
=
\frac{|\widehat S\cap S|}{\max(|\widehat S|,1)}$, $\mathrm{Recall}
=
\frac{|\widehat S\cap S|}{|S|}$, and $
\mathrm{F1}
=
\frac{
2\,
\mathrm{Precision}\,
\mathrm{Recall}
}{
\mathrm{Precision}+\mathrm{Recall}
}$.
In addition, letting $\mu_i$ denote the true conditional mean from \eqref{eq:main_truth_model_proj} and $\widehat\mu_i$ the fitted conditional mean from the final refitted model, we compute the mean-function error $
\mathrm{MSE}_{\mu}
=
\frac{1}{n}\sum_{i=1}^n(\widehat\mu_i-\mu_i)^2$,
which measures how accurately the fitted model recovers the true signal. 
All reported downstream metrics are computed under the same data-generating mechanism and the same fitting pipeline for every competing design, so that performance differences can be attributed to the design itself rather than to differences in the estimation procedure.

\subsection{Simulation Results and Empirical Investigation}\label{sec:sim_results}

\subsubsection{Main comparison}\label{subsubsec:main_sim}

We begin with the main comparison under the default untuned specification of CE-$\Phi_{\mathrm{BCD}}$, with $w_1=w_2=1$, in order to assess its baseline performance relative to the competing model-free criteria across the three constraint regimes. Figure~\ref{fig:sim_core_combined} together with the projection-based summaries in Figure~\ref{fig:proj_logdet_mean_box} show that no single criterion uniformly dominates all diagnostics across all settings. Rather, the relative performance depends on which structural aspect of the TCARD design is emphasized. Against this backdrop, CE-$\Phi_{\mathrm{BCD}}$ is best interpreted as a balanced default criterion: it does not necessarily attain the optimum in every individual panel, but it remains consistently competitive across the aliasing-oriented metrics and the projection-based information summaries.

Under the strong-constraint regime ($k/p=0.1$), the cardinality restriction itself becomes the dominant structural bottleneck: with only about $10\%$ of factors active per run, the feasible TCARD space is substantially tighter than in the other regimes. As a result, most methods already achieve nearly saturated $B_1$-efficiency, indicating that marginal replication balance is comparatively easy to maintain once the row-sum constraint is so restrictive. By contrast, $B_2$-efficiency still retains some ability to discriminate among methods. Within this compressed setting, CE-$\Phi_{\mathrm{BCD}}$ remains consistently among the best-performing criteria: it preserves near-perfect $B_1$-efficiency, stays in the leading group for $B_2$-efficiency, and is also highly competitive in the projection-based log-determinant summaries across $q=3,4,5$. This pattern suggests that, even when the feasible design space is strongly restricted, the composite criterion continues to exploit the remaining degrees of freedom effectively. In particular, although the scope for improving marginal and pairwise summaries is necessarily limited under such extreme sparsity, CE-$\Phi_{\mathrm{BCD}}$ is still able to translate these limited adjustments into strong projected information quality without incurring visible deterioration in balance-based diagnostics. Distance-based criteria and single-component criteria can occasionally match or exceed CE-$\Phi_{\mathrm{BCD}}$ in individual panels, but their advantages are less uniform across the full collection of metrics. Thus, the main message in the strong-constraint regime is that CE-$\Phi_{\mathrm{BCD}}$ remains a particularly robust default: when only limited structural improvement is possible, it continues to deliver one of the strongest overall compromises between balance and projected information quality.

In the moderate-constraint regime ($k/p=0.25$), which is the most practically relevant setting among those considered here, the feasible design space remains rich enough for local search to produce meaningful improvements without becoming so unconstrained that most reasonable criteria perform similarly. As a result, the competing criteria behave more coherently, and many CE variants cluster near one in $B_1$-efficiency, indicating that marginal replication balance is relatively easy to achieve and remains stable across objectives. The more informative separation therefore shifts to $B_2$-efficiency and the projection-based information summaries, where CE-$\Phi_{\mathrm{BCD}}$ and CE-V2 are typically among the strongest methods. This is the regime in which the composite nature of $\Phi_{\mathrm{BCD}}$ is most clearly beneficial. Once marginal balance is largely controlled, further gains depend primarily on how effectively the criterion regularizes pairwise concurrence without sacrificing replication balance. CE-$\Phi_{\mathrm{BCD}}$ achieves this tradeoff well: it retains near-best $B_1$-efficiency while remaining in the leading group for $B_2$-efficiency and projected log-determinant across multiple $(p,n/p)$ combinations. More broadly, the moderate-constraint regime is where criterion choice appears to matter most in practice: marginal balance is no longer the main discriminator, and methods that better control concurrence regularity tend to deliver more favorable information geometry and, consequently, more stable downstream performance.

Under the weak-constraint regime ($k/p=0.5$), 
the feasible design space remains relatively rich, and many criteria already produce highly regular designs. In the aliasing-oriented summaries, CE-$\Phi_{\mathrm{BCD}}$ and CE-V1 remain essentially saturated in $B_1$-efficiency because both criteria directly penalize replication dispersion, whereas CE-$UE(s^2)$ exhibits noticeably weaker $B_1$-efficiency, especially under the tighter budget($n/p=1.5$), while achieving the strongest $B_2$-efficiency among the CE-based methods. This pattern is consistent with Theorem~\ref{thm:UE_equiv}: when $k/p=0.5$, the coefficient $p-4k+3=3-p$ is negative for all settings considered here, so minimizing $UE(s^2)$ places relatively greater emphasis on pairwise concurrence regularity and can reduce concurrence dispersion at the expense of increased variability in $\mathbf r$.
The more relevant question in this regime is how this structural tradeoff carries over to the projection-based information summaries. Empirically, CE-$UE(s^2)$ is often among the leading methods in the projected log-determinant panels, particularly when $n/p=3$, whereas its advantage becomes less consistent when $n/p=1.5$. 
A useful interpretation comes from the identity
$\mathbf{C}=\mathbf{X}^\top\mathbf{X}-\frac{1}{n}\mathbf{r}\mathbf{r}^\top$.
Here, \(\mathbf{X}^\top\mathbf{X}\) reflects the underlying concurrence structure that \(UE(s^2)\) tends to regularize, whereas the correction term \(\frac{1}{n}\mathbf{r}\mathbf{r}^\top\) captures the effect of replication imbalance. Thus, improving concurrence regularity does not automatically guarantee better projected information quality, because some of that gain may be offset if the replication counts become less even. This tradeoff is more favorable when the run budget is relatively generous. When \(n/p=3\), the benefit of improved concurrence regularity is more likely to remain visible in the projected information matrices, and CE-\(UE(s^2)\) is therefore often among the leading methods in the projected log-determinant panels. When \(n/p=1.5\), however, replication imbalance becomes more consequential relative to the available information, so the gains from better concurrence regularity can be partially offset. This explains why the projected log-determinant advantage of CE-\(UE(s^2)\) becomes less consistent under the tighter budget.

The competitive projected log-determinant of CE-$\Phi_p$ in this regime is also understandable from its objective. The criterion $\Phi_p$ promotes separation among runs in Hamming space. Under the weak-constraint regime, the feasible set of $k$-subsets is rich, so the coordinate-exchange search finds runs that are well-separated 
across many directions. This advantage is therefore a projection-geometry effect driven by run separation, not by direct regularization of replication or pairwise concurrence. CE-$\Phi_p$ does not explicitly target the column-wise aliasing structure, and its advantage in this regime is not expected to be uniform across the other regimes. Note that this improvement comes at a computational cost. Evaluating $\Phi_p$ requires all $\binom{n}{2}$ pairwise Hamming distances and updating $O(n)$ terms after each swap, compared with the $O(k)$ incremental update of $\Phi_{\mathrm{BCD}}$. In practice, CE-$\Phi_p$ requires approximately three to five times the computation time of CE-$\Phi_{\mathrm{BCD}}$ for the same number of restarts.

Taken together, no single criterion uniformly dominates all diagnostics across regimes. This is expected because the compared objectives emphasize different aspects of TCARD geometry. Accordingly, we interpret CE-\(\Phi_{\mathrm{BCD}}\) primarily as a balanced default rather than as a universally optimal rule for every metric. Its main advantage is that it jointly regularizes the two column-wise dispersion components while remaining model-free, which leads to consistently competitive projection and spectral performance and avoids the one-sided behavior seen in criteria that optimize only one component.

\begin{sidewaysfigure}[htp]
\centering
\includegraphics[width=\linewidth]{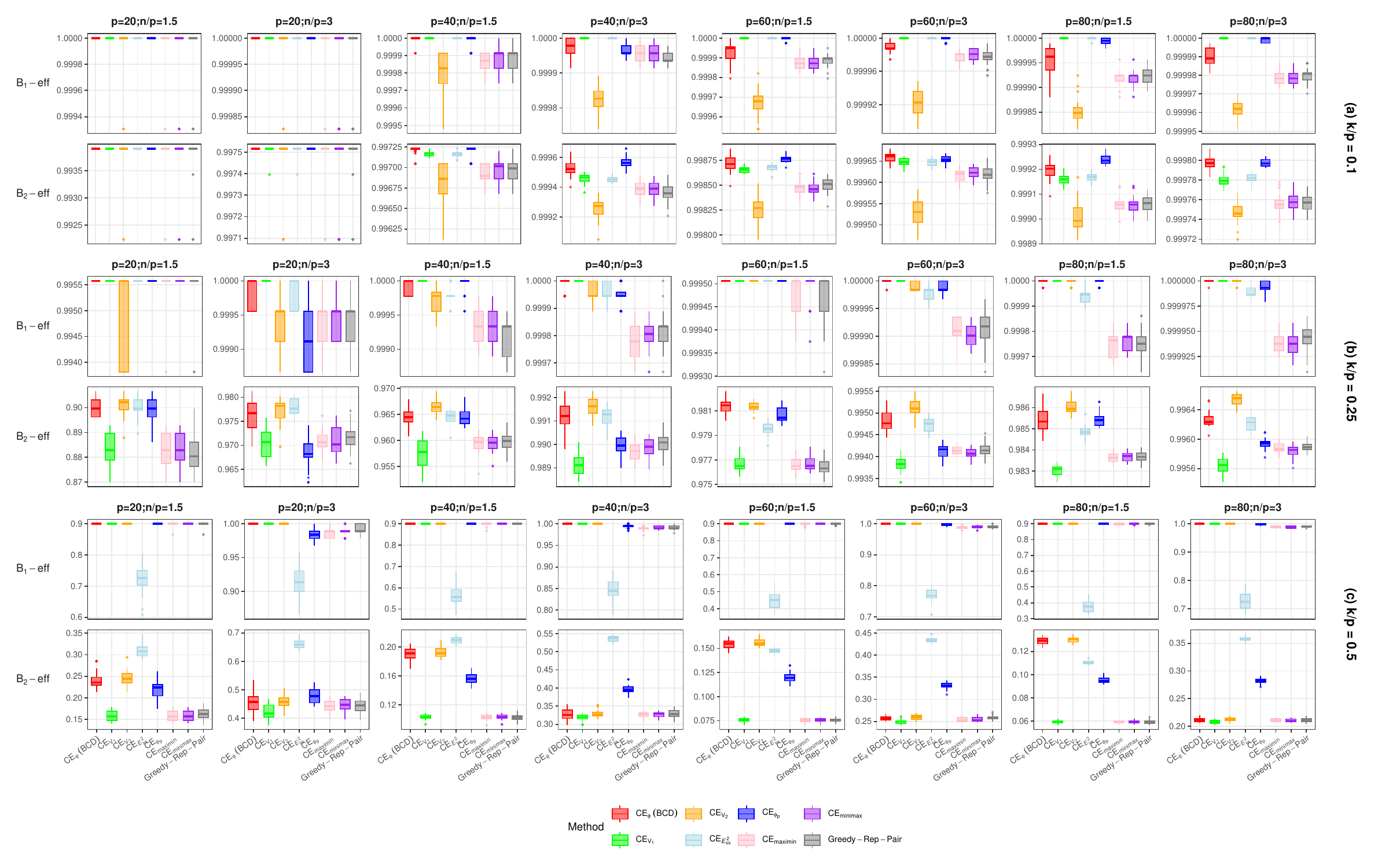}
\caption{Aliasing metrics under the main comparison: $B_1$-efficiency (replication balance) and $B_2$-efficiency (concurrence regularity) with $w_1=w_2=1$.}
\label{fig:sim_core_combined}
\end{sidewaysfigure}

\begin{sidewaysfigure}[htp]
\centering
\includegraphics[width=\linewidth]{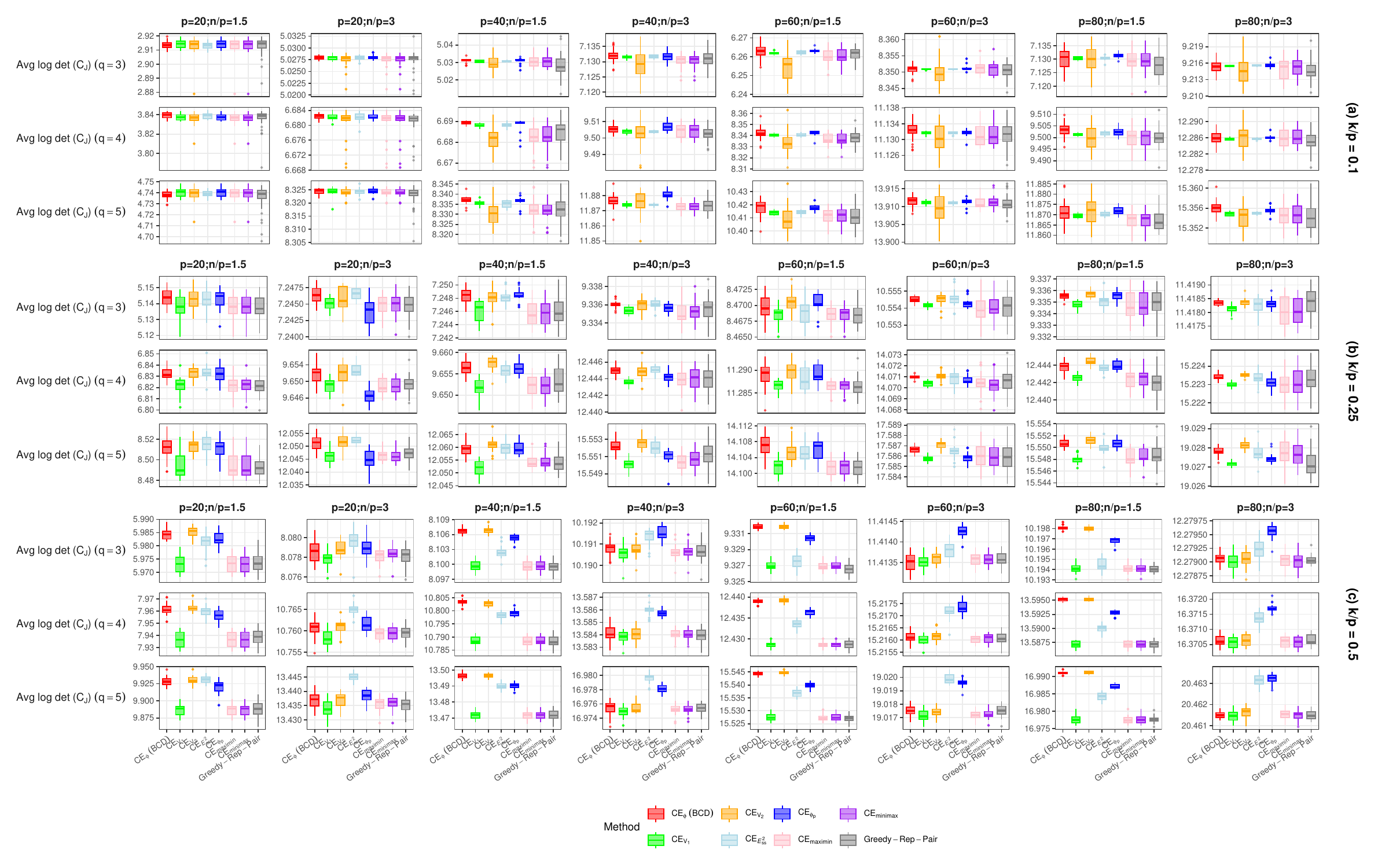}
\caption{Projection-based information quality (average of ${\log\det(\mathbf{C}_J)}$, $q\in\{3,4,5\}$) for all compared methods with $w_1=w_2=1$.}
\label{fig:proj_logdet_mean_box}
\end{sidewaysfigure}

\subsubsection{Weight sensitivity: oracle $q$-aware tuning}\label{subsubsec:tune_w1}

We next investigate whether tuning the weight $w_1$ in the composite criterion can improve downstream performance when oracle knowledge of the target projection order $q$ is available. Since the moderate-constraint regime with small budget ($k/p=0.25, n/p=1.5$) is the most practically representative setting considered here and also provides the clearest separation among tuned designs, we present its downstream F1, precision, recall, and mean-function error results in the main text. Results for the other regimes and budgets are reported in Appendix Section~\ref{app:add_num_res}. 

\begin{sidewaysfigure}[htp]
\centering
\includegraphics[width=\linewidth]{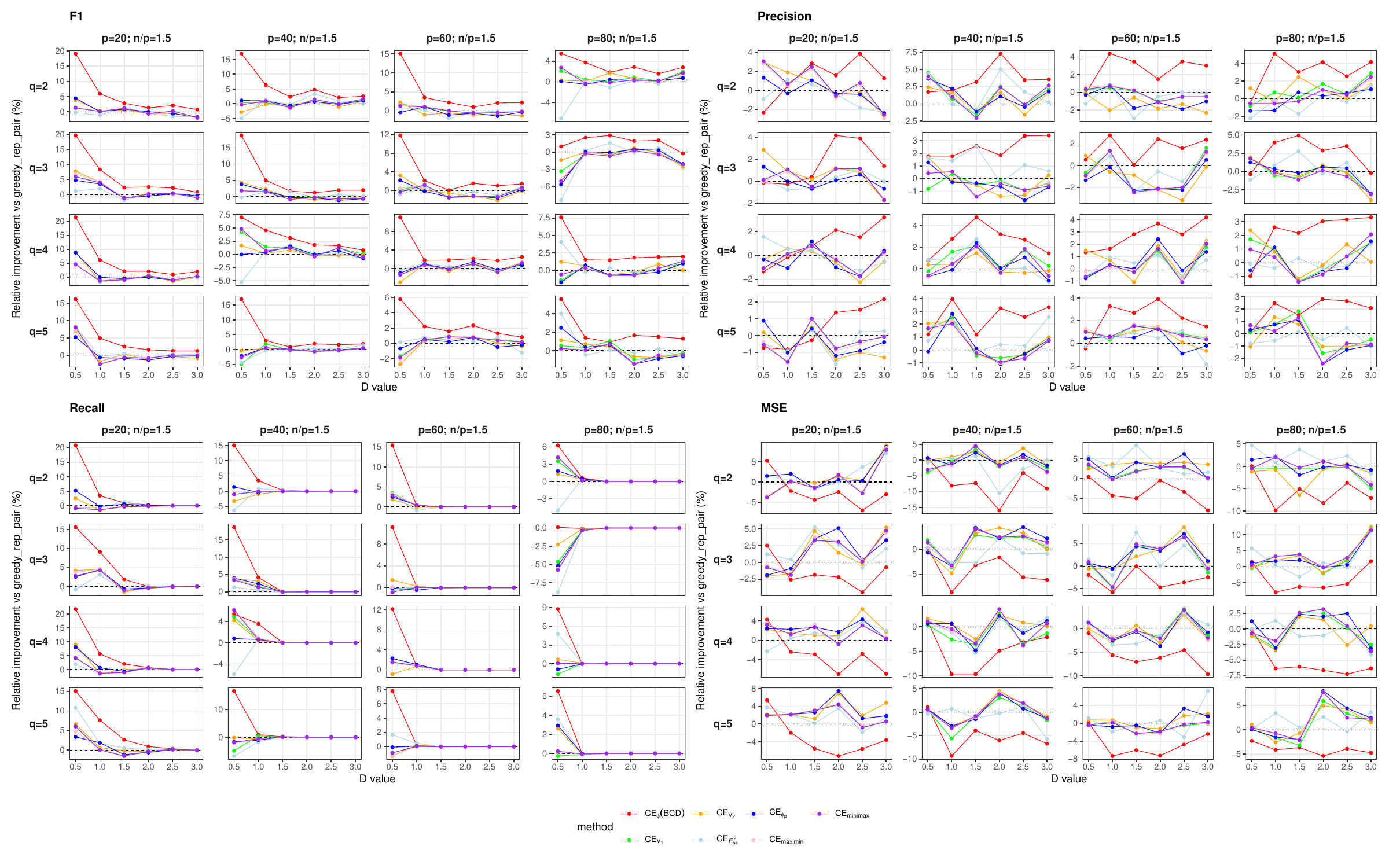}
\caption{Relative improvement over the greedy-rep-pair baseline in F1, precision, recall, and MSE across experimental settings ($k/p=0.25, n/p=1.5$)}
\label{fig:w1_lineplot}
\end{sidewaysfigure}

Using F1-score as a representative downstream objective, oracle \(q\)-aware tuning yields the clearest gains for CE-\(\Phi_{\mathrm{BCD}}\), with the improvement concentrated in the low-\(D\) region where design quality has the strongest effect on screening performance. These gains are most visible at smaller projection orders, and they generally attenuate as \(D\) increases, so that the separation among methods becomes less pronounced in the high-\(D\) regime. The precision and recall panels clarify the mechanism behind this pattern. When \(D\) is small, the F1 improvement is driven primarily by recall, indicating substantially better recovery of the active set. As \(D\) becomes larger, the recall advantage shrinks markedly, and the remaining F1 gain is explained more by modest but persistent improvements in precision. Thus, the source of the F1 advantage shifts from sensitivity in the low-\(D\) regime to selectivity in the high-\(D\) regime. The MSE results show a more compressed pattern, with smaller relative changes than those observed for F1, but CE-\(\Phi_{\mathrm{BCD}}\) still delivers mostly negative relative improvements across scenarios, indicating lower estimation error than the baseline. Taken together, these results show that oracle tuning of CE-\(\Phi_{\mathrm{BCD}}\) improves variable-selection quality without sacrificing predictive accuracy, whereas the remaining criteria stay much closer to zero across all four metrics and exhibit less systematic benefit.

Overall, these results show that the weight $w_1$ is a meaningful calibration knob in the composite criterion. When the target projection scale is known, tuning $w_1$ can lead to tangible downstream gains. We emphasize, however, that F1 is used here only as one concrete example of a downstream task. The same tuning strategy can be carried out with other objectives, such as prediction error or alternative screening metrics, depending on the scientific goal. In this sense, the composite criterion provides a flexible interface through which the design can be calibrated to the needs of the downstream analysis.

\section{Case Study: Prompt-Component Screening for a Large Language Model}
\label{sec:case_study}

Large language models are highly sensitive to how tasks are described to them, and a large body of work has developed natural-language prompting techniques that can substantially change a model's behavior without any parameter updates \citep{brown2020language}. Recent surveys catalog dozens of such techniques spanning reasoning scaffolds, role setting, self-verification, in-context examples, and output formatting \citep{sahoo2024systematic}. The marginal effect of any individual technique, however, is task- and model-dependent: chain-of-thought prompting has been shown to decrease accuracy on a range of tasks where explicit verbal reasoning interferes with the model's default behavior \citep{liu2024mindyourstep}, and adding a persona to the system prompt does not reliably improve performance on objective tasks \citep{zheng2024helpful}. Combinations of techniques are also difficult to reason about a priori: prompt components can interact non-additively, so that adding one technique may enhance, attenuate, or even cancel the effect of another \citep{khojah2025impact}. Choosing an effective prompt is therefore a screening problem: out of a larger set of candidates, we want to identify the small subset that actually improves performance on the target task, using as few evaluation runs as possible.

A cardinality constraint arises naturally in this setting. Context windows are finite, evaluation budgets are limited, and prompts loaded with too many simultaneous instructions tend to confuse the model rather than help it. We therefore restrict attention to prompts that activate exactly $k = 3$ components out of $p = 15$ candidates, yielding a small instruction budget motivated by recent evidence that LLMs’ ability to follow all instructions deteriorates as the number of simultaneous directives increases \citep{harada2025curse}. This section applies the TCARD framework to this concrete instance, with the goal of identifying which components help, which hurt, and which are inactive, on a total budget of 21 LLM evaluation runs. The case study is a worked example rather than a benchmark, but it gives us a real response surface on which to compare design-construction methods under exactly the conditions our simulations target.

We use GSM8K \citep{cobbe2021training} as the downstream task. GSM8K contains 8,500 grade-school math word problems written by human problem writers, split into 7,500 training problems and 1,000 test problems. Each problem takes between 2 and 8 steps to solve, and the solutions mainly use basic arithmetic. The task is well-matched to our goal because it is hard enough that prompt choice has a measurable effect on accuracy, but structured enough that the accuracy signal is stable across repeated evaluation.

\subsection{Experimental setup}
\label{sec:case_study_setup}

We screen $p = 15$ candidate prompt components, each either present (coded $1$) or absent (coded $0$) in a given prompt. The components are listed in Table~\ref{tab:components_list} and cover a range of common prompt-engineering techniques: role setting, reasoning strategy, in-context examples, self-verification, and output formatting. Each component is a short instruction that can be added to the prompt on its own, and their order in the prompt is fixed, so the only thing that changes between runs is which components are turned on.

\begin{table*}[h]
\centering
\caption{The 15 candidate prompt components used in the case study, listed with the exact instruction text appended to each prompt. }
\label{tab:components_list}
\small
\begin{tabular}{lp{10cm}p{4cm}}
\toprule
 Label & Instruction text & Reference \\
\midrule
 StepByStep    & ``Think through this step by step.'' & \citet{kojima2022large} \\
 ExpertRole    & ``You are an expert mathematician.'' & \citet{zheng2024helpful} \\
 OneExample    & One worked example prepended to the prompt (see below). & \citet{brown2020language} \\
 ReadCarefully & ``Read the question carefully and make sure you use exactly the numbers given in the problem.'' & \citet{deng2023rephrase} \\
 Verify        & ``After you finish, verify your answer is correct.'' & \citet{weng2023verify} \\
 Restate       & ``Begin by restating the problem in your own words.'' & \citet{yugeswardeenoo2024question} \\
 Units         & ``Pay attention to units and convert them if necessary.'' & \citet{park2022language} \\
 Estimate      & ``Before solving, estimate what a reasonable answer would be.'' & \citet{ma2024large} \\
 DrawTable     & ``Organize the given information in a table or list before solving.'' & \citet{wang2024chaintable} \\
 ShowArith     & ``Write out each arithmetic calculation explicitly; do not skip steps.'' & \citet{nye2021scratchpad} \\
 Algebra       & ``Define variables and set up equations before computing.'' & \citet{chen2022pot} \\
 SubProblems   & ``If the problem has multiple parts, solve each part separately.'' & \citet{zhou2023leasttomost} \\
 Concise       & ``Keep your reasoning concise and avoid unnecessary text.'' & \citet{xu2025chainofdraft} \\
 DoubleCheck   & ``Double-check each intermediate calculation before proceeding.'' & \citet{madaan2023selfrefine} \\
 Summary       & ``End with a one-sentence summary stating your final answer.'' & \citet{zhou2023thot} \\
\bottomrule
\end{tabular}

\vspace{-1ex}
\begin{quote}\small
\textit{Here is a worked example:} \\
\textit{Question:} A farmer has 3 fields. Each field has 12 rows of corn, and each row has 8 stalks. How many stalks of corn does the farmer have in total? \\
\textit{Solution:} Each field has $12 \times 8 = 96$ stalks. The farmer has 3 fields, so the total is $3 \times 96 = 288$ stalks. The final answer is 288.
\end{quote}
\end{table*}

Under the TCARD requirement $X \mathbf{1}_p = k \mathbf{1}_n$, every prompt uses exactly $k = 3$ components out of $p = 15$, giving $\binom{15}{3} = 455$ possible configurations. Our evaluation budget allows $n = 21$ runs per design, covering about $4.6\%$ of the configuration space. This places the case study in the moderate-constraint regime ($k/p = 0.2$, $n/p \approx 1.4$) studied by simulation in Section~\ref{sec:simulation}. The response model is the open-weight Llama~3.1~8B instruction-tuned model \citep{dubey2024llama}, queried at temperature $0$ through a local inference server for reproducibility. For each run $t$ of a design, we construct the prompt by activating the components indicated by row $\mathbf{X}_t$, evaluate the model on a fixed sample of $M = 200$ problems from the GSM8K training split, and record the percent correct as the response $\mathbf{y}_t \in [0, 100]$. The same 200-problem sample is used across all runs and all designs, so differences in $\mathbf{y}$reflect only the prompt configuration.

\subsection{Pilot calibration of the weight parameter}
\label{sec:case_study_pilot}

The $w_1$ tuning step in Algorithm~\ref{alg:tune_w1_phi_meanmean} needs a screening plan with plausible values for the effect size $D$ and the noise level $\sigma$. Rather than guess at these, we run a small one-at-a-time pilot before building any designs. The pilot has one baseline run (no components active) and 15 singleton runs (one component on at a time), each evaluated on $M_0 = 50$ problems. This costs $16 \times 50 = 800$ LLM calls, which is small compared to the main experiment, and it gives us a quick estimate of each component's marginal effect.

From the pilot experiment we obtain the absolute pilot effects $\widehat D$ across the 15 components. The actual number of active components in the fitted response is not known before the experiment is run, so we build three tuned TCARD designs with different $w_q$ at $q \in \{3, 4, 5\}$, covering the range of plausible active counts under the $k = 3$ constraint. For each $w_q$, we sweep $w_1$ over a logarithmic grid, build a TCARD design by minimizing $\Phi$ at each $w_1$, and pick the value $w_1^\star$ whose simulated screening $F_1$ is largest standardized across the $(q, \hat{D})$-grid. This gives three committed designs: CE-$\Phi_{\mathrm{BCD}}(w_{q=3})$, CE-$\Phi_{\mathrm{BCD}}(w_{q=4})$, and CE-$\Phi_{\mathrm{BCD}}(w_{q=5})$. Comparing their outputs on the same experiment acts as a sensitivity analysis across plausible priors for $q$, characterizing whether the framework's strong-effect identification is robust or stable in the region around $q=k$. As a benchmark for the contribution of tuning itself, we also build CE-$\Phi_{\mathrm{BCD}}$(default), which uses $w_1=w_2=1$ across $q$. The pilot uses no data from the main screening experiment, and all three tuned designs are finalized before any further LLM queries.

We then fit all competing TCARD methods from Section~\ref{sec:sim_methods} on the present instance $(n, p, k) = (21, 15, 3)$ and run the main experiment on each: the three tuned CE-$\Phi_{\mathrm{BCD}}$ variants, CE-$\Phi_{\mathrm{BCD}}$(default), CE-$V_1$, CE-$V_2$, CE-UE($s^2$), CE-$\Phi_p$, greedy-rep-pair, and Random Design. CE-$V_1$ gave degenerate screening output with $R^2 = 0$ and no components selected along the cross-validated regularization path, and hence it is left out of the comparison below.

\subsection{Modeling and analysis of each design's output}
\label{sec:case_study_screening}
 
\begin{table*}[t]
\centering
\caption{Design diagnostics and screening-quality summaries for the prompt-component case study ($n=21$, $p=15$, $k=3$). $\mathcal{V}_1$ and $\mathcal{V}_2$ are replication imbalance and pairwise concurrence dispersion (lower is better). $R^2_{\mathrm{refit}}$ is the $R^2$ of OLS refit on the Lasso-selected support; $R^2_{\mathrm{full}}$ is the $R^2$ of the full no-intercept OLS on all $p=15$ main effects. `SE ratio' is the largest-to-smallest standard-error ratio across components, computed over Lasso-selected components for the refit and over all $p=15$ components for the full OLS; the ideal value is 1. `---' indicates fewer than two components with positive SE (ratio undefined). `\#Sel' is the number of Lasso-selected components, `\#$p<0.05$' is the number of components significant in the full OLS at $p<0.05$, and `\#Overlap' is the number of components flagged by both.}
\label{tab:design_and_screening}
\begin{tabular}{lrrrrrrrrr}
\toprule
 & & & \multicolumn{3}{c}{Lasso + OLS refit} & \multicolumn{3}{c}{Full OLS} &  \\
\cmidrule(lr){4-6}\cmidrule(lr){7-9}
Design & $\mathcal{V}_1$ & $\mathcal{V}_2$ & $R^2_{\mathrm{refit}}$ & SE ratio & \#Sel &$R^2_{\mathrm{full}}$ & SE ratio & \#$p<0.05$ & \#Overlap \\
\midrule
CE-$\Phi_{\mathrm{BCD}}$($w_{q=3}$) & 0 & 34 & 0.91 & 1.15 & 8 & 0.95 & 1.26 & 3 & 3 \\
CE-$\Phi_{\mathrm{BCD}}$($w_{q=4}$) & 0 & 24 & 0.86 & 1.05 & 5 & 0.94 & 1.16 & 4 & 4 \\
CE-$\Phi_{\mathrm{BCD}}$($w_{q=5}$) & 0 & 46 & 0.61 & 1.16 & 4 & 0.84 & 1.37 & 0 & 0 \\
CE-$\Phi_{\mathrm{BCD}}$(default) & 0 & 28 & 0.18 & --- & 1 & 0.79 & 1.21 & 1 & 1 \\
CE-UE($s^2$) & 2 & 26 & 0.41 & 1.08 & 2 & 0.87 & 1.18 & 2 & 2 \\
CE-$V_2$ & 4 & 28 & 0.28 & --- & 1 & 0.74 & 1.33 & 1 & 1 \\
CE-$\Phi_p$ & 2 & 32 & 0.33 & 1.00 & 2 & 0.84 & 1.19 & 2 & 2 \\
Greedy-rep-pair & 2 & 54 & 0.66 & 1.06 & 4 & 0.79 & 7.63 & 0 & 0 \\
Random Design & 40 & 138 & 0.52 & 1.16 & 3 & 0.84 & 2.87 & 1 & 1 \\
\bottomrule
\end{tabular}
\end{table*}

\begin{figure}[htb]
\centering
\includegraphics[width=0.9\linewidth]{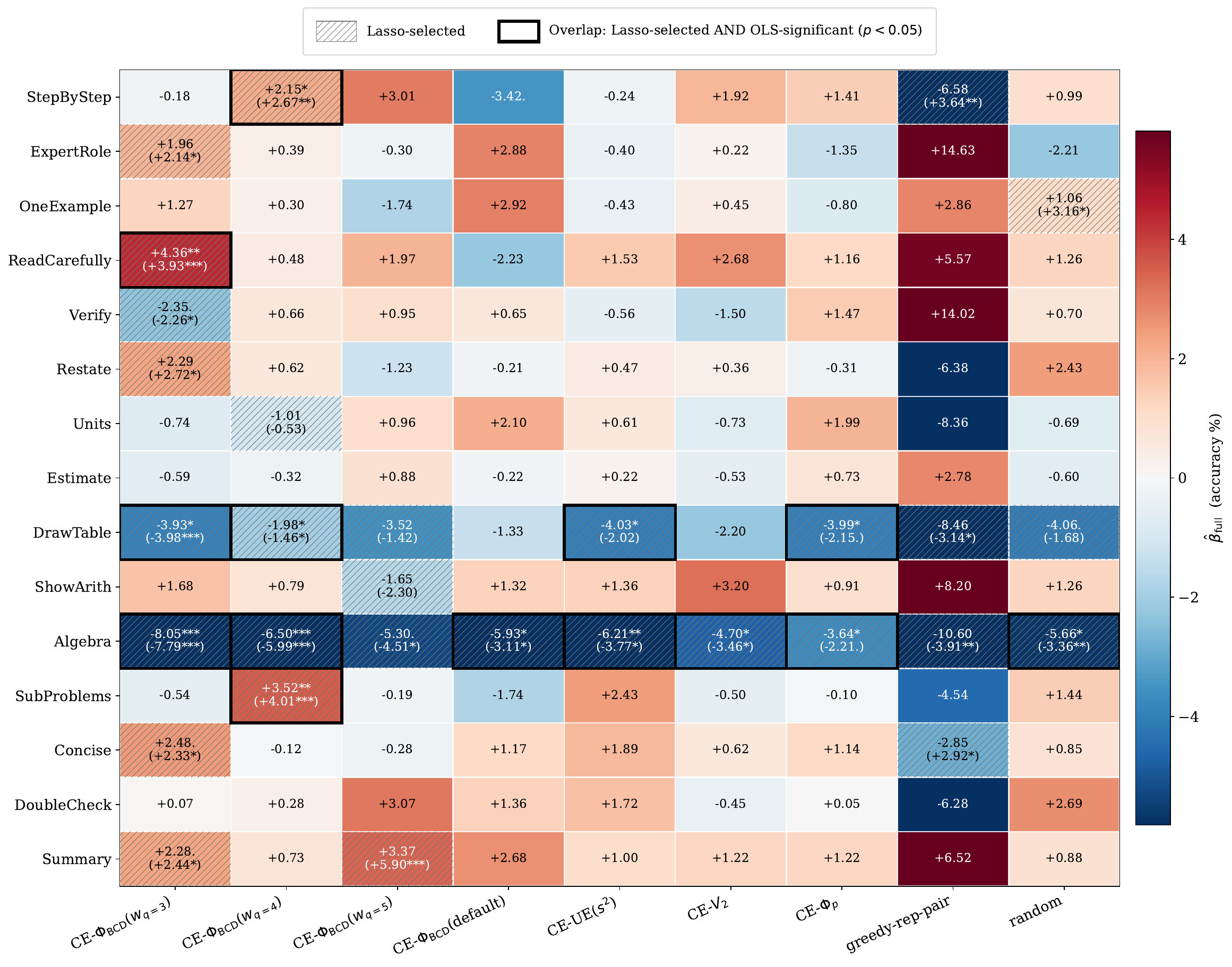}
\caption{Design-quality diagnostic comparing two models of the same response. Each cell shows the full no-intercept OLS coefficient $\hat \beta_{\mathrm{full}}$ on the first line (with full-OLS significance stars) and, where the component was selected by cross-validated Lasso, the OLS-refit coefficient $\hat \beta_{\mathrm{refit}}$ in parentheses on the second line (with refit significance stars). Cell background encodes $\hat \beta_{\mathrm{full}}$ on a diverging scale. Diagonal hatching marks Lasso-selected cells. A thick black border marks components that are both Lasso-selected and significant in the full OLS at $p<0.05$. Significance codes: ${}^{***}\,p<0.001$, ${}^{**}\,p<0.01$, ${}^{*}\,p<0.05$, ${}^{.}\,p<0.10$.}
\label{fig:overlap_heatmap}
\end{figure}

Table~\ref{tab:design_and_screening} reports replication imbalance $\mathcal{V}_1$ and pairwise concurrence dispersion $\mathcal{V}_2$, along with screening-quality summaries from two parallel analyses. 
Both use the centered response $\mathbf{y}_c = \mathbf{y} - \bar y \mathbf{1}_n$ and the uncentered design matrix $\mathbf{X}$, and both are fit without an intercept column: under the TCARD constraint, the all-ones vector lies in the column space of $\mathbf{X}$, so $\mathbf{y}_c = \mathbf{X}\boldsymbol{\beta} + \varepsilon$ is identified, and the resulting coefficients sum to zero under balanced replication. \emph{Model~A} is a 5-fold cross-validated Lasso (with standardized columns and no intercept) followed by an OLS refit on the Lasso-selected support. \emph{Model~B} is the no-intercept OLS on all $p = 15$ components, with $n - p = 10$ residual degrees of freedom, which gives classical $t$-tests for every component and does not suffer from selection bias. We use Model~A for parsimony and for downstream validation, and Model~B as a robust check on significance. The `\#Overlap' column in Table~\ref{tab:design_and_screening} counts the components flagged by both analyses.

A good screening design has to satisfy three requirements at once. \emph{Identifiability} requires that every component be tested approximately the same number of times, so that no effect is systematically under-sampled, which is what $\mathcal{V}_1$ measures. \emph{Separability} requires that every pair of components co-occur in a similar number of runs, so that their effects can be estimated free of collinearity, which is what $\mathcal{V}_2$ measures. \emph{Uniform precision} requires that all coefficient estimates carry comparable standard errors, so that inference on a small effect is not systematically harder than inference on a large one, which is what the SE ratio measures. These three requirements are not independent, since a design that optimizes one at the cost of another will fail downstream on either fit quality (low $R^2$), selection quality (few cross-verified components), or both.

The three tuned CE-$\Phi_{\mathrm{BCD}}$ variants satisfy all three requirements jointly. Each achieves perfect replication balance ($\mathcal{V}_1 = 0$), concurrence dispersion in the range $\mathcal{V}_2 \in [24, 46]$, full-OLS SE ratios between $1.16$ and $1.37$ (within a factor of $1.4$ of the theoretical ideal of $1$), and refit SE ratios as low as $1.05$. These structural properties translate directly into screening performance: the three variants reach $R^2_{\mathrm{refit}} = 0.61$--$0.91$ and $R^2_{\mathrm{full}} = 0.84$--$0.95$, and each recovers between three and four components that are flagged by both analyses. All three independently identify Algebra and DrawTable as strong negative effects. This is an empirically interesting finding given that algebraic formulation \citep{chen2022pot} and tabular reasoning \citep{wang2024chaintable} are both documented as helpful prompting strategies in the literature, but consistent with recent evidence that structured reasoning techniques can reduce performance on tasks where the additional structure interferes with the model's default behavior \citep{liu2024mindyourstep}. On the positive side, CE-$\Phi_{\mathrm{BCD}}$($w_{q=3}$) and CE-$\Phi_{\mathrm{BCD}}$($w_{q=4}$) jointly flag SubProblems and StepByStep as helpful candidates, both of which are supported as effective prompting strategies for math reasoning \citep{zhou2023leasttomost, kojima2022large}. 

The failure modes of the competing methods each show what happens when one of the three requirements is missed. Random is disqualified on identifiability and separability ($\mathcal{V}_1 = 40$, $\mathcal{V}_2 = 138$): its full-OLS SE ratio of $2.87$ prevents meaningful inference, and only Algebra survives both analyses. Greedy-rep-pair corrects replication balance to $\mathcal{V}_1 = 2$ but ignores pairwise structure ($\mathcal{V}_2 = 54$), and the consequence is an SE ratio of $7.63$: its least-precisely-estimated coefficient has a standard error over seven times larger than its most-precisely-estimated one, and zero components reach significance in the full OLS despite four Lasso selections. Both baselines illustrate the same failure pattern where a Lasso fit that does not survive unregularized testing, leaving the practitioner with a screening result they cannot defend.

The more substantive comparison is against CE-$V_2$ and CE-UE($s^2$). Both attain reasonable replication balance and lie in the low-$\mathcal{V}_2$ cluster, yet they recover far less structure than the proposed method. The CE-$V_2$ flags only Algebra with $R^2_{\mathrm{refit}} = 0.28$ and one cross-verified component. The CE-UE($s^2$) flags Algebra and DrawTable with $R^2_{\mathrm{refit}} = 0.41$ and two cross-verified components. The three tuned CE-$\Phi_{\mathrm{BCD}}$ variants all substantially outperform both on every fit metric while recovering two to four times as many cross-verified components. This is because CE-$V_2$ controls concurrence dispersion in isolation, CE-UE($s^2$) controls a single-matrix summary that does not separately penalize replication imbalance, and neither captures enough of the projected information matrix to produce a reliable screening result. The proposed criterion controls both components jointly, and the empirical advantage follows.

Internal to the CE-$\Phi_{\mathrm{BCD}}$ family, CE-$\Phi_{\mathrm{BCD}}$(default) provides an ablation of the tuning step. The default variant is structurally balanced by any conventional standard, yet produces $R^2_{\mathrm{refit}} = 0.18$ and only one Lasso-selected component. Structural balance without pilot-informed weight tuning is not enough: the weight configuration is what aligns the design's projected information matrix with the scale of effects the downstream pipeline is looking for.

\subsection{Test-split validation}
\label{sec:case_study_validation}


\begin{figure}[h]
\centering
\includegraphics[width=\linewidth]{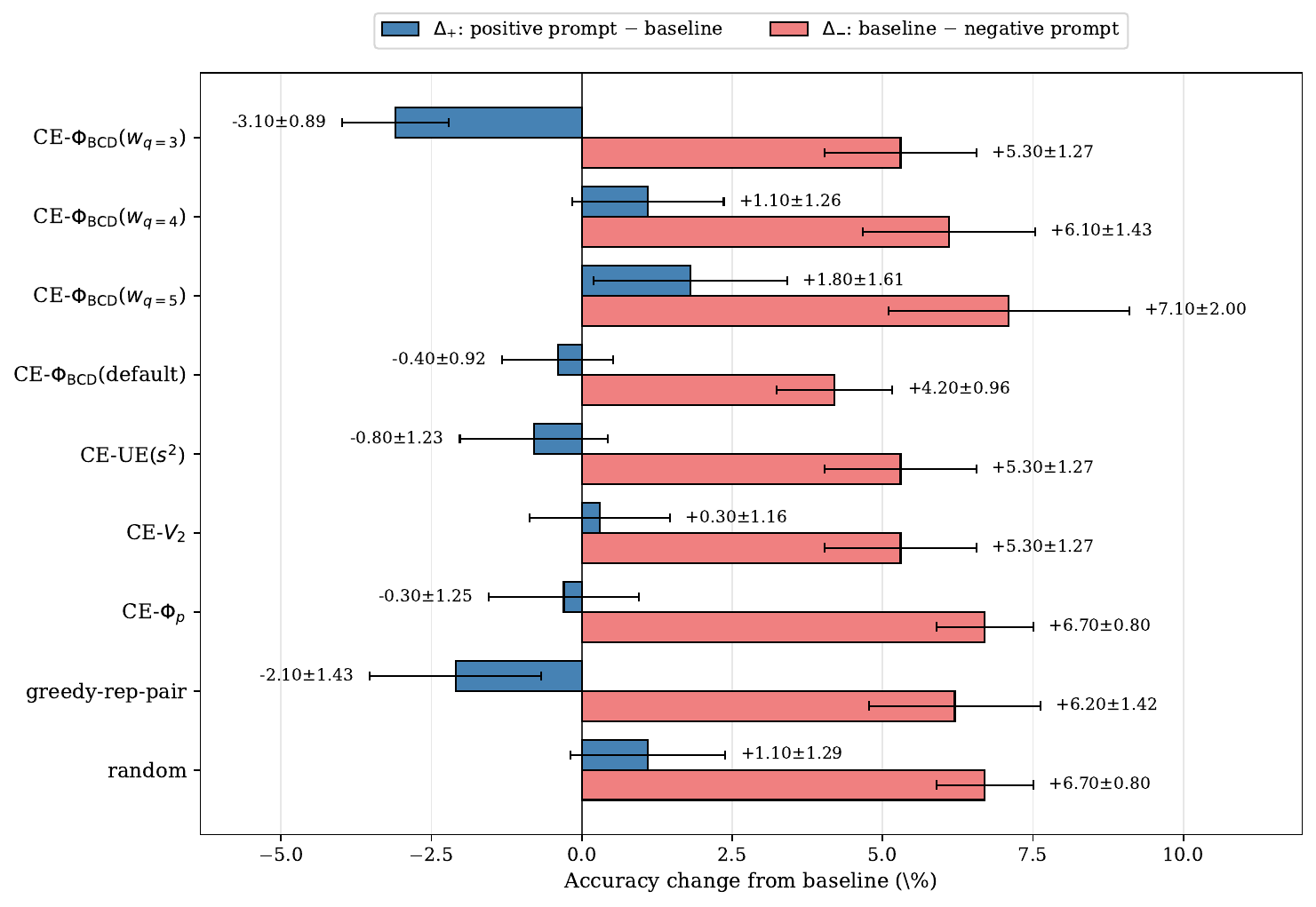}
\caption{Movement from baseline accuracy for each design, averaged across $S = 10$ validation seeds (mean $\pm$ SE). $\Delta_{+}$ (blue) is how much the positive prompt raises accuracy above baseline; $\Delta_{-}$ (red) is how much the negative prompt lowers it below baseline. Both positive $\Rightarrow$ correct screening in both directions.}
\label{fig:validation_bargap}
\end{figure}

A remaining concern is that the fit statistics in Table~\ref{tab:design_and_screening} are computed on the same response vector used to fit the screening models, and may reward designs that produce ungeneralizable effects. We address this with a train-test validation protocol on a disjoint set of GSM8K problems sampled from the test split. For each design, we construct two prompt configurations of cardinality $k=3$ from its full-OLS coefficients: the \emph{positive} configuration activates the three components with the largest positive $\hat\beta_{\mathrm{full}}$, and the \emph{negative} configuration activates the three with the most negative $\hat\beta_{\mathrm{full}}$. Both configurations, along with the same baseline (no components active, mean accuracy $\approx 83\%$), are evaluated on $M_{\mathrm{val}} = 100$ GSM8K problems with the same model and decoding parameters as the screening experiment. We repeat the evaluation across $S = 10$ independent seeds. The validation experiment tests two predictions. A design whose positive coefficients track truly helpful components should produce $\Delta_{+} = (\text{positive prompt accuracy}) - (\text{baseline accuracy}) > 0$. A design whose negative coefficients track truly harmful components should produce $\Delta_{-} = (\text{baseline accuracy}) - (\text{negative prompt accuracy}) > 0$. 

Figure~\ref{fig:validation_bargap} shows a clear asymmetry between the two sides of the validation test. The negative side is easy: every design except CE-$\Phi_{\mathrm{BCD}}$($w_{q=3}$) gets it essentially right. Eight of the nine designs achieve $\Delta_{-}$ between $+4.2\%$ and $+7.1\%$, confirming that their negative-side predictions generalize. This is the consequence of a shared structural fact, namely that Algebra and DrawTable (identified as harmful in Section~\ref{sec:case_study_screening}) produce large, robust negative effects, and nearly every design in the comparison places both of them in its negative top-3 by coefficient ranking. Even CE-$\Phi_{\mathrm{BCD}}$(default), which recovered only Algebra under Lasso selection, assigns DrawTable to its negative top-3 by $\hat\beta_{\mathrm{full}}$ and benefits accordingly ($\Delta_{-} = +4.2\%$). On this instance, identifying harmful components is not what distinguishes a good screening design from a bad one.

However, the positive side is hard, and only the tuned CE-$\Phi_{\mathrm{BCD}}$ variants with well-matched $q$ succeed. Only three designs achieve meaningfully positive $\Delta_{+}$: $w_{q=5}$ at $+1.80\% \pm 1.61\%$, $w_{q=4}$ at $+1.10\% \pm 1.26\%$, and Random Design at $+1.10\% \pm 1.29\%$. Every other design sits at or below zero, where the three components identified as most beneficial actively do not improve or even hurt test-split accuracy. This asymmetry between large negative effects and small positive ones is consistent with a near-ceiling regime: Llama 3.1 8B already reaches $\approx 83\%$ on GSM8K under a null prompt \citep{dubey2024llama}, leaving limited headroom for individual prompting techniques to produce large positive gains, while poorly-matched techniques can still cause substantial drops \citep{cheng2025revisiting}. The statistical evidence required to identify small positive effects on a 21-run budget is therefore stronger than the evidence required to identify large negative ones, and most designs in the comparison fail to meet that bar. Among the designs that do, CE-$\Phi_{\mathrm{BCD}}$($w_{q=4}$) and CE-$\Phi_{\mathrm{BCD}}$($w_{q=5}$) are the only ones backed by screening-stage evidence. The proposed criterion therefore holds a narrow but defensible advantage on the harder side of the validation test.

Moreover, Random Design's $\Delta_{+} = +1.10\%$ matches $w_{q=4}$, and its $\Delta_{-} = +6.7\%$ is among the largest in the figure. Against the screening evidence of Table~\ref{tab:design_and_screening}, this performance is not reproducible. Random Design's Lasso fit recovered only a single significant component, and its full-OLS SE ratio of $2.87$ indicates highly non-uniform precision. The reason the validation performance looks competitive is that Random Design's coefficient ranking happens to surface Algebra and DrawTable on the negative side by chance rather than through principled significance testing; the practitioner has no basis to trust the ranking before seeing the validation result. The validation experiment and the screening-stage evidence must be read together, not separately: a design can happen to rank components correctly on testing data without producing a screening fit the practitioner can act on. The remaining designs fail on the positive side in ways consistent with their screening-stage diagnostics. CE-$\Phi_{\mathrm{BCD}}$(default) produces $\Delta_{+} = -0.40\%$ and only a single cross-verified component at the screening stage ($R^2_{\mathrm{refit}} = 0.18$). Greedy-rep-pair produces $\Delta_{+} = -2.10\%$, consistent with the full-OLS SE ratio of $7.63$ that prevented any component from reaching significance. CE-$\Phi_p$ and CE-UE($s^2$) produce near-zero $\Delta_{+}$, reflecting the limited signal their screening stage extracted.

Overall, the validation results support a useful conclusion. Designs that control both components of the projected information matrix and use weight tuning are the only designs in the comparison that produce both a correct negative-side prediction and a correct positive-side prediction on testing data. Designs that fail structurally or that skip the tuning step can recover the easier negative-side signal but fail on the harder positive-side prediction, meaning the practitioner would commit to prompt configurations that do not help as intended.

\section{Discussion}\label{sec:discussion}

This paper has developed a unified algebraic, algorithmic, and statistical framework for two-level treatment-cardinality-constrained designs, organized around two combinatorial summaries intrinsic to the TCARD structure: factor replication balance and pairwise concurrence balance. On the algebraic side, we linked TCARDs to BIBD-type balance, formalized nearly balanced TCARDs through replication and concurrence regularity, and established existence in important boundary regimes. On the algorithmic side, we introduced the counts-only, model-free criterion $\Phi_{\mathrm{BCD}}$, which operates directly on these two summaries while admitting efficient incremental evaluation during coordinate exchange. Under the main-effects model, $\Phi_{\mathrm{BCD}}$ is sufficient for the $M$-stage of the $(M,S)$ principle, admits an exact algebraic connection to centered $\mathrm{UE}(s^2)$, and relates to Bayesian $D$-optimality through a counts-based perturbation analysis.
The prompt-component screening case study on GSM8K with Llama~3.1~8B illustrates these properties in a real-world setting: the CE-$\Phi_{\mathrm{BCD}}$ design recovers twice as many significant main
effects as the nearest competing method, produces the largest positive-minus-negative validation gap on held-out problems, and exposes two clean failure modes of criteria that control only one of
the two counts summaries.

The framework also clarifies how the weight parameter $w_1$ should be interpreted. Rather than a fixed universal constant, $w_1$ is more appropriately viewed as a task-dependent calibration parameter whose tuning can improve downstream performance when the target projection
structure is known. In the more common case where no information about the downstream analysis is available, we recommend the surrogate choice $q^\star = k$, which targets the projection order matched to the cardinality constraint itself. This choice is low-risk: Appendix~\ref{app:qeqk_tuning} shows that tuning gains at $q = k$ do not transfer to other projection orders, but that the tuned design does not systematically deteriorate at mismatched $q$ either. The prompt-component case study reinforces this recommendation. The untuned default with $w_1 = w_2 = 1$
remains a reasonable fallback when even a pilot-scale calibration is infeasible, but $q^\star = k$ is the operationally preferable default for most settings.

Our framework also helps clarify several limitations and follow-up questions raised in the recent CRowS method proposed by \cite{smucker2025large}. In particular, that line of work highlights the need for sharper lower bounds, for reduced reliance on the assumption of tight row constraints, and for a more rigorous characterization of when upper-bound row constraints become active at the optimum. Our contribution is complementary: rather than asking when a more general upper-bound row-constrained formulation collapses to the tight case, we study the exact-cardinality regime itself, thereby isolating the tight-row subclass and providing a cleaner algebraic, spectral, and information-theoretic analysis within that setting. Likewise, whereas CRowS fixes a design criterion and then evaluates downstream screening performance, our framework allows the criterion itself to be calibrated to a downstream objective through the choice of $w_1$. This does not by itself solve the false-positive or interaction problems emphasized in that discussion, but it provides a natural route toward more task-specific control of screening behavior. 

These results also suggest several promising directions for future work. A first extension is to move beyond two-level exact-cardinality designs and develop analogous theory and construction methods for continuous-factor settings, where sparse projection quality must be balanced against geometric space-filling considerations in a fundamentally different way. A second direction is to extend the present framework to mixed-factor designs that combine discrete and continuous variables, or binary and multi-level factors, which arise naturally in many practical experiments but require new balance and concurrence summaries beyond the current two-level TCARD setting. A third direction is to study grouped-factor or structured-factor regimes, in which factors are organized into pre-specified blocks, functional modules, or hierarchical groups and the design must respect both cardinality constraints and group-level selection structure. Such extensions are especially relevant for modern screening problems in which sparsity is meaningful not only at the individual-factor level but also at the group or subsystem level. More broadly, the present framework opens several additional avenues for methodological development. These include extending the criterion to richer downstream objectives, incorporating adaptive or sequential calibration of the weight parameter, developing faster search algorithms for larger and more complex design spaces, and studying how the counts-based perspective can be generalized to other constrained design problems beyond the exact-cardinality case. Taken together, these directions suggest that the current work should be viewed not as an endpoint, but as a foundation for a broader class of sparse, structured, and task-aware design methodologies.

Overall, the main message of this paper is that treatment-cardinality-constrained design quality can be understood through a small number of interpretable combinatorial summaries. The proposed criterion $\Phi_{\mathrm{BCD}}$ turns these summaries into a practical design objective while remaining closely tied to classical information-based principles. As a result, the framework provides both a strong model-free default for routine design construction and a flexible interface for task-aware calibration when downstream objectives are known.

\bibliographystyle{apalike}
\bibliography{sample-base}
\clearpage

\newpage
\setcounter{footnote}{0}

	\begin{center}
		{\Large \bf Supplementary Materials for \\
        ``TCARD: Nearly Balanced Two-Level Designs with Treatment Cardinality Constraints with an Application to LLM Prompt Engineering"}\\
		\vskip 10pt
	
	\end{center}


\setcounter{page}{1}
\setcounter{section}{1}
\setcounter{equation}{0}
\renewcommand{\theequation}{A.\arabic{equation}}
\setcounter{figure}{0}
\renewcommand{\thefigure}{A.\arabic{figure}}    

\appendix
\renewcommand{\thesection}{A.\arabic{section}}



\section{Proof of Theorem~\ref{thm:NB_TCARD_exist}}\label{app:thm1-proof}

\begin{proof}
Throughout, $\mathbf{X}\in\{0,1\}^{n\times p}$ is a $k$-TCARD design matrix (each row has exactly $k$ ones) and
\[
\mathbf{C} := \mathbf{X}^\top \mathbf{H} \mathbf{X},\qquad \mathbf{H}:=\mathbf{I}_n-\tfrac{1}{n}\mathbf{1}_n\mathbf{1}_n^\top
\]
is the intercept-eliminated main-effects information matrix. Write the replication vector $\mathbf{r}:=\mathbf{X}^\top\mathbf{1}_n=(r_1,\ldots,r_p)^\top$ and note that
\begin{equation}\label{eq:C_expand}
\mathbf{C} = \mathbf{X}^\top \mathbf{X} - \tfrac{1}{n}\,\mathbf{r}\mathbf{r}^\top.
\end{equation}
In particular, for $i\neq j$,
\begin{equation}\label{eq:C_offdiag}
\mathbf{C}_{ij} = (\mathbf{X}^\top \mathbf{X})_{ij} - \tfrac{1}{n}\,r_ir_j = \lambda_{ij} - \tfrac{1}{n}\,r_ir_j,
\end{equation}
and for the diagonal,
\begin{equation}\label{eq:C_diag}
\mathbf{C}_{ii} = r_i - \tfrac{1}{n}\,r_i^2.
\end{equation}

Since $\mathbf{X}\mathbf{1}_p = k\mathbf{1}_n$ and $\mathbf{H}\mathbf{1}_n=\mathbf{0}$, we have
\begin{equation}\label{eq:C1zero}
\mathbf{C}\mathbf{1}_p = \mathbf{X}^\top \mathbf{H}\mathbf{X}\mathbf{1}_p = \mathbf{X}^\top \mathbf{H}(k\mathbf{1}_n) = \mathbf{0},
\end{equation}
so every row sum of $\mathbf{C}$ is zero. Combining~\eqref{eq:C1zero} with~\eqref{eq:C_offdiag}--\eqref{eq:C_diag} yields, for each~$i$,
\[
0 = \sum_{j=1}^{p}\mathbf{C}_{ij}
= \Bigl(r_i - \frac{r_i^2}{n}\Bigr) + \sum_{j\neq i}\Bigl(\lambda_{ij}-\frac{r_ir_j}{n}\Bigr)
\;\;\Longrightarrow\;\;
\sum_{j\neq i}\lambda_{ij} = r_i(k-1),
\]
where we used $\sum_{j\neq i}r_j = nk - r_i$. This recovers~\eqref{eq:row_lambda} directly from the information matrix.

Without loss of generality, permute columns so that the near-balanced replication vector takes the form
\[
r_1 = \cdots = r_s = f, \qquad r_{s+1} = \cdots = r_p = f+1,
\]
with $f=\lfloor nk/p\rfloor$ and $s = p-(nk-pf)$, and set $h:=p-s$. Let $\mathbf{I}_s,\mathbf{J}_s$ denote the $s\times s$ identity and all-ones matrices, and similarly $\mathbf{I}_h,\mathbf{J}_h$. Write $\mathbf{J}_{s,h}:=\mathbf{1}_s\mathbf{1}_h^\top$ and $\mathbf{J}_{h,s}:=\mathbf{J}_{s,h}^\top$.

Define the two constants
\[
d_L := f - \frac{f^2}{n}, \qquad d_H := (f+1) - \frac{(f+1)^2}{n},
\]
and the off-diagonal centering-correction block matrix
\begin{equation}\label{eq:Q_def}
\mathbf{Q} :=
\begin{pmatrix}
\frac{f^2}{n}(\mathbf{J}_s - \mathbf{I}_s) & \frac{f(f+1)}{n}\,\mathbf{J}_{s,h}\\[4pt]
\frac{f(f+1)}{n}\,\mathbf{J}_{h,s} & \frac{(f+1)^2}{n}(\mathbf{J}_h - \mathbf{I}_h)
\end{pmatrix},
\end{equation}
as well as the diagonal matrix
\begin{equation}\label{eq:P_def}
\mathbf{P} :=
\begin{pmatrix}
d_L\,\mathbf{I}_s & 0\\
0 & d_H\,\mathbf{I}_h
\end{pmatrix}.
\end{equation}
Note that $\mathbf{Q}$ has zero diagonal, while $\mathbf{P}$ matches $\mathrm{diag}(\mathbf{C})$ by~\eqref{eq:C_diag}.

Recall $\kappa := \lfloor(k-1)f/(p-1)\rfloor$ and
\[
\omega := p - 1 - \bigl((k-1)f - \kappa(p-1)\bigr) \in \{0,1,\ldots,p-1\},
\]
so that
\begin{equation}\label{eq:remainder_identity}
(k-1)f = \kappa(p-1) + (p-1-\omega).
\end{equation}

\paragraph{(A) Type~I ($\omega\ge k-1$).}
Assume a nearly balanced TCARD of Type~I exists. Under NB2 and $\omega\ge k-1$, every pair concurrence satisfies
\[
\lambda_{ij}\in\{\kappa,\;\kappa+1\} \qquad (i\neq j),
\]
for both $r_i=f$ and $r_i=f+1$. 

For factors with $r_i = f$, by~NB2 and the definition of $\kappa$, we have
$\lambda_{ij} \in \{\lfloor f(k-1)/(p-1)\rfloor,\,
                    \lceil f(k-1)/(p-1)\rceil\}
             = \{\kappa,\kappa+1\}$.
For factors with $r_i = f+1$, write
\[
  (f+1)(k-1) = f(k-1) + (k-1)
              = \kappa(p-1) + (p-1-\omega) + (k-1).
\]
Dividing by $p-1$,
\[
  \frac{(f+1)(k-1)}{p-1}
  = \kappa + \frac{(p-1-\omega)+(k-1)}{p-1}
  = \kappa + 1 + \frac{k-1-\omega}{p-1}.
\]
Since $\omega \geq k-1$ we have $k-1-\omega \leq 0$, so
$\frac{(f+1)(k-1)}{p-1} \in \bigl(\kappa,\,\kappa+1\bigr]$,
giving $\lfloor (f+1)(k-1)/(p-1)\rfloor = \kappa$ and
$\lceil  (f+1)(k-1)/(p-1)\rceil = \kappa+1$.
Moreover, the fractional part $\frac{k-1-\omega}{p-1} \geq -(p-1)/(p-1) = -1$,
so the floor is exactly $\kappa$ (not $\kappa-1$) and the ceiling is
$\kappa+1$ (not $\kappa+2$).
Hence $\lambda_{ij}\in\{\kappa,\kappa+1\}$ for every pair, as claimed.

Define the $p\times p$ matrix
\begin{equation}\label{eq:A_typeI_def}
\mathbf{A} := \mathbf{Q} - \mathbf{P} - \kappa(\mathbf{J}_p - \mathbf{I}_p).
\end{equation}
Set
\begin{equation}\label{eq:G_typeI_def}
\mathbf{G} := \mathbf{C} + \mathbf{A}.
\end{equation}
We claim that $\mathbf{G}$ is the adjacency matrix of a simple graph on $\{1,\ldots,p\}$. Indeed, by construction $\mathbf{G}_{ii} = \mathbf{C}_{ii} + \mathbf{A}_{ii}$. Since $\mathbf{Q}$ has zero diagonal and $(\mathbf{J}_p-\mathbf{I}_p)$ has zero diagonal, $\mathbf{A}_{ii} = -\mathbf{P}_{ii}$, hence $\mathbf{G}_{ii} = \mathbf{C}_{ii} - \mathbf{P}_{ii} = 0$. For $i\neq j$, since $\mathbf{P}_{ij}=0$ and $(\mathbf{J}_p-\mathbf{I}_p)_{ij}=1$,
\[
\mathbf{G}_{ij} = \mathbf{C}_{ij} + \mathbf{Q}_{ij} - \kappa
= \Bigl(\lambda_{ij} - \frac{r_ir_j}{n}\Bigr) + \frac{r_ir_j}{n} - \kappa
= \lambda_{ij} - \kappa \in \{0,1\}.
\]
Thus $\mathbf{G}$ is a $0/1$ matrix with zero diagonal, hence an adjacency matrix.

Since $\mathbf{G}$ has zero diagonal, $\deg(i) = \sum_{j=1}^{p}\mathbf{G}_{ij}$. Using~\eqref{eq:C1zero}, $\sum_j \mathbf{C}_{ij}=0$, so $\deg(i) = \sum_{j=1}^{p}\mathbf{A}_{ij}$. A direct calculation using~\eqref{eq:A_typeI_def}--\eqref{eq:P_def} gives
\[
\sum_{j=1}^{p}\mathbf{A}_{ij}
= \underbrace{\sum_{j\neq i}\frac{r_ir_j}{n}}_{=\,r_i(nk-r_i)/n}
\;-\;\Bigl(r_i - \frac{r_i^2}{n}\Bigr) - \kappa(p-1)
= r_i(k-1) - \kappa(p-1).
\]
By~\eqref{eq:remainder_identity},
\[
\deg(i) =
\begin{cases}
f(k-1) - \kappa(p-1) = p-1-\omega, & i\le s,\\[3pt]
(f+1)(k-1) - \kappa(p-1) = p-1-\omega+(k-1), & i > s.
\end{cases}
\]
Equivalently, the degree sequence is
\begin{equation}\label{eq:deg_sequence_typeI}
\underbrace{(p-\omega-1,\;\ldots,\;p-\omega-1)}_{s\;\text{times}},\quad
\underbrace{(p-\omega+k-2,\;\ldots,\;p-\omega+k-2)}_{(p-s)\;\text{times}}.
\end{equation}
Since $\mathbf{G}$ must be a simple graph, this degree sequence must be graphical (Erd\H{o}s--Gallai).

Consider the complement graph $\overline{\mathbf{G}}$, whose degrees are
\[
\overline{d}(i) = (p-1) - \deg(i) =
\begin{cases}
\omega, & i\le s,\\
\omega-k+1, & i > s.
\end{cases}
\]
Ordering $\overline{d}$ non-increasingly and applying the Erd\H{o}s--Gallai condition with $t=s$ gives
\[
s\omega \;\le\; s(s-1) + \sum_{i=s+1}^{p}\min\{\omega-k+1,\;s\}.
\]
If $\omega-k+1\le s$ the right-hand side becomes $s(s-1)+(p-s)(\omega-k+1)$, which rearranges to $s(\omega-s+1)\le(p-s)(\omega-k+1)$. If $\omega-k+1>s$ the inequality reduces to $\omega\le p-1$, which holds by definition. This proves Part~(A).

\paragraph{(B) Type~II ($\omega<k-1$).}
Assume a nearly balanced TCARD of Type~II exists. Under NB2 and $\omega<k-1$, factors with $r_i=f$ have concurrences in $\{\kappa,\kappa+1\}$, while factors with $r_i=f+1$ have concurrences in $\{\kappa+1,\kappa+2\}$, since $(f+1)(k-1)=(\kappa+1)(p-1)+(k-1-\omega)$ with $k-1-\omega>0$. For any cross-class pair $(i\le s,\;j>s)$, the concurrence must lie in the intersection $\{\kappa,\kappa+1\}\cap\{\kappa+1,\kappa+2\}=\{\kappa+1\}$, so $\lambda_{ij}=\kappa+1$.

Define the $p\times p$ block matrix
\begin{equation}\label{eq:B_typeII_def}
\mathbf{B} :=
\begin{pmatrix}
\bigl(\frac{f^2}{n}-\kappa\bigr)(\mathbf{J}_s-\mathbf{I}_s) - d_L\,\mathbf{I}_s
& \bigl(\frac{f(f+1)}{n}-(\kappa+1)\bigr)\,\mathbf{J}_{s,h}\\[6pt]
\bigl(\frac{f(f+1)}{n}-(\kappa+1)\bigr)\,\mathbf{J}_{h,s}
& \bigl(\frac{(f+1)^2}{n}-(\kappa+1)\bigr)(\mathbf{J}_h-\mathbf{I}_h) - d_H\,\mathbf{I}_h
\end{pmatrix}.
\end{equation}
Set $\mathbf{G} := \mathbf{C} + \mathbf{B}$. We show that $\mathbf{G} = \mathrm{diag}(\mathbf{N},\mathbf{M})$ where $\mathbf{N}$ (on $\{1,\ldots,s\}$) and $\mathbf{M}$ (on $\{s+1,\ldots,p\}$) are adjacency matrices.

For the diagonal: since $(\mathbf{J}_s-\mathbf{I}_s)_{ii}=0$, we have $\mathbf{B}_{ii}=-d_L$ for $i\le s$ and $\mathbf{B}_{ii}=-d_H$ for $i>s$, so $\mathbf{G}_{ii}=\mathbf{C}_{ii}-d_L=0$ for $i\le s$ and $\mathbf{G}_{ii}=\mathbf{C}_{ii}-d_H=0$ for $i>s$, both by~\eqref{eq:C_diag}. For $i\neq j$ within the low block,
\[
\mathbf{G}_{ij} = \mathbf{C}_{ij} + \Bigl(\frac{f^2}{n}-\kappa\Bigr)
= \Bigl(\lambda_{ij}-\frac{f^2}{n}\Bigr) + \Bigl(\frac{f^2}{n}-\kappa\Bigr)
= \lambda_{ij}-\kappa \in \{0,1\}.
\]
Within the high block, similarly $\mathbf{G}_{ij}=\lambda_{ij}-(\kappa+1)\in\{0,1\}$.
Across blocks, $\lambda_{ij}=\kappa+1$ forces
\[
\mathbf{G}_{ij} = \Bigl((\kappa+1)-\frac{f(f+1)}{n}\Bigr)+\Bigl(\frac{f(f+1)}{n}-(\kappa+1)\Bigr) = 0.
\]
Therefore $\mathbf{G}$ is block-diagonal with adjacency-matrix blocks $\mathbf{N}$ and $\mathbf{M}$.

Because $\mathbf{G}_{ij} = 0$ for all cross-block pairs, the graph~$\mathbf{G}$ is block-diagonal. Thus, it decomposes into an induced subgraph $\mathbf{N}$ on vertices
$\{1,\dots,s\}$ and an induced subgraph $\mathbf{M}$ on $\{s+1,\dots,p\}$
with no edges between the two blocks.
Consequently, when computing degrees within $\mathbf{N}$, summing $\mathbf{G}_{ij}$
over all $j\in\{1,\dots,p\}$ gives the same result as summing over
$j\in\{1,\dots,s\}$ alone (the cross-block terms contribute zero). Thus, for $i\le s$,
\[
\deg_\mathbf{N}(i) = \sum_{j=1}^{p}\mathbf{B}_{ij}
= f(k-1)-\kappa(p-1)-h = (p-1-\omega)-(p-s) = s-\omega-1,
\]
where the second equality uses~\eqref{eq:remainder_identity} and $h=p-s$. Thus $\mathbf{N}$ is $(s-\omega-1)$-regular on $s$ vertices.

For $i>s$, a parallel calculation gives
\[
\deg_\mathbf{M}(i) = (f+1)(k-1)-(\kappa+1)(p-1) = k-\omega-1.
\]
Hence $\mathbf{M}$ is $(k-\omega-1)$-regular on $h=p-s$ vertices.

A simple $r$-regular graph on $q$ vertices exists if and only if $0\le r\le q-1$ and $qr$ is even. Applying this to $\mathbf{N}$ and $\mathbf{M}$ yields the range $\omega+1\le s\le p-k+\omega$. The parity constraints $s(s-\omega-1)$ even and $(p-s)(k-\omega-1)$ even are equivalent to $\omega s$ even and $(p-s)(p-k+\omega-s)$ even, respectively, since $s(s-1)$ and $(p-s)(p-s-1)$ are automatically even. This proves Part~(B).

\subsection*{(C) Guaranteed existence for $k=2$ and $k=p-1$}

\noindent\textbf{(C1) $k=2$.}
Each run selects an unordered pair $\{i,j\}$, so the design corresponds to a multigraph on $p$ vertices with $n$ edges: $r_i$ is the vertex degree and $\lambda_{ij}$ is the edge multiplicity. Let
\[
\alpha := \Bigl\lfloor\frac{n}{\binom{p}{2}}\Bigr\rfloor,\qquad
m := n - \alpha\binom{p}{2} \in \Bigl\{0,1,\ldots,\binom{p}{2}-1\Bigr\}.
\]
Take $\alpha$ complete copies of $K_p$ (which gives $\lambda_{ij}=\alpha$
for all pairs and $r_i = \alpha(p-1)$ for all $i$, satisfying NB1 trivially)
and add $m$ additional edges.
 
For the $m$ extra edges we need a simple graph $H_m$ on $p$ vertices
with exactly $m$ edges whose degree sequence is as nearly regular as
possible (so that the combined degrees $r_i = \alpha(p-1)+d_i$ satisfy
$d_i\in\{d,d+1\}$, where $d = \lfloor 2m/p\rfloor$).
Such a graph exists by the following argument.
The target degree sequence $(d+1,\dots,d+1,d,\dots,d)$ with
$2m - pd$ vertices of degree $d+1$ and $p-(2m-pd)$ of degree $d$
has total degree $2m$, which is even.
By the Erd\H{o}s--Gallai theorem, a degree sequence
$d_1\geq\dots\geq d_p \geq 0$ with $\sum d_i$ even is graphical if
and only if for each $k=1,\dots,p$,
\[
  \sum_{i=1}^k d_i \leq k(k-1) + \sum_{i=k+1}^p \min\{d_i,k\}.
\]
For our nearly-regular sequence with $d_i \in\{d,d+1\}$ and
$d \leq p-1$ this condition holds: the left side is at most $k(d+1)$
and the right side is at least $k(k-1)+(p-k)\min\{d,k\}$.
When $k\leq d$ the right side is $k(k-1)+(p-k)k = k(p-1)\geq k(d+1)$
for $d \leq p-2$; when $k > d$ similar arithmetic confirms the bound.
Hence $H_m$ exists for all $m$.
 
Adding $H_m$ to $\alpha$ copies of $K_p$ gives a multigraph with
$\lambda_{ij}\in\{\alpha,\alpha+1\}$ and $r_i\in\{\alpha(p-1)+d,\alpha(p-1)+d+1\}$,
satisfying both NB1 and NB2.  Hence a nearly balanced TCARD always
exists when $k=2$.

\medskip
\noindent\textbf{(C2) $k=p-1$.}
Each run contains all factors except one. Encode row $t$ by its missing index $m(t)\in\{1,\ldots,p\}$: $x_{ti}=0$ if and only if $i=m(t)$, and $x_{ti}=1$ otherwise. Let $d_i:=\#\{t:m(t)=i\}$ so that $\sum_i d_i = n$. Then
\[
r_i = n - d_i, \qquad \lambda_{ij} = n - d_i - d_j \quad (i\neq j).
\]
Choosing $d_i\in\{\lfloor n/p\rfloor,\;\lceil n/p\rceil\}$ ensures that $r_i$ differs by at most one across factors (NB1), and for each fixed $i$ the concurrences $\lambda_{ij}=n-d_i-d_j$ differ by at most one as $j$ varies (NB2). Hence a nearly balanced TCARD exists for all~$n$ when $k=p-1$.
\end{proof}

\section{Proof of Theorem~\ref{thm:B1B2_opt}}\label{app:thm2-proof}

\begin{proof}
\textbf{Part (i).}
Since $(2r_j-n)^2 = 4r_j^2 - 4nr_j + n^2$,
\[
  n^2 B_1(\mathbf{X}) = \sum_{j=1}^p (2r_j-n)^2
             = 4\sum_{j=1}^p r_j^2 - 4n\sum_{j=1}^p r_j + pn^2.
\]
The constraint $\sum_j r_j = nk$ is fixed over the TCARD space,
so minimising $B_1$ is equivalent to minimising $\sum_j r_j^2$
subject to the fixed-sum constraint.
By the standard convexity--majorisation argument,
$\sum_j r_j^2$ is minimised if and only if
$r_j\in\{f,f+1\}$ for all $j$ (where $f=\lfloor nk/p\rfloor$),
which is exactly~NB1.
 
\noindent\textbf{Part (ii).}
Fix $(n,p,k)$ and condition on a replication vector $r=(r_1,\dots,r_p)$
satisfying~NB1.
Write $a_{j\ell} := n - 2(r_j+r_\ell)$, so that
$n^2 B_2(\mathbf{X}) = \sum_{j<\ell}(a_{j\ell}+4\lambda_{j\ell})^2$.
Expanding,
\[
  \sum_{j<\ell}(a_{j\ell}+4\lambda_{j\ell})^2
  = \sum_{j<\ell}a_{j\ell}^2
    + 8\sum_{j<\ell}a_{j\ell}\lambda_{j\ell}
    + 16\sum_{j<\ell}\lambda_{j\ell}^2.
\]
The first term depends only on $r$.  For the cross term, using
symmetry $\lambda_{j\ell}=\lambda_{\ell j}$ and the constraints~(8)--(9),
\[
  \sum_{j<\ell}(r_j+r_\ell)\lambda_{j\ell}
  = \sum_{j=1}^p r_j\sum_{\ell\neq j}\lambda_{j\ell}
  = (k-1)\sum_{j=1}^p r_j^2,
\]
so
\[
  \sum_{j<\ell}a_{j\ell}\lambda_{j\ell}
  = n\cdot n\tbinom{k}{2} - 2(k-1)\sum_j r_j^2,
\]
which is constant when $r$ is fixed.
Therefore, conditional on~NB1, minimising $B_2$ is equivalent to
minimising $\sum_{j<\ell}\lambda_{j\ell}^2$ subject to the constraints~(7)--(9).
 
For each fixed $j$, constraint~(9) gives the fixed row-sum
$\sum_{\ell\neq j}\lambda_{j\ell} = r_j(k-1)$.
By the convexity of $x^2$, the minimum of $\sum_{\ell\neq j}\lambda_{j\ell}^2$
over non-negative integers with this fixed sum is achieved if and only if
the $p-1$ values differ by at most one, i.e.\
$\lambda_{j\ell}\in\{\lfloor r_j(k-1)/(p-1)\rfloor,
                      \lceil r_j(k-1)/(p-1)\rceil\}$
for all $\ell\neq j$, which is exactly~NB2 for row $j$.
Hence, for each fixed $j$,
\begin{equation}\label{eq:rowlb}
  \sum_{\ell\neq j}\lambda_{j\ell}^2
  \;\geq\;
  \sum_{\ell\neq j}\Bigl(\bigl\lfloor r_j(k-1)/(p-1)\bigr\rfloor
                         + \epsilon_{j\ell}\Bigr)^2
  =: m_j,
\end{equation}
where $\epsilon_{j\ell}\in\{0,1\}$ and
$\sum_\ell \epsilon_{j\ell} = r_j(k-1) - (p-1)\lfloor r_j(k-1)/(p-1)\rfloor$.
Summing over $j$ and using $\lambda_{j\ell}=\lambda_{\ell j}$,
\[
  2\sum_{j<\ell}\lambda_{j\ell}^2
  = \sum_{j=1}^p\sum_{\ell\neq j}\lambda_{j\ell}^2
  \;\geq\; \sum_{j=1}^p m_j.
\]
 
It remains to show that the lower bound $\sum_j m_j$ is attained
by some symmetric integer matrix $\Lambda$ satisfying the TCARD
constraints~(7)--(9) and such that $\Lambda_{j\ell}\in\{\alpha_j,\alpha_j+1\}$
for all $\ell\neq j$ (where $\alpha_j = \lfloor r_j(k-1)/(p-1)\rfloor$),
i.e.\ that NB2 can be satisfied simultaneously for all $j$.
 
Suppose $r_j = f$ for $j\leq s$ and $r_j = f+1$ for $j > s$ (NB1).
We want a $p\times p$ symmetric $0/1/2$-valued matrix
with zeros on the diagonal, row sums $r_j(k-1)$ for each $j$,
and entries in $\{\alpha_j,\alpha_j+1\}$ off-diagonal.
This is precisely the existence of a nearly balanced TCARD established
in Theorem~2, which guarantees the existence of
$\lambda_{j\ell}\in\{\kappa,\kappa+1\}$ (Type~I) or
$\lambda_{j\ell}\in\{\alpha_j,\alpha_j+1\}$ (Type~II)
for all pairs $(j,\ell)$ simultaneously.
Importantly, the symmetric matrix $\Lambda$ constructed in
Theorem~2 is consistent (i.e.\ $\Lambda_{j\ell}=\Lambda_{\ell j}$)
by construction.
Hence the per-row NB2 conditions are simultaneously satisfied. Therefore, $\sum_{j<\ell}\lambda_{j\ell}^2$ is minimised
if and only if NB2 holds for every $j$, i.e.\ if and only if
NB2 in Definition~1 holds.

\end{proof}

\section{Proof of Proposition~\ref{prop:force_M}}
\begin{lemma}[Existence of a balancing swap]\label{lem:balance_swap}
If there exist indices \(a,b\) with \(r_a\ge r_b+2\), then there is a row \(s\) such that
\(x_{sa}=1\) and \(x_{sb}=0\). Flipping \((x_{sa},x_{sb})=(1,0)\mapsto(0,1)\) preserves each row’s
weight \(k\) and updates the replications as \(r_a\mapsto r_a-1,\ r_b\mapsto r_b+1\).
\end{lemma}

This is a simple pigeonhole argument: if column \(a\) appears at least two more times than column \(b\), some row must contain \(a\) but not \(b\), enabling a \(1\!\leftrightarrow\!0\) swap that strictly balances the two totals. Such swap strictly reduces the replication sum of squares and, in the worst case, increases the concurrence sum by a controlled amount.

\begin{lemma}[One–step bounds for the centered criterion]\label{lem:one_step_centered}
Under the swap in Lemma~\ref{lem:balance_swap},
\begin{align}
\Delta\!\Big(\sum_{i=1}^{p} (r_i-\bar r)^2\Big)
&=-2(r_a-r_b-1)\ \le\ -2, \label{eq:delta_r_centered}\\[2pt]
\Delta\!\Big(\sum_{1\le i<j\le p} (\lambda_{ij}-\bar\lambda)^2\Big)
&\le 2\,r_b\,(k-1). \label{eq:delta_lam_centered}
\end{align}
\end{lemma}

\begin{proof}
Only columns \(a\) and \(b\) change in the replication term. Writing
\(d_a=r_a-\bar r\) and \(d_b=r_b-\bar r\), after the swap we have \(d_a\mapsto d_a-1\) and
\(d_b\mapsto d_b+1\), hence
\[
(d_a-1)^2+(d_b+1)^2-d_a^2-d_b^2
=-2(d_a-d_b-1)=-2(r_a-r_b-1)\le -2,
\]
which gives \eqref{eq:delta_r_centered}.

For the concurrence term, only the \(2(k-1)\) pairs \((a,j)\) and \((b,j)\) change, where
\(j\) ranges over the other \(k-1\) active factors in row \(s\).
For each such \(j\), \(\lambda_{aj}\mapsto \lambda_{aj}-1\) and \(\lambda_{bj}\mapsto \lambda_{bj}+1\).
Then the net change for the two affected pairs is
\[
(\lambda_{bj}+1-c)^2-(\lambda_{bj}-c)^2
+(\lambda_{aj}-1-c)^2-(\lambda_{aj}-c)^2
=2(\lambda_{bj}-\lambda_{aj})+2.
\]
Because \(a\) and \(j\) co-occur in row \(s\) before the swap, \(\lambda_{aj}\ge 1\).
Also, \(\lambda_{bj}\le r_b\) since
\(\lambda_{bj}=\sum_{t=1}^n x_{tb}x_{tj}\le \sum_{t=1}^n x_{tb}=r_b\).
Therefore
\[
2(\lambda_{bj}-\lambda_{aj})+2 \ \le\ 2(r_b-1)+2 \ =\ 2r_b.
\]
Summing this bound over the \(k-1\) choices of \(j\) yields \eqref{eq:delta_lam_centered}.
\end{proof}

\begin{lemma}[Concurrence dispersion penalizes replication imbalance]\label{lem:lambda_lb}
For any TCARD with replication counts \(\{r_i\}\) and concurrences \(\{\lambda_{ij}\}\),
\begin{equation}\label{eq:lambda_lb}
\sum_{1\le i<j\le p}(\lambda_{ij}-\bar\lambda)^2
\;\ge\;
\frac{(k-1)^2}{2(p-1)}\sum_{i=1}^p(r_i-\bar r)^2.
\end{equation}
\end{lemma}

\begin{proof}
Using \(\sum_{j\ne i}\lambda_{ij}=r_i(k-1)\), we have
\(\sum_{j\ne i}(\lambda_{ij}-\bar\lambda)=(k-1)(r_i-\bar r)\).
Apply Cauchy--Schwarz to each \(i\):
\(\sum_{j\ne i}(\lambda_{ij}-\bar\lambda)^2\ge \frac{1}{p-1}\big(\sum_{j\ne i}(\lambda_{ij}-\bar\lambda)\big)^2\),
then sum over \(i\) and divide by \(2\).
\end{proof}
Lemma~\ref{lem:lambda_lb} shows that the concurrence term alone already imposes a quadratic penalty on
replication imbalance. Consequently, moderate weights (and in particular \(w_1=w_2=1\) in our experiments)
typically suffice to drive the search into the nearly balanced replication regime, even though
Proposition~\ref{prop:force_M} provides a conservative worst-case guarantee.

\subsection*{Proof of Proposition~\ref{prop:force_M}}
\begin{proof}
Assume the weight ratio satisfies
\begin{equation}\label{eq:w_threshold_centered}
\frac{w_1}{w_2}\;>\;\frac{2}{p-1}\Big\lfloor \frac{nk}{p}\Big\rfloor (k-1).
\end{equation}
Suppose, toward a contradiction, that a global minimizer $\mathbf{X}$ of \(\Phi_{\mathrm{BCD}}\) has two columns \(a,b\)
with \(r_a\ge r_b+2\). Choose \(b\) to be a minimally replicated column. Then
\(r_b\le \lfloor nk/p\rfloor\) since \(\sum_i r_i=nk\).
By Lemma~\ref{lem:balance_swap}, there exists a feasible swap. Combining Lemma~\ref{lem:one_step_centered}
with the definition of \(\Phi_{\mathrm{BCD}}\) gives
\[
\Delta\Phi_{\mathrm{BCD}}
=\frac{w_1}{p}\,\Delta\!\Big(\sum_i (r_i-\bar r)^2\Big)
+\frac{w_2}{\binom{p}{2}}\,\Delta\!\Big(\sum_{i<j} (\lambda_{ij}-\bar\lambda)^2\Big)
\ \le\ -\frac{2w_1}{p}+\frac{2w_2}{\binom{p}{2}}\,r_b\,(k-1).
\]
The right-hand side is strictly negative whenever
\(
w_1/w_2 > \frac{2}{p-1}\,r_b\,(k-1),
\)
and hence under \eqref{eq:w_threshold_centered} since \(r_b\le \lfloor nk/p\rfloor\).
Thus the swap strictly decreases \(\Phi_{\mathrm{BCD}}\), contradicting global optimality. Therefore no pair
\(a,b\) can satisfy \(r_a\ge r_b+2\), so all replications differ by at most one; equivalently
\(r_i\in\{\lfloor nk/p\rfloor,\ \lceil nk/p\rceil\}\) for all \(i\).
\end{proof}

\section{Proof of Proposition~\ref{prop:force_S}}
\begin{proof}
From~\eqref{eq:tr_identities},
\[
\operatorname{tr}(\mathbf{C}^2)\;=\;\underbrace{\sum_i\!\Big(r_i-\frac{r_i^2}{n}\Big)^2}_{\text{constant for fixed }r}
\;+\;2\sum_{i<j}\Big(\lambda_{ij}-\frac{r_ir_j}{n}\Big)^2.
\]
(i) Case I: perfect replication balance ($r_i\equiv\bar r$).
Write $M=\binom{p}{2}$ and define
$c_{ij}:=r_i r_j/n$. Under $r_i\equiv\bar r$, we have $c_{ij}\equiv c:=\bar r^2/n$.
Also, because each row has exactly $k$ ones,
\[
\sum_{i<j}\lambda_{ij}=\sum_{t=1}^n \binom{k}{2}=n\binom{k}{2}
\quad\text{is fixed over the design space.}
\]
Now expand:
\[
\sum_{i<j}(\lambda_{ij}-\bar\lambda)^2
=\sum_{i<j}\lambda_{ij}^2-2\bar\lambda\sum_{i<j}\lambda_{ij}+M\bar\lambda^2,
\]
\[
\sum_{i<j}(\lambda_{ij}-c)^2
=\sum_{i<j}\lambda_{ij}^2-2c\sum_{i<j}\lambda_{ij}+M c^2.
\]
Since $\sum_{i<j}\lambda_{ij}$ is fixed, both expressions differ from $\sum_{i<j}\lambda_{ij}^2$ by a design-independent constant. Therefore
\[
\arg\min \sum_{i<j}(\lambda_{ij}-\bar\lambda)^2
=\arg\min \sum_{i<j}\lambda_{ij}^2
=\arg\min \sum_{i<j}(\lambda_{ij}-c)^2.
\]
Finally, $\operatorname{tr}(\mathbf{C}^2)$ equals a constant (first term) plus $2\sum_{i<j}(\lambda_{ij}-c)^2$,
hence minimizing $\sum_{i<j}(\lambda_{ij}-\bar\lambda)^2$ is equivalent to minimizing $\operatorname{tr}(\mathbf{C}^2)$
in the perfectly balanced case.
\hfill$\square$

(ii) Case II: two replication levels and rectangle swaps within a level.
Assume the replication vector is at its discrete optimum so that
\[
r_i\in\{\underline r,\overline r\},\qquad
\underline r=\big\lfloor\tfrac{nk}{p}\big\rfloor,\ \overline r=\underline r+1.
\]
Set $c_{ij}:=r_i r_j/n$. Then $c_{ij}$ can take at most three constants:
$\underline r^2/n$, $(\underline r\,\overline r)/n$, or $\overline r^2/n$, according to pair type.
Consider the rectangle swap on two rows $s,t$ and two columns $i,j$ that belong to the
same replication class (so $r_i=r_j$):
\[
\begin{bmatrix}x_{si}&x_{sj}\\ x_{ti}&x_{tj}\end{bmatrix}
=
\begin{bmatrix}1&0\\ 0&1\end{bmatrix}
\ \longleftrightarrow\
\begin{bmatrix}0&1\\ 1&0\end{bmatrix}.
\]
This move preserves the row cardinality constraint $k$ and the column totals of $i$ and $j$, hence preserves the
entire replication vector $r$. It changes only concurrences that involve $i$ or $j$, and only through
rows $s$ or $t$. Fix any other column $\ell\neq i,j$. Then
\[
\Delta\lambda_{i\ell}=-\Delta\lambda_{j\ell},\qquad |\Delta\lambda_{i\ell}|,|\Delta\lambda_{j\ell}|\le 1,
\]
while $\Delta\lambda_{ij}=0$.

Now examine the mixed term: $\sum_{u<v} c_{uv}\lambda_{uv}=\sum_{u<v}\frac{r_u r_v}{n}\,\lambda_{uv}$.
Under the rectangle swap, only pairs with $u=i$ or $u=j$ can change, and for each $\ell\neq i,j$,
\[
\Delta\big(c_{i\ell}\lambda_{i\ell}+c_{j\ell}\lambda_{j\ell}\big)
=\frac{r_i r_\ell}{n}\,\Delta\lambda_{i\ell}+\frac{r_j r_\ell}{n}\,\Delta\lambda_{j\ell}
=\frac{r_\ell}{n}\,r_i\big(\Delta\lambda_{i\ell}+\Delta\lambda_{j\ell}\big)=0,
\]
because $r_i=r_j$ and $\Delta\lambda_{i\ell}=-\Delta\lambda_{j\ell}$.
Summing over all $\ell$ shows
\[
\Delta\sum_{u<v} c_{uv}\lambda_{uv}=0,
\]
i.e., the mixed term is invariant under class-preserving rectangle swaps. Since $r$ is fixed,
$\sum_{u<v}c_{uv}^2$ is also constant. Therefore, on the swap-connected neighborhood generated by such moves,
\[
\sum_{i<j}(\lambda_{ij}-c_{ij})^2
=\underbrace{\sum_{i<j}\lambda_{ij}^2}_{\text{variable}}
-2\underbrace{\sum_{i<j}c_{ij}\lambda_{ij}}_{\text{invariant}}
+\underbrace{\sum_{i<j}c_{ij}^2}_{\text{constant}}
\]
differs from $\sum_{i<j}\lambda_{ij}^2$ by an additive constant, so they have the same minimizers.
Finally, with $r$ fixed, the minimizers of $\sum_{i<j}(\lambda_{ij}-c_{ij})^2$, of $\sum_{i<j}\lambda_{ij}^2$, and of
$\operatorname{tr}(\mathbf{C}^2)$ coincide on this class-preserving rectangle-swap neighborhood.
\hfill$\square$

\noindent\emph{Remarks.}
(i) Part (ii) is a structural statement for the fixed-\(r\) subproblem and should be interpreted on the move graph induced by class-preserving rectangle swaps.

(ii) This is not a literal description of the CE algorithm in Section~\ref{subsec:CE}, which uses within-row \(1\!\leftrightarrow\!0\) swaps and may change \(r\) during the early stage of the search. Rather, part (ii) explains why, once the replication vector has essentially stabilized at its discrete optimum, further improvement is naturally tied to regularizing the concurrence profile.

(iii) When \(p\mid nk\), we have perfect replication balance and \(c_{ij}\equiv c\) for all pairs, so the equivalence becomes global and no neighborhood restriction is needed.

\end{proof}

\section{Proof of Theorem~\ref{thm:UE_equiv}}
\label{app:UE_proof}
\begin{proof}
By $Z_{ti}=2x_{ti}-1$, we have for each $i$,
\begin{equation}\label{eq:intercept-factor}
S_{1,i+1}=\sum_{t=1}^n Z_{ti}=2r_i-n.
\end{equation}
For $i\ne j$,
\begin{equation}\label{eq:factor-factor}
S_{i+1,j+1}
=\sum_{t=1}^n Z_{ti}Z_{tj}
=4\lambda_{ij}-2r_i-2r_j+n.
\end{equation}
From \eqref{eq:intercept-factor},
\begin{equation}\label{eq:mean-IF}
\overline{S}_{1,\cdot}
=\frac{1}{p}\sum_{i=1}^p(2r_i-n)
=2\bar r-n.
\end{equation}
Using \eqref{eq:factor-factor},
\begin{align}
\overline{S}_{\cdot,\cdot}
&=\frac{1}{\binom{p}{2}}
\sum_{1\le i<j\le p}\bigl(4\lambda_{ij}-2r_i-2r_j+n\bigr)\nonumber\\
&=
4\bar\lambda
-\frac{2}{\binom{p}{2}}\sum_{1\le i<j\le p}(r_i+r_j)
+n\nonumber\\
&=4\bar\lambda-4\bar r+n.
\label{eq:mean-FF}
\end{align}
Combining \eqref{eq:intercept-factor} and \eqref{eq:mean-IF}, we have
\begin{equation}\label{eq:IF-SS}
\sum_{i=1}^p\bigl(S_{1,i+1}-\overline{S}_{1,\cdot}\bigr)^2
=4\sum_{i=1}^p(r_i-\bar r)^2.
\end{equation}
From \eqref{eq:factor-factor} and \eqref{eq:mean-FF}, we have
\begin{align}
S_{i+1,j+1}-\overline{S}_{\cdot,\cdot}
&=
\bigl(4\lambda_{ij}-2r_i-2r_j+n\bigr)-\bigl(4\bar\lambda-4\bar r+n\bigr)\nonumber\\
&=
4(\lambda_{ij}-\bar\lambda)-2\bigl[(r_i-\bar r)+(r_j-\bar r)\bigr].
\label{eq:FF-centered}
\end{align}
Let $\Delta_i=r_i-\bar r$ and $\Delta_{ij}=\lambda_{ij}-\bar\lambda$. Then
\begin{align}
\sum_{i<j}\bigl(S_{i+1,j+1}-\overline{S}_{\cdot,\cdot}\bigr)^2
&=\sum_{i<j}\Bigl[16\Delta_{ij}^2
-16\Delta_{ij}(\Delta_i+\Delta_j)
+4(\Delta_i+\Delta_j)^2\Bigr]\nonumber\\
&=16\sum_{i<j}\Delta_{ij}^2
-16\sum_{i<j}\Delta_{ij}(\Delta_i+\Delta_j)
+4\sum_{i<j}(\Delta_i+\Delta_j)^2.
\label{eq:FF-expand}
\end{align}
Since for each $i$, $\sum_{j\ne i}\lambda_{ij}=(k-1)r_i$, we have $\sum_{j\ne i}\Delta_{ij}=(k-1)r_i-(p-1)\bar\lambda$, and $\sum_{i<j}\Delta_{ij}(\Delta_i+\Delta_j) =(k-1)\sum_{i=1}^p \Delta_i r_i$.
Since $\sum_i\Delta_i=0$ and $r_i=\Delta_i+\bar r$, we get
$\sum_{i=1}^p \Delta_i r_i=\sum_{i=1}^p \Delta_i^2$, so
\begin{equation}\label{eq:cross-simplified}
\sum_{i<j}\Delta_{ij}(\Delta_i+\Delta_j)
=(k-1)\sum_{i=1}^p\Delta_i^2.
\end{equation}
Expanding and using $\sum_i\Delta_i=0$,
\begin{equation}\label{eq:pair-Delta}
\sum_{i<j}(\Delta_i+\Delta_j)^2
=(p-2)\sum_{i=1}^p\Delta_i^2.
\end{equation}
Substituting \eqref{eq:cross-simplified} and \eqref{eq:pair-Delta} into \eqref{eq:FF-expand} yields
\begin{align}
\sum_{i<j}\bigl(S_{i+1,j+1}-\overline{S}_{\cdot,\cdot}\bigr)^2
&=16\sum_{i<j}\Delta_{ij}^2
-16(k-1)\sum_{i=1}^p\Delta_i^2
+4(p-2)\sum_{i=1}^p\Delta_i^2\nonumber\\
&=16\sum_{i<j}(\lambda_{ij}-\bar\lambda)^2
+4(p-4k+2)\sum_{i=1}^p(r_i-\bar r)^2.
\label{eq:FF-final}
\end{align}
Adding the intercept--factor part \eqref{eq:IF-SS} gives
\[
UE(s^2)(\mathbf X)
=
16\sum_{i<j}(\lambda_{ij}-\bar\lambda)^2
+
4(p-4k+3)\sum_{i=1}^p(r_i-\bar r)^2,
\]
which proves \eqref{eq:UE_identity}. Under the definition \eqref{eq:Phi_BCD}, choosing $w_1=(p-4k+3)/4p$, $w_2=(p-1)/2p$ yields $UE(s^2)(\mathbf X)=\Phi_{\mathrm{BCD}}(\mathbf X)$. Therefore, when \(p>4k-3\), the two criteria have the same minimizer.
\end{proof}

\section{Proof of Theorem~\ref{thm:BayesD_main}}
\subsection{Proof of Lemma~\ref{lem:E-Frobenius_BayesD}}
\begin{proof}
The matrix $\boldsymbol{\Lambda}-\boldsymbol{\Lambda}_0$ has diagonal entries $(r_i-\bar r)$ and off-diagonal entries
$(\lambda_{ij}-\bar\lambda)$, hence
\[
\|\boldsymbol{\Lambda}-\boldsymbol{\Lambda}_0\|_F^2
=\sum_i(r_i-\bar r)^2 + 2\sum_{i<j}(\lambda_{ij}-\bar\lambda)^2
= \mathcal{V}_1(\mathbf{X})+2\mathcal{V}_2(\mathbf{X}).
\]
Also $\mathbf{rr}^\top-\mathbf{r}_0\mathbf{r}_0^\top=\mathbf{r}_0\mathbf{b}^\top+\mathbf{b}\mathbf{r}_0^\top+\mathbf{bb}^\top$ and $\|\mathbf{uv}^\top\|_F=\|\mathbf{u}\|_2\|\mathbf{v}\|_2$, yielding
\[
\|\mathbf{rr}^\top-\mathbf{r}_0\mathbf{r}_0^\top\|_F
\le 2\|\mathbf{r}_0\|_2\|\mathbf{b}\|_2+\|\mathbf{b}\|_2^2
=2\sqrt p\,\bar r\,\sqrt{\mathcal{V}_1(\mathbf{X})}+\mathcal{V}_1(\mathbf{X}).
\]
Since $\mathbf{E}=(\boldsymbol{\Lambda}-\boldsymbol{\Lambda}_0)-\frac{1}{n}(\mathbf{rr}^\top-\mathbf{r}_0\mathbf{r}_0^\top)$, the upper bound
$U(\mathcal{V}_1,\mathcal{V}_2)$ follows from the triangle inequality and the lower bound
$L(\mathcal{V}_1,\mathcal{V}_2)$ follows from the reverse triangle inequality.
\end{proof}

\subsection{Proof of Theorem~\ref{thm:gap-upper_BayesD}}
\begin{proof}
Write $\mathbf{C}=\mathbf{C}_0+\mathbf{E}$ and set $\boldsymbol{\Delta}:=\mathbf{A}^{-1/2}\mathbf{E}\mathbf{A}^{-1/2}$. Then
\[
G_\alpha(\mathbf{X})=\log\det(\mathbf{A})-\log\det(\mathbf{A+E})=-\log\det(\mathbf{I}+\boldsymbol{\Delta}).
\]
Condition \eqref{eq:neighborhood_BayesD} implies
$\|\boldsymbol{\Delta}\|_2\le \|\mathbf{A}^{-1/2}\|_2^2\|\mathbf{E}\|_2 \le \|\mathbf{E}\|_F/\alpha \le \rho$.
For $\|\boldsymbol{\Delta}\|_2\le\rho<1$, the standard perturbation inequality gives
\[
-\log\det(\mathbf{I}+\boldsymbol{\Delta})\le |\mathrm{tr}(\boldsymbol{\Delta})| + \frac{1}{2(1-\rho)}\|\boldsymbol{\Delta}\|_F^2.
\]
Next, $|\mathrm{tr}(\boldsymbol{\Delta})|=|\mathrm{tr}(\mathbf{A}^{-1}\mathbf{E})|\le \|\mathbf{A}^{-1}\|_F\|\mathbf{E}\|_F$ and
$\|\boldsymbol{\Delta}\|_F\le \|\mathbf{A}^{-1/2}\|_2^2\|\mathbf{E}\|_F=\|\mathbf{E}\|_F/\alpha$.
Finally, apply Lemma~\ref{lem:E-Frobenius_BayesD} to bound $\|\mathbf{E}\|_F$ by $U(\mathcal{V}_1,\mathcal{V}_2)$.
\end{proof}

\subsection{Proof of Theroem~\ref{thm:BayesD_main}}
\begin{proof}
Fix $\mathbf{X}\in\boldsymbol{\Delta}_{n,p,k}$ and write the eigenvalues of $\mathbf{C}$ as
$0=\lambda_1(\mathbf{C})<\lambda_2(\mathbf{C})\le\cdots\le \lambda_p(\mathbf{C})$.
Define the normalized determinant polynomial
\[
F_\mathbf{X}(\alpha):=\alpha^{-1}\det(\mathbf{C}+\alpha \mathbf{I}_p)=\prod_{i=2}^p (\lambda_i(\mathbf{C})+\alpha),\qquad \alpha\ge 0.
\]
For any fixed $\alpha>0$, maximizing $f_\alpha(\mathbf{X})$ over $\mathbf{X}\in\boldsymbol{\Delta}_{n,p,k}$ is equivalent to maximizing
$F_\mathbf{X}(\alpha)$, because $\log(\alpha)$ is common to all competitors.

\paragraph{Proof of (i).}
Let $\mathcal{D}^*:=\arg\max_{\mathbf{X}\in\boldsymbol{\Delta}_{n,p,k}} F_\mathbf{X}(0)$, i.e., the centered $D$-optimal set since
$F_\mathbf{X}(0)=\operatorname{pdet}(\mathbf{C})$.
For each $\mathbf{Y}\in\boldsymbol{\Delta}_{n,p,k}\setminus\mathcal{D}^*$, define
\[
g_\mathbf{Y}(\alpha):=\max_{\mathbf{X}\in\mathcal{D}^*}F_\mathbf{X}(\alpha)-F_\mathbf{Y}(\alpha).
\]
Because $\boldsymbol{\Delta}_{n,p,k}$ is finite, the function $\max_{\mathbf{X}\in\mathcal{D}^*}F_\mathbf{X}(\alpha)$ is continuous, hence $g_Y(\alpha)$ is continuous.
Moreover,
\[
g_\mathbf{Y}(0)=\max_{\mathbf{X}\in\mathcal{D}^*}F_\mathbf{X}(0)-F_\mathbf{Y}(0) \;>\;0.
\]
Thus, for each such $\mathbf{Y}$, there exists $\alpha_\mathbf{Y}^\ast>0$ such that $g_\mathbf{Y}(\alpha)>0$ for all $\alpha\in(0,\alpha_\mathbf{Y}^\ast)$.
Let
\[
\alpha^\ast:=\min_{\mathbf{Y}\in\boldsymbol{\Delta}_{n,p,k}\setminus\mathcal{D}^*}\alpha_\mathbf{Y}^\ast \;>\;0.
\]
Then for every $\alpha\in(0,\alpha^\ast)$ and every $Y\notin\mathcal{D}^*$, we have
$F_\mathbf{Y}(\alpha)<\max_{\mathbf{X}\in\mathcal{D}^*}F_\mathbf{X}(\alpha)$, so no $Y\notin\mathcal{D}^*$ can maximize $F_X(\alpha)$.
Therefore $\arg\max f_\alpha \subseteq \mathcal{D}^*$ for all $\alpha\in(0,\alpha^\ast)$.

\paragraph{Proof of (ii).}
Since $F_\mathbf{X}(\alpha)$ is a polynomial of degree $p-1$,
\[
F_\mathbf{X}(\alpha)=\alpha^{p-1}+c_{\mathbf{X},p-2}\alpha^{p-2}+c_{\mathbf{X},p-3}\alpha^{p-3}+\cdots+c_{\mathbf{X},0},
\]
where the first two coefficients satisfy
\[
c_{\mathbf{X},p-2}=\sum_{i=2}^p\lambda_i(\mathbf{C})=\mathrm{tr}(\mathbf{C}),
\qquad
c_{\mathbf{X},p-3}=\sum_{2\le i<j\le p}\lambda_i(\mathbf{C})\lambda_j(\mathbf{C})
=\tfrac12\{\mathrm{tr}(\mathbf{C})^2-\mathrm{tr}(\mathbf{C}^2)\}.
\]
Under the TCARD row-sum constraint $\sum_{i=1}^p r_i=nk$, we have the identity
\[
\mathrm{tr}(\mathbf{C})=nk-\frac{1}{n}\sum_{i=1}^p r_i^2
= nk-\frac{1}{n}\{p\bar r^2+\mathcal V_1(\mathbf{X})\}.
\]
Hence minimizing $\mathcal V_1(\mathbf{X})$ is equivalent to maximizing $c_{\mathbf{X},p-2}=\mathrm{tr}(\mathbf{C})$.

Now assume $p\mid nk$ so that column-balanced designs exist and the minimum of $\mathcal V_1$ over $\boldsymbol{\Delta}_{n,p,k}$ equals $0$.
Let $\mathbf{X}_\phi$ be a global minimizer of $\Phi_{\mathrm{BCD}}$.
Then necessarily $\mathcal V_1(\mathbf{X}_\phi)=0$, and within the column-balanced subclass it minimizes $\mathcal V_2$.

When $\mathcal V_1(\mathbf{X})=0$ (i.e., $r_i\equiv\bar r$), the diagonal entries of $\mathbf{C}$ are constant across $i$, and
$\mathbf{C}_{ij}=\lambda_{ij}-\bar r^2/n$ for $i\neq j$. Since $\sum_{i<j}(\lambda_{ij}-\bar\lambda)=0$,
minimizing $\mathcal V_2(\mathbf{X})=\sum_{i<j}(\lambda_{ij}-\bar\lambda)^2$ over the column-balanced subclass
is equivalent to minimizing $\sum_{i<j} \mathbf{C}_{ij}^2$, hence to minimizing $\mathrm{tr}(\mathbf{C}^2)$
(the diagonal contribution to $\mathrm{tr}(\mathbf{C}^2)$ is constant under $\mathcal V_1=0$).
Therefore, among all designs with $\mathcal V_1=0$, $\mathbf{X}_\phi$ maximizes $c_{\mathbf{X},p-3}$.

Fix any $\mathbf{X}\in\boldsymbol{\Delta}_{n,p,k}$ with $\mathbf{X}\neq \mathbf{X}_\phi$ and define the polynomial difference
$h_\mathbf{X}(\alpha):=F_{\mathbf{X}_\phi}(\alpha)-F_\mathbf{X}(\alpha)$.
If $\mathcal V_1(\mathbf{X})>0$, then $c_{\mathbf{X}_\phi,p-2}>c_{\mathbf{X},p-2}$, so the leading nonzero coefficient of $h_\mathbf{X}(\alpha)$
(at order $\alpha^{p-2}$) is positive. If $\mathcal V_1(\mathbf{X})=0$ but $\mathcal V_2(\mathbf{X})>\mathcal V_2(\mathbf{X}_\phi)$,
then $c_{\mathbf{X}_\phi,p-2}=c_{\mathbf{X},p-2}$ and $c_{\mathbf{X}_\phi,p-3}>c_{\mathbf{X},p-3}$, so the leading nonzero coefficient of $h_\mathbf{X}(\alpha)$
(at order $\alpha^{p-3}$) is positive. In either case, there exists $\alpha_\mathbf{X}^{\ast\ast}>0$ such that
$h_\mathbf{X}(\alpha)>0$ for all $\alpha>\alpha_\mathbf{X}^{\ast\ast}$.
Because $\boldsymbol{\Delta}_{n,p,k}$ is finite, we may take
\[
\alpha^{\ast\ast}:=\max_{\mathbf{X}\in\boldsymbol{\Delta}_{n,p,k}\setminus\{\mathbf{X}_\phi\}} \alpha_\mathbf{X}^{\ast\ast} < \infty,
\]
so that for all $\alpha>\alpha^{\ast\ast}$ we have $F_{\mathbf{X}_\phi}(\alpha)>F_\mathbf{X}(\alpha)$ for every $\mathbf{X}\neq \mathbf{X}_\phi$.
Thus $\mathbf{X}_\phi$ maximizes $F_\mathbf{X}(\alpha)$, equivalently maximizes $f_\alpha(\mathbf{X})$, for all sufficiently large $\alpha$.
\end{proof}

\section{Additional numerical results}\label{app:add_num_res}

\begin{sidewaysfigure}[htp]
\centering
\includegraphics[width=\linewidth]{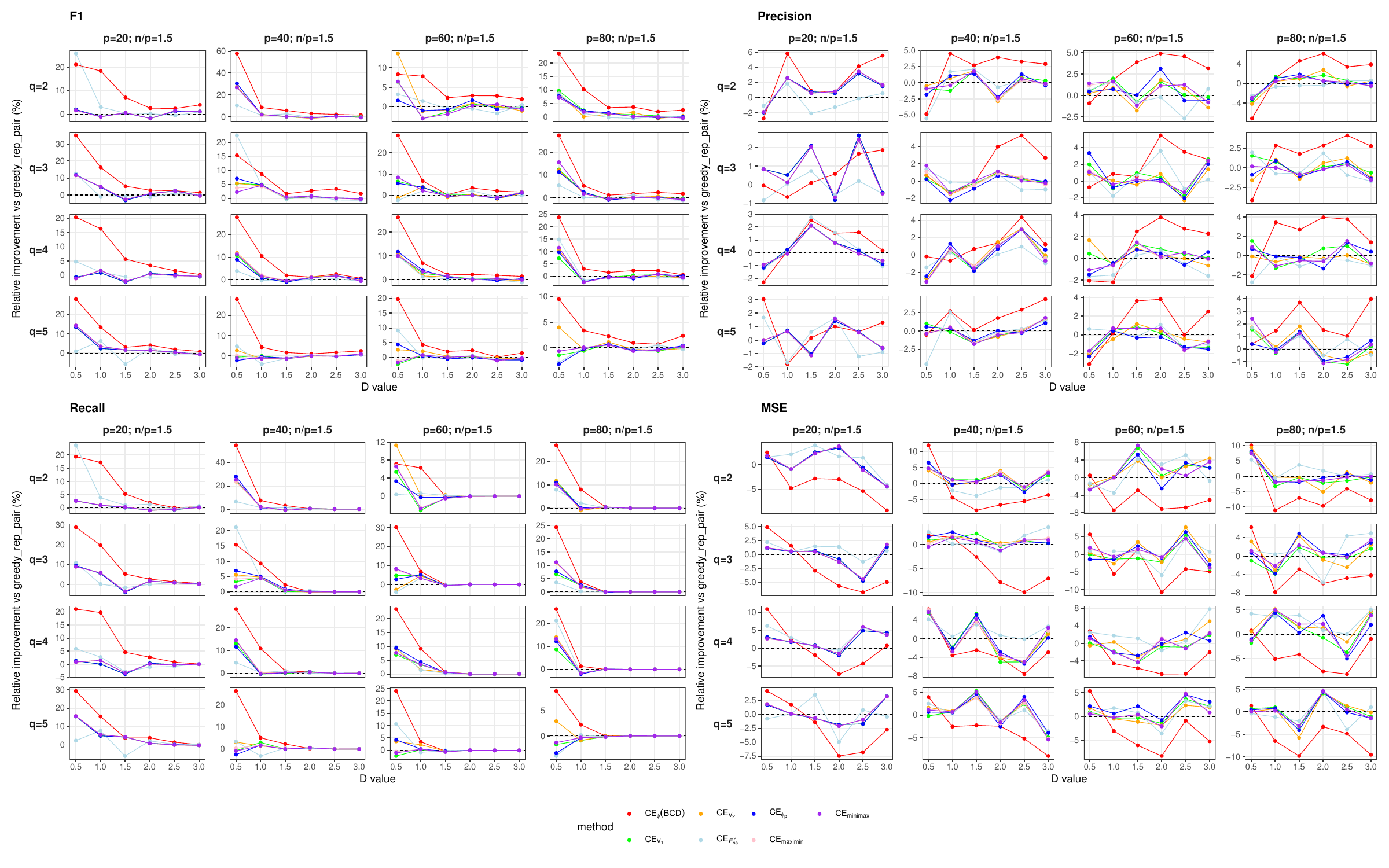}
\caption{Relative improvement over the greedy-rep-pair baseline in F1, precision, recall, and MSE across experimental settings ($k/p=0.1, n/p=1.5$)}
\label{fig:w1_lineplot_k0.1_n1.5}
\end{sidewaysfigure}

\begin{sidewaysfigure}[htp]
\centering
\includegraphics[width=\linewidth]{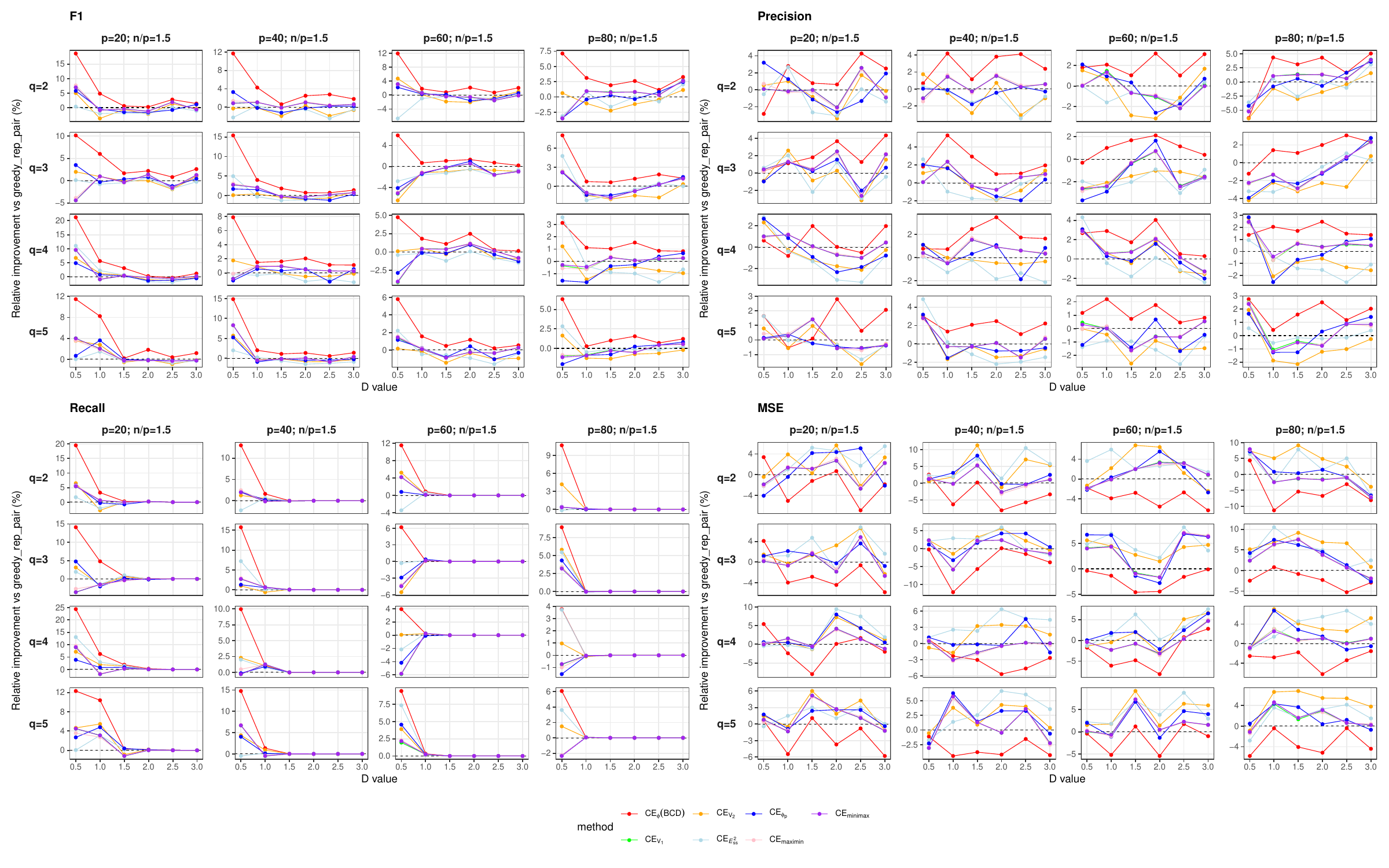}
\caption{Relative improvement over the greedy-rep-pair baseline in F1, precision, recall, and MSE across experimental settings ($k/p=0.5, n/p=1.5$)}
\label{fig:w1_lineplot_k0.5_n1.5}
\end{sidewaysfigure}

\begin{sidewaysfigure}[htp]
\centering
\includegraphics[width=\linewidth]{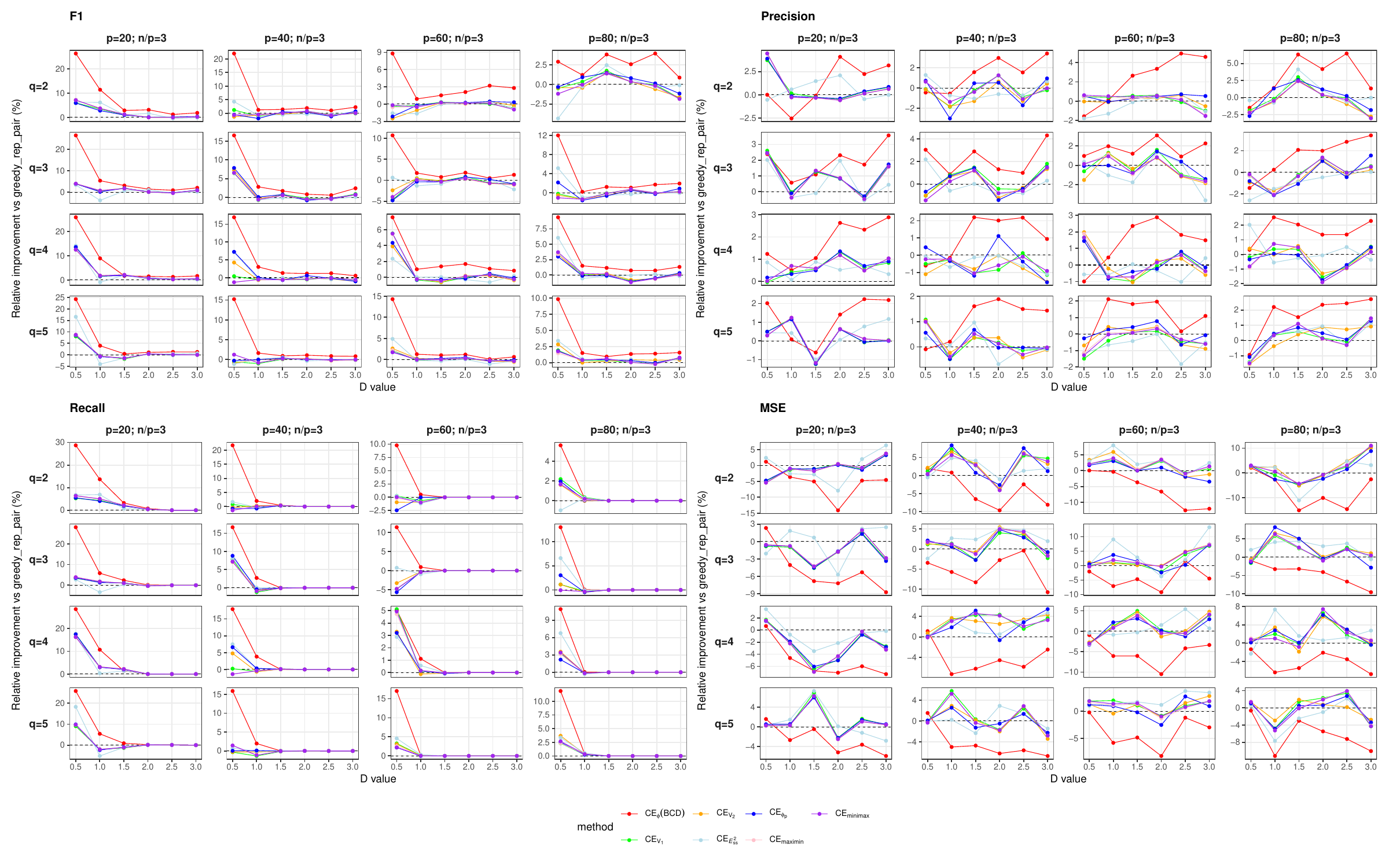}
\caption{Relative improvement over the greedy-rep-pair baseline in F1, precision, recall, and MSE across experimental settings ($k/p=0.1, n/p=3$)}
\label{fig:w1_lineplot_k0.1_n3}
\end{sidewaysfigure}

\begin{sidewaysfigure}[htp]
\centering
\includegraphics[width=\linewidth]{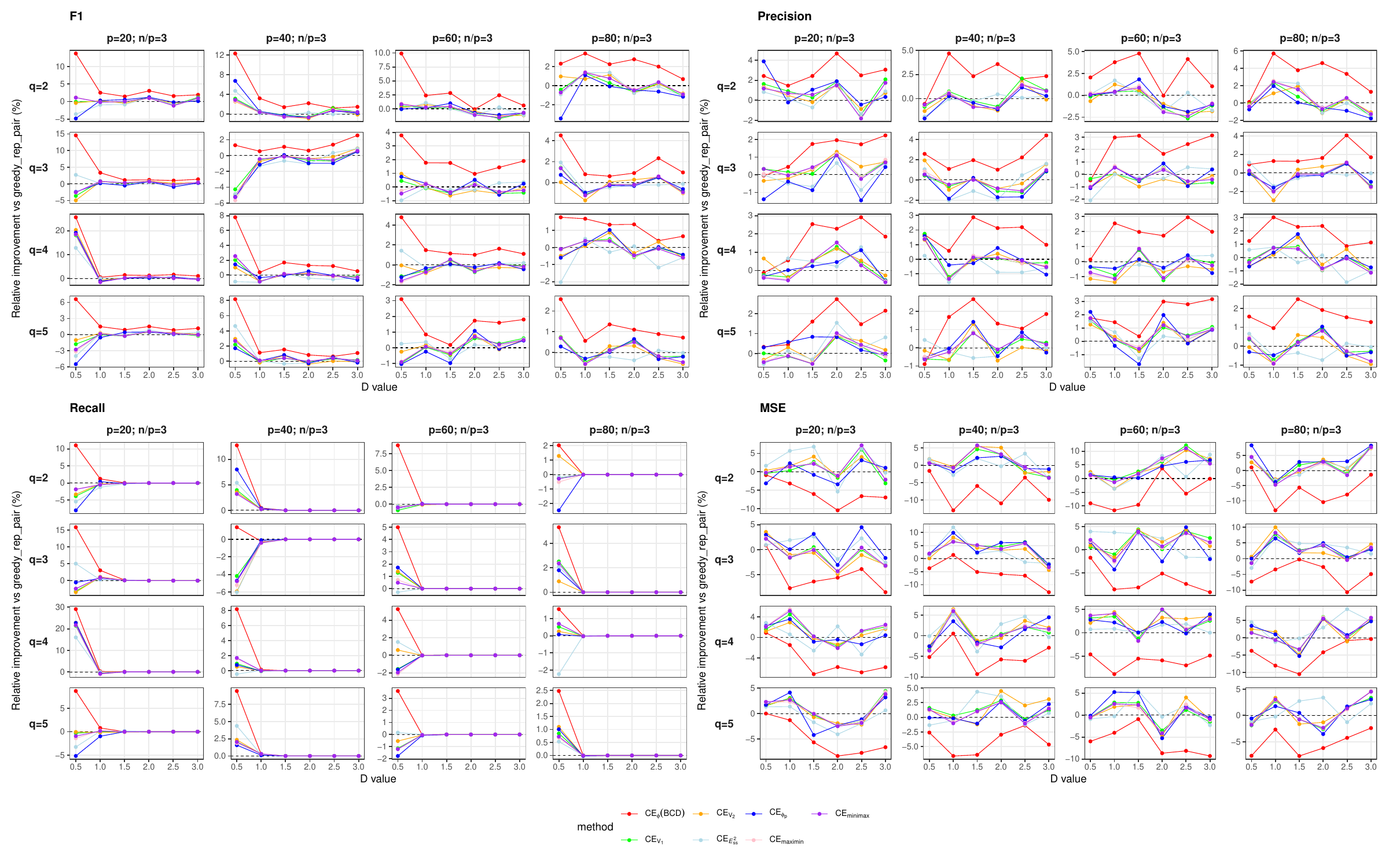}
\caption{Relative improvement over the greedy-rep-pair baseline in F1, precision, recall, and MSE across experimental settings ($k/p=0.25, n/p=3$)}
\label{fig:w1_lineplot_k0.25_n3}
\end{sidewaysfigure}

\begin{sidewaysfigure}[htp]
\centering
\includegraphics[width=\linewidth]{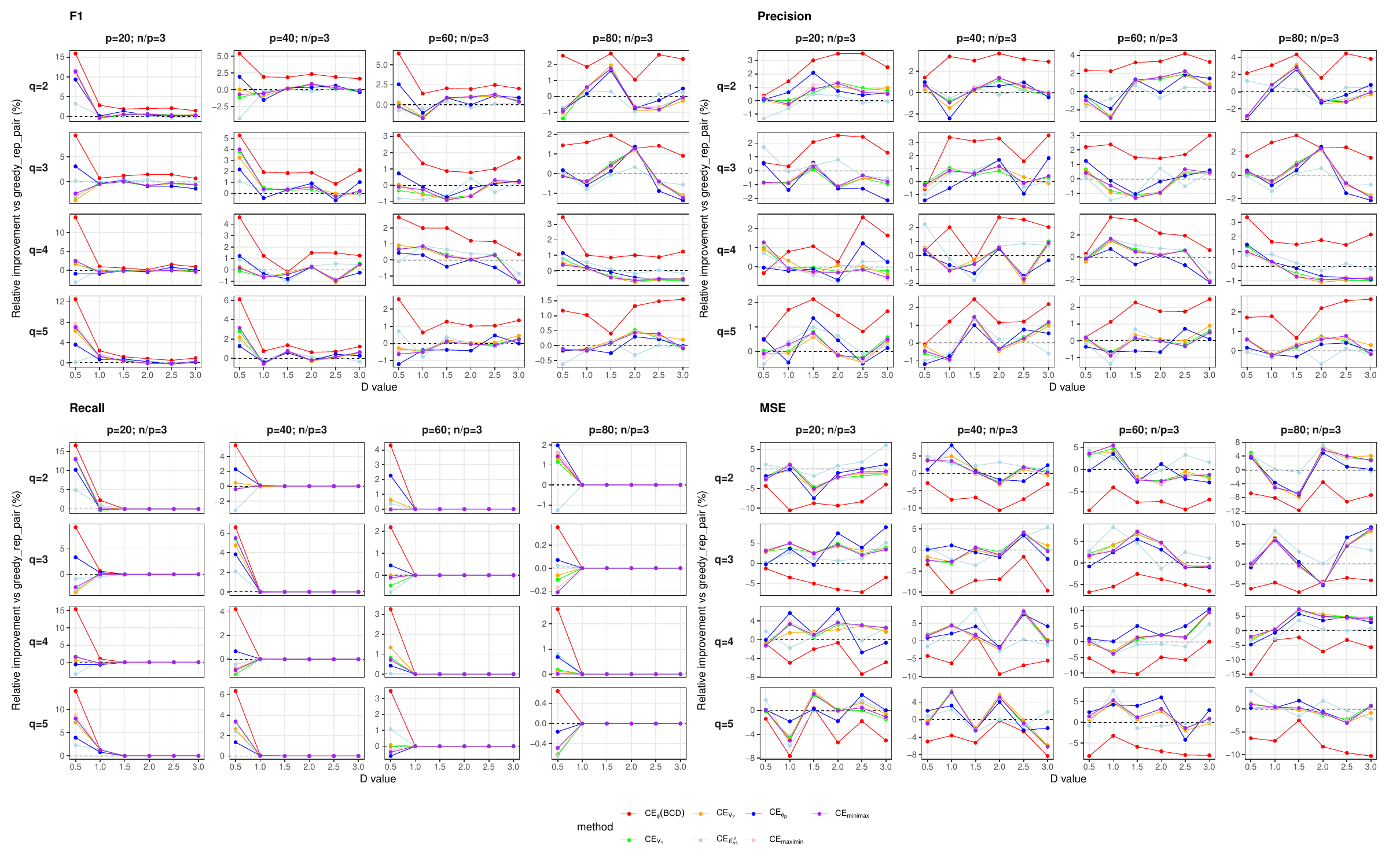}
\caption{Relative improvement over the greedy-rep-pair baseline in F1, precision, recall, and MSE across experimental settings ($k/p=0.5, n/p=3$)}
\label{fig:w1_lineplot_k0.5_n3}
\end{sidewaysfigure}

\subsection{Surrogate tuning at $q^\star = k$}\label{app:qeqk_tuning}

To assess the practical value of the surrogate choice $q^\star = k$, we repeat the simulation-based tuning procedure (Algorithm~2) with the screening plan restricted to
$\mathcal{Q} = \{k\}$, so that $w_1$ is calibrated using only
the projection order $q = k$.  The resulting design is then
evaluated at all projection orders $q \in \{2, 3, 4, 5\}$.
We focus on the strong-constraint regime $k/p = 0.1$ with
$n/p = 1.5$, since this is the only regime in which the surrogate $q^\star = k$ falls within the evaluation range: $k = 2, 3, 4, 5$ for $p = 20, 30, 40, 50$, respectively.

Figure~\ref{fig:qeqk_tuning} reports the relative improvement of each method over \textit{greedy-rep-pair}, decomposed into F1, precision, recall, and MSE.  A clear diagonal pattern emerges: CE-$\Phi_{\mathrm{BCD}}$ shows substantial F1 and recall gains at the cell where $q = k$ but does not improve over the baseline at the remaining projection orders within each~$p$. This indicates that the $w_1$ calibration is sharply projection-order-specific: the weight that optimizes screening performance at one value of~$q$ does not transfer to other values. At the mismatched cells, the tuned design neither
helps nor hurts, behaving similarly to the untuned default. The surrogate $q^\star = k$ is therefore best understood as a low-risk option rather than a broadly effective heuristic:
genuinely beneficial when the active set size is close to~$k$,
and unlikely to cause harm otherwise.

\begin{sidewaysfigure}[htp]
\centering
\includegraphics[width=\linewidth]{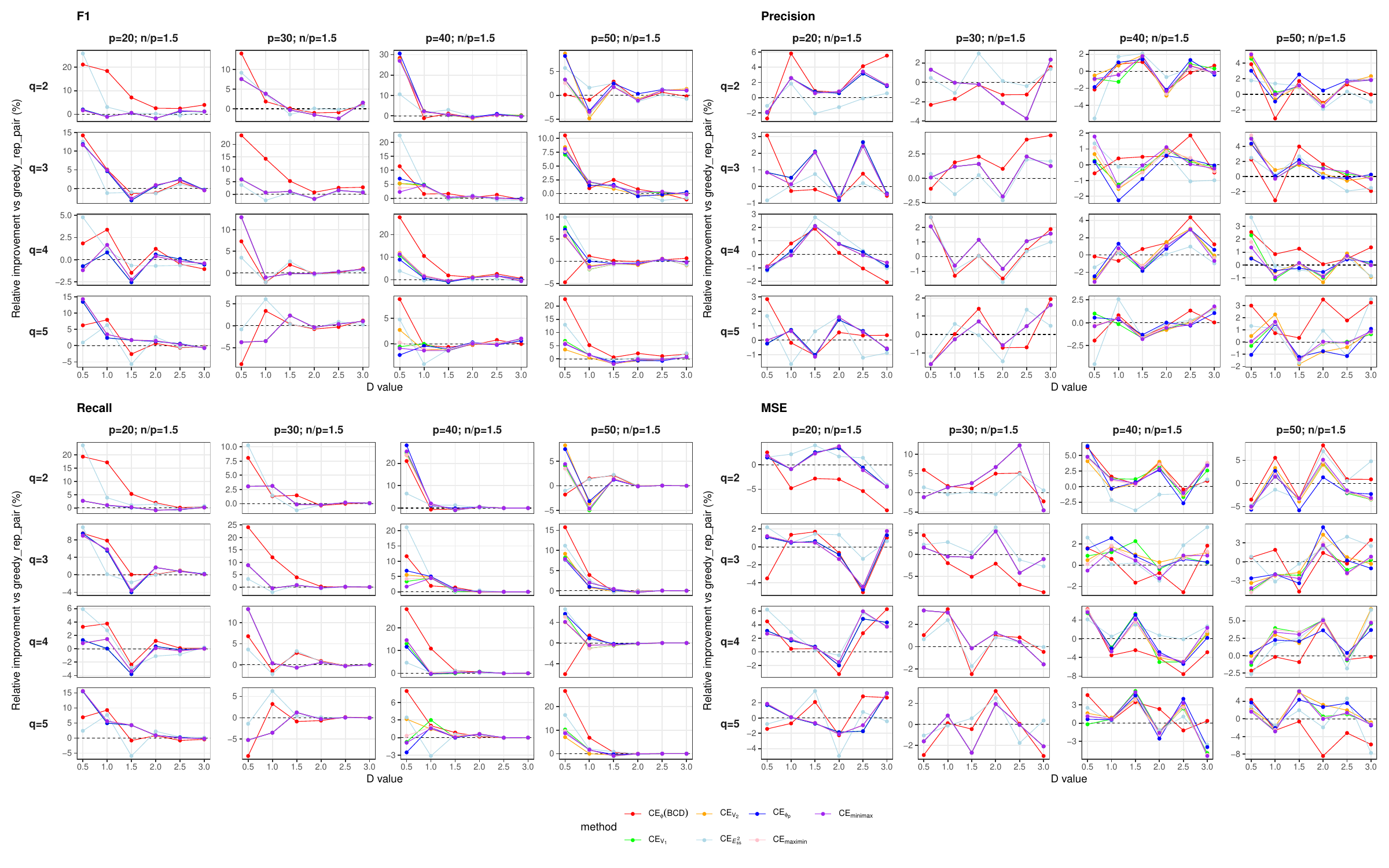}
\caption{Relative improvement over the greedy-rep-pair baseline in F1, precision, recall, and MSE across experimental settings ($k/p=0.1, n/p=1.5$)}
\label{fig:qeqk_tuning}
\end{sidewaysfigure}

\newpage 

\newpage
\end{document}